\documentclass[11pt]{amsart}
\usepackage{macros}
\usepackage{graphicx}

\newcommand{\cyc}{\op{Cyc}}
\newcommand{\T}{\mathscr{T}}
\newcommand{\del}{\partial}
\newcommand{\pa}{\partial}
\newcommand{\gl}{\mf{gl}}

\newcommand{\bbracket}[1]{\left[#1\right]}
\newcommand{\fbracket}[1]{\left\{#1\right\}}
\newcommand{\bracket}[1]{\left(#1\right)}
\newcommand{\CS}{I^{CS}}
\newcommand{\BV}{Batalin-Vilkovisky }

\theoremstyle{definition}

\newtheorem{eg}{Example}[section]

\theoremstyle{remark}
\newtheorem{rmk}{Remark}[section]

\allowdisplaybreaks[4]

\theoremstyle{thm}
\newtheorem*{theoremA}{Theorem A}
\newtheorem*{theoremB}{Theorem B}

\author{Kevin Costello}
 \address{Perimeter Institute for theoretical physics, Waterloo, Ontario Canada N2L 2Y5}
       \email{kcostello@perimeterinstitute.ca}

\author{Si Li}
 \address{Yau Mathematical Science Center, Tsinghua University, Beijing, 100084, China}
       \email{sli@math.tsinghua.edu.cn}

\title{Quantization of open-closed BCOV theory-I}

\begin{document}
\maketitle

\begin{abstract}

This is the first in a series of papers which analyze the problem of quantizing the theory coupling Kodaira-Spencer gravity (or BCOV theory) and holomorphic Chern-Simons on Calabi-Yau manifolds using the formalism for perturbative QFT developed by the first author.  In this paper, we focus on flat space $\C^d$ for $d$ odd.  We prove that there exists a unique quantization of the theory coupling BCOV theory and holomorphic Chern-Simons theory with gauge group the supergroup $GL(N \mid N)$. We deduce a canonically defined quantization of BCOV theory on its own. 

We also discuss some conjectural links between BCOV theory in various dimensions and twists of physical theories: in complex dimension $3$ we conjecture a relationship to twists of $(1,0)$ supersymmetric theories and in complex dimension $5$ to a twist of type IIB supergravity. 

\end{abstract}

\tableofcontents

\section{Introduction}
The Kodaira-Spencer theory of gravity was introduced by Bershadsky, Cecotti, Ooguri and Vafa \cite{BerCecOog94} as the closed string field theory corresponding to the $B$-twisted topological string theory.  In the original formulation, BCOV theory (as we call it)  is a field theory defined on a Calabi-Yau manifold of dimension $3$.  In \cite{CosLi11}, we formulated BCOV theory on a Calabi-Yau manifold of any dimension.

In \cite{BerCecOog94}, the problem of constructing a perturbative quantization of BCOV theory is left open.  After all, BCOV theory is a $6$-dimensional, interacting, and non-renormalizable theory.  One therefore expects the theory to have an infinite number of counter-terms, and therefore an infinite number of coupling constants. 

In fact, there is an even more serious potential problem: as one proceeds in the loop expansion, there are infinitely many potential anomalies to quantization. So it is not at all obvious that any quantum theory even exists.

One of the main properties of quantum field theory is \emph{locality}, which allows us to build up the whole quantum theory from local data on the underlying manifold. In this paper, we analyze the problem of specifying a canonically-defined quantization of BCOV theory on a local piece of odd dimensional Calabi-Yau manifold, i.e., a flat space $\C^d$ for any odd $d$.   We do this by considering the coupling of BCOV theory to holomorphic Chern-Simons theory. Holomorphic Chern-Simons is the open-string field  theory for the topological $B$-model, and was introduced by Witten \cite{Wit92}.  We can call the coupled theory open-closed BCOV theory. The precise theorem is the following.
\begin{theorem*}
On $\C^d$ for $d$ odd, there exists a unique perturbative quantization of open-closed BCOV theory, where in the open sector the gauge Lie algebra is the super Lie algebra $\gl(N \mid N)$, and we require the quantization to be compatible with inclusions $\gl(N \mid N) \into \gl(N + k \mid N + k)$. 
\end{theorem*}

In other words, we construct a canonical quantization of coupling BCOV theory with large $N$ holomorphic Chern-Simons theory (with gauge group $\gl(N\mid N)$). In particular, one gets a quantization of the purely closed theory.  (Roughly speaking, this is for the same reason that closed string theory can exist on its own but that open string theory requires closed strings).    As a corollary, we find 
\begin{corollary*}
 There is a canonical quantum BCOV theory on $\C^d$ (for $d$ odd) which extends to a quantization of the coupled open-closed theory. 
\end{corollary*}
In other words, all the infinitely many possible counter-terms that can appear because the theory is non-renormalizable are uniquely fixed by the requirement of compatability with holomorphic Chern-Simons.

Some comments are in order:
\begin{enumerate} 
\item It is not possible to quantize holomorphic Chern-Simons theory by itself without coupling to the closed-string field theory. There is a one-loop anomaly in holomorphic Chern-Simons theory which is cancelled by the closed string sector.  
\item It is not possible to quantize the coupled open-closed theory where on the open sector we use the Lie algebra $\gl(N)$ as our gauge Lie algebra, because of a one-loop anomaly.
\item There is a variant of this story (in $3$ complex dimensions) which works for holomorphic Chern-Simons theory where the Lie algebra of the gauge group is $\mf{sl}(N \mid N)$. In this case, to cancel the anomaly, we need to use a variant of BCOV theory which we call $(1,0)$ BCOV theory.  The terminology is because of a conjectural relationship with the $(1,0)$ tensor multiplet in $6$ dimensions, which is similar to a relationship studied in the literature between BCOV theory and the $(2,0)$ tensor multiplet.   We construct a unique quantization of the system coupling $(1,0)$ BCOV theory to $\mf{sl}(N \mid N)$ holomorphic Chern-Simons theory, not just on $\C^3$ but on a variety of non-compact Calabi-Yaus. 
  \end{enumerate}
The theorem is proved using the obstruction theory methods developed in \cite{Cos11} for perturbative quantization.  For any quantum field theory, there is a cochain complex (built from possible Lagrangians) whose $H^0$ describes possible deformations of the theory, and whose $H^1$ describes anomalies, which are obstructions to quantization. For a renormalizable theory, one needs only consider a scale-invariant subcomplex of this obstruction-deformation complex, which will typically have finite dimensional cohomology groups.  In this case, $H^0$ is the space of marginal deformations, which is typically finite dimensional, and the finite dimensionality of $H^1$ indicates that there are only finitely many possible anomalies. For a non-renormalizable theory, the scale invariance argument does not apply, and $H^0$ and $H^1$ are both typically infinite-dimensional.

For BCOV theory on its own, or holomorphic Chern-Simons on its own, this is what we find. However, for the coupled open-closed BCOV theory, we find a remarkable cancellation: beyond one loop, the obstruction-deformation complex for the coupled theory has zero cohomology.  Potential obstructions and deformations for the closed-string sector are precisely cancelled by obstructions and deformations from the mixed sectors and from the purely open sector.  We will explain the heuristics of this argument in more detail later in the introduction. 

At one loop, we have to perform a detailed calculation to verify that a possible anomaly is cancelled.   This cancellation is very similar to the Green-Schwartz mechanism for anomaly cancellation.  
 
\subsection{Relationship to string theory}
Before discussing the proof in detail, we will explain a little about how this result relates to ideas from string theory, and in particular explain a conjectural framework relating our open-closed BCOV theory in $6$ real dimensions to certain field theories of current interest in physics (see \cite{GaiTom14} and the references therein).  

BCOV theory on $\C^3$ is the closed-string field theory for the topological $B$-model in three complex dimensions.  That is, the fields of BCOV theory are the ($S^1$-equivariant) closed-string states of the topological $B$-model, and the classical BCOV action functional can be described in terms of genus zero correlation functions of the $B$-model. 

The yoga of string theory tells us that for any string theory, one can construct a gravity theory as the low-energy limit of the closed-string field theory. For example, type IIB supergravity is supposed to arise in this way from type IIB string theory. This gravity theory will, of course, be non-renormalizable, and will have infinitely many coupling constants and infinitely many possible anomalies. However, the low energy limit of string theory is supposed to provide a canonically defined quantization of this non-renormalizable theory, which is why string theory can produce a quantum theory of gravity.  

BCOV theory fits into this framework. It is a non-renormalizable theory of gravitational type (gravitational, because the fields of BCOV theory describe fluctuations of the complex structure of a Calabi-Yau, hence the Calabi-Yau metric).  However, the yoga of string theory tells us that the topological string should provide a canonically defined quantization. This has not been rigorously proven: although categorical methods \cite{KonSoi06, Cos07a,Lur09} can provide a rigorous theory of the topological $B$-model, it is not known how to use these methods to produce something satisfying the locality axioms of a quantum field theory.  

What this paper achieves is to produce the canonically-defined quantization of BCOV theory that one expects from string theory yoga, but \emph{without} using topological string theory.  Although our methods are inspired by string theory, we do not use the string world-sheet at all.   We find that a simple compatibility between BCOV theory and holomorphic Chern-Simons theory is enough to fix the quantization uniquely.  

One can hope that our methods apply beyond topological strings, and give a new approach to quantizing other string theories.  We hope to investigate this in future publications.

\subsection{Twisted supergravity}
One point at which our work connects directly to string theory is via the following conjecture.
\begin{conjecture}
BCOV theory on $\C^5$ is a twist of type IIB supergravity theory on $\R^{10}$. 
\end{conjecture}
This conjecture will be explored in more detail in other publications.  

The concept of twisted supergravity is a slightly subtle one, and we will discuss it in more detail elsewhere. For now, let us explain the main idea.  Any supergravity theory has gauged local supersymmetry. To quantize the theory, one introduces Fadeev-Popov ghosts corresponding to local supersymmetries.  The ghosts for fermionic symmetries will be bosonic fields.  

Twisted supergravity is simply supergravity in perturbation theory around an unusual background, where the bosonic ghost field has some non-zero value.  To satisfy the equations of motion, this bosonic ghost field must be a (generalized) Killing spinor and have square zero.  One reason for using the terminology ``twisted supergravity'' is that, if we consider a supersymmetric field theory in a supergravity background where the bosonic ghost field takes value  some supercharge $Q$, this has the same effect as replacing the observables of the supersymmetric field theory by their $Q$-cohomology \footnote{Or more precisely, adding $Q$ to the BRST operator of the theory}.  This procedure is often called twisting\footnote{More properly, in the terminology of Witten, twisting involves both changing the action of the Lorentz group on the space of fields and then adding such a $Q$ to the BRST operator. It is often useful to use the term twisting for the more general procedure where one simply adds a supercharge to the BRST operator. }.

Different choices of bosonic ghost field give us different twists of supergravity.  We need to describe which particular twist of type IIB supergravity we conjecture is equivalent to BCOV theory on $\C^5$.  In type IIB supergravity in the flat Minkowski background, where all other bosonic fields are zero, there are $32$ covariant constant spinors. There is a unique $SU(5)$ invariant spinor $Q$ which is of weight one under the ``naive'' $R$-symmetry group $SO(2)$. The twisted supergravity we have in mind is the one where the bosonic ghost field takes value $Q$.

If one accepts this conjecture, then the results of this paper give us a new method to quantize part of type IIB supergravity.  

The open-string analog of the twist we consider was analyzed  by Baulieu \cite{Bau10}, who showed that $5$-complex dimensional holomorphic Chern-Simons is the holomorphic twist of $10$-dimensional maximally supersymmetric gauge theory. It is therefore natural to conjecture that the coupled open-closed BCOV theory we construct is the holomorphic twist of the theory coupling type IIB supergravity with the maximally supersymmetric $10$ dimensional gauge theory with gauge Lie algebra $\mf{gl}(N \mid N)$. 

In fact, Baulieu's result gives some strong evidence for the conjeture that BCOV theory is a twist of type IIB supergravity.  As we will see shortly, the fields of BCOV theory with the linearized BRST operator are the universal cochain complex which can couple to holomorphic Chern-Simons by single-trace operators.  At the classical level, type IIB supergravity can be coupled to maximally supersymmetric gauge theory in $10$ dimensions, which is the gauge theory living on the $D9$ brane, and this coupling is by single-trace operators.  This implies that twisted type IIB supergravity can be coupled to the twisted $10$ dimensional gauge theory, which is holomorphic Chern-Simons.   Since BCOV theory is the universal object admitting such a coupling, we find in this way a cochain map from the fields of twisted supergravity (equipped with the linearized BRST operator) to those of BCOV theory. 

\subsection{Relation to $6$ dimensional theories with $(1,0)$ supersymmetry}
There is a conjectural relationship between open-closed BCOV theory and certain $6$-dimensional field theories constructed from branes in type IIA string theory and in $M$-theory. The $(1,0)$ supersymmetry algebra in $6$ dimensions has a unique $SU(3)$-invariant supercharge $Q$, up to rotation by the $R$-symmetry group.  Taking the cohomology with respect to such a $Q$ gives a twist of any $(1,0)$ supersymmetric field theory. This twist is holomorphic in the sense discussed in \cite{Cos11b}.  

The following conjecture grew out of conversations with Davide Gaiotto.
\begin{conjecture}
Consider the following system of branes in type IIA string theory on $\R^{10}$. We have a single $NS5$ brane on $\R^6 \subset \R^{10}$, where the $\R^6$ is the locus where $x_6,\dots,x_9$ are zero.  We have $k$ semi-infinite $D6$ branes spanning the half space where $x_6 \ge 0$, $x_7,x_8,x_9 = 0$, and $k$ semi-infinite $D6$ branes spanning the opposite half-space where $x_6 \le 0$, $x_7,x_8,x_9 = 0$.

On the world-volume for the $NS5$ brane in this configuration is a theory with $(1,0)$ supersymmetry.  The conjecture is that  the  holomorphic twist of this theory is open-closed BCOV theory where the Lie algebra of the gauge group for the open sector is $\mf{gl}(k \mid k)$. 
\end{conjecture}
Let us discuss the evidence for this conjecture. Let us first discuss the case $k = 0$, in which case the conjecture relates BCOV theory to the theory living on the $NS5$ brane.

The fields of the $NS5$ brane theory form a single tensor multiplet with $(2,0)$ supersymmetry. One of the fields of BCOV theory is a closed $(2,1)$ form, which one expects to identify with $9$ of the $10$ components of the self dual $3$-form of the free $(2,0)$ tensor multiplet. Marino, Minasian, Moore and Strominger argued in \cite{MarMinMoo99} that the supersymmetric equations of motion of the $M5$ brane theory on a Calabi-Yau manifold include the Kodaira-Spencer equations which are the equations of motion of BCOV theory.  Kapustin also argued \cite{Kap04} that the partition function of the type IIA $NS5$ brane theory is the same as the partition function of the topological $B$-model, i.e. of BCOV theory. 

This provides some evidence in the case $k = 0$.  For the more general case, one needs to understand a relationship between the $D6$ brane gauge theory and holomorphic Chern-Simons, which we will sketch without proof (and hopefully return to in a future publication).  The claim is that there is a certain supercharge $Q$ in the supersymmetry algebra acting on the $D6$ brane gauge theory, which is $SU(3)$-invariant and also has the feature that every translation is $Q$-exact.  Up to rotation by the $R$-symmetry group and scaling, these two features fix $Q$ uniquely.  The $Q$-cohomology of the $D6$ brane gauge theory is then a topological theory in the weak sense that the operators corresponding to translation in space time are trivial. The claim is that this twisted $7d$ theory has a non-topological boundary condition where holomorphic Chern-Simons theory lives on the boundary. 

From this claim, it is easy to imagine how $\mf{gl}(k \mid k)$ holomorphic Chern-Simons theory arises. One collection of $D6$ branes will contribute $\mf{gl}(k)$ holomorphic Chern-Simons theory living on the boundary, which is the $NS5$ brane. The other collection of $k$ $D6$ branes will also contribute $\gl(k)$ holomorphic Chern-Simons, but with the opposite level.  Strings connecting the two collections of $D6$ branes give bi-fundamental matter, which when twisted become the fields of holomorphic Chern-Simons corresponding to the odd elements of the Lie algebra $\gl(k \mid k)$. 

This conjecture is in a sense $T$-dual to a result of Mikhalyov and Witten \cite{MikWit14}. They considered a configuration of branes in type IIB consisting of $2k$ $D3$ branes ending on a $3$-manifold living inside an $NS5$ brane, with $k$ $D3$ branes coming from each side. In this situation, they showed that after $Q$-cohomology for a certain supercharge $Q$, the theory living on the $3$-manifold is  ordinary Chern-Simons theory for the group $GL(k \mid k)$. 

 If one applies $T$-duality to $3$ of the $6$ directions of the IIA $NS5$ brane in our situation, the $D6$ branes are  converted into  $D3$ branes and one goes from our situation to theirs. 

To understand the link between the work of Mikhalyov and Witten and our conjecture, one needs to understand why holomorphic Chern-Simons and ordinary Chern-Simons should be $T$-dual.  According to Kapustin \cite{Kap04}, one expects that the $NS5$ brane theory in type IIB can be described by the $A$-model topological string.  $T$-dualizing $3$ directions in a type IIB $NS5$ brane brings us to a type IIA $NS5$ brane, which should be the topological $B$-model. This relationship should be mirror symmetry between the topological $A$- and $B$-models.  Holomorphic Chern-Simons theory lives on a space-filling brane in the topological $B$-model, and this brane is converted under $T$-duality to a $3$-dimensional brane in the topological $A$-model.  Ordinary Chern-Simons theory is the field theory living on a brane in the topological $A$-model, and thus is the $T$-dual.  

We also consider a variant of BCOV theory which we call $(1,0)$ BCOV theory, only in three complex dimensions. We conjecture that free $(1,0)$ BCOV theory is the holomorphic twist  of the free $(1,0)$ tensor multiplet.   We will show that $(1,0)$ BCOV theory can be coupled to $\mf{sl}(N \mid N)$ holomorphic Chern-Simons theory and that the coupled theory admits a unique quantization. Davide Gaiotto suggested that the following might be true. 
\begin{conjecture}
Consider a single $M5$ brane on $\R^7 \times (\R^4 / \Z_N)$, where the $M5$ brane lives at the singular point in the $A_N$ singularity and spans $6$ of the $7$ directions in $\R^6$.

There is a theory with $(1,0)$ supersymmetry living on the world-volume of this $M5$ brane \cite{GaiTom14}.  Then, the conjecture is that the holomorphic twist of this theory is $(1,0)$ BCOV theory coupled to $\mf{sl}(N \mid N)$ holomorphic Chern-Simons theory. 
\end{conjecture} 
Let's present some evidence for this conjecture.  Standard $M$-theory philosophy says that $M$-theory on $\R^7 \times (\R^4 / \Z_k)$ is equivalent to $7$-dimensional maximally supersymmetric gauge theory with group $SU(k)$.  The $M5$ brane can be viewed as a domain wall from the $7d$ gauge theory to itself. Further, the motion of the $M5$ brane is described by a $(1,0)$ tensor multiplet on the brane, and after twisting should produce the fields of $(1,0)$ BCOV theory.  The two copies of the gauge theory living on each side of the $M5$ brane have a boundary condition on the brane, which is (conjecturally) the same as the one discussed earlier which introduces holomorphic Chern-Simons.  In this way, we find that the theory on the $M5$ brane is, after taking $Q$-cohomology, $\mf{sl}(k \mid k)$ holomorphic Chern-Simons coupled to $(1,0)$ BCOV theory.

\subsection{The fields of BCOV theory}
Let us now introduce the fields of BCOV theory, focusing on dimension $3$ for simplicity.

Let $\Omega^{-\ast,\ast}(\C^3)$ denote the space of differential forms on $\C^3$, with degrees arranged so that a $(p,q)$ form is in degree $q - p$.  The fields of BCOV theory, in our formulation, can be identified with 
$$
\Omega^{-\ast,\ast}(\C^3)[[t]] [-1]
$$
where $t$ is a formal variable of cohomological degree $2$. The symbol $[-1]$ indicates a shift of cohomological degrees, so that $\omega \in \Omega^{p,q}(\C^3)$, then the field $t^k \omega$ is in degree $2k+q-p+1$. The linearized BRST differential on the space of fields of BCOV theory is the operator $\dbar + t \partial$.  

It is often more convenient to identify the fields of BCOV theory with polyvector fields, which are isomorphic to forms. For now, however, since we will not yet be discussing the interaction in BCOV theory, we will stick to forms.  

In the original formulation, the fields of BCOV theory is the space
$$
\op{Ker} \partial \subset \Omega^{-\ast,\ast}(C^3)[-1].
$$
In particular, the fields contain a closed $(2,1)$-form of cohomological degree $0$, corresponding to a variation of complex structure.  The variable $t$ that we introduce plays the role of descendants.  The fact that the linearized BRST operator includes the term $t \partial$ in our formulation replaces the fact that in the original formulation the space of fields is $\Ker \partial$.

We formulate BCOV theory in the BV formalism, in which to specify a theory one needs to specify a differential graded manifold with an odd symplectic structure.  BCOV theory is a degenerate theory in this sense: the space of fields has an odd Poisson, as opposed to odd symplectic, structure. The Poisson kernel is 
$$
\pi = (\partial \otimes 1)\delta_{Diag}
$$
where $\delta_{Diag}$ is the delta-function on the diagonal in $\C^3 \times \C^3$. Thus, $\delta_{Diag}$ is a form on $\C^3 \times \C^3$ with distributional coefficients.  

This Poisson kernel lives in the tensor square of the space of fields. In fact, it lives in the tensor square of the subspace of fields which have no powers of $t$ in them.

Although this is a very abstract way of formulating a theory, one can still construct familiar objects such as the propagator in this formalism. The propagator is
$$
P = (\partial \dbar^{-1} \otimes 1) \delta_{Diag}
$$
where $\dbar^{-1}$ is the inverse to the $\dbar$ operator. (Explicitly, $\dbar^{-1} = \dbar^\ast \tr^{-1}$ where $\tr^{-1}$ is the Green's operator for the Laplacian).  The descendant fields -- those involving powers of $t$ -- do not propagate, and thus can be viewed as background fields. Including descendant fields is natural for various reasons, and is especially important away from dimension $3$.    

A variant of our BCOV theory which is very close to the original formulation is obtained if one only considers those fields in $\Omega^{k,\ast} t^l$ where $k+l \le 3$. Fields of this nature form a subcomplex of the full space of fields we wrote above, and contain all the propagating fields.  This smaller space of fields can be viewed as a direct sum of $4$ complexes of the form 
$$\Omega^{i,\ast} \xto{\partial} t \Omega^{i+1,\ast} \dots   \xto{\partial} t^{3-i} \Omega^{3,\ast}$$
where $0 \le i \le 3$. 

These complexes can be viewed as resolutions of the subspace $\op{Ker} \partial \subset \Omega^{i,\ast}$, so that this smaller space of fields is a resolution of the space of fields of the original BCOV theory. 

The remainder of the space of fields of BCOV theory consists of infinitely many copies of the de Rham complex of the form
$$
t^k \Omega^{0,\ast} \to t^{k+1} \Omega^{1,\ast} \to t^{k+2} \Omega^{2,\ast} \to t^{k+3} \Omega^{3,\ast}
$$
where $k \ge 1$. These fields do not propagate, but do interact. They can thus be seen as background fields, and by setting these background fields to zero one obtains a quantization of the smaller space of fields.  These background fields are not essential for the purposes of this paper.  We include, them, however, because they play an important role in our approach to a holomorphic twist of the AdS/CFT correspondence, which we will discuss in future work.  

\subsection{Holomorphic Chern-Simons}
The fundamental field of holomorphic Chern-Simons theory is a connection $A \in \Omega^{0,1}(\C^3,\gl_N)$ with action
$$
S_{hCS}(A) = \int_{\C^3} \tfrac{1}{2} \op{Tr} A \dbar A + \tfrac{1}{3} \op{Tr} (A^3) .
$$
The integration is again the standard holomorphic volume form on $\C^3$. The field $A$ is acted on by the complex gauge group of maps from $\C^3$ to $GL(n,\C)$, and this action preserves $S_{hCS}$. In the BV formalism, the space of fields (including ghosts, anti-fields etc.) is $\Omega^{0,\ast}(\C^3)\otimes \gl_N[1]$, and the full BV action functional takes the same form as $S_{hCS}$ above, except that now $A$ is no longer constrained to be a $(0,1)$-form.

We will denote the cubic term in the holomorphic Chern-Simons interaction by $I^{CS}$; we should think of it as being associated to a disc with three marked points on the boundary. 

One can couple BCOV theory and holomorphic Chern-Simons. We are only interested, for now, in the term in the coupled action which depends linearly on the fields of BCOV theory. We call this term $I^{1-disk}$, and view it as being associated to a disc with one marked point in the middle (at which we put the field of BCOV theory) and any number of marked points on the boundary. 

The explicit formula for $I^{1-disk}$ is a little complicated (see Definition \ref{defn-1-disk}). To get a sense of its structure, we give here some expressions for the coupling between those fields on BCOV theory which are of cohomological degree zero and holomorphic Chern-Simons. 
\begin{enumerate} 
 \item If $\alpha \in \Omega^{3,2}(\C^3)$, then $\alpha$ couples to holomorphic Chern-Simons via the interaction
$$
A \mapsto \int \alpha \op{Tr} A.
$$
\item If $\alpha \in \Omega^{2,1}(\C^3)$ with $\partial \alpha = 0$, then $\alpha$ couples by the interaction
$$
A \mapsto \tfrac{1}{2} \int \alpha \op{Tr} (A \partial A).
$$
\item A field $\alpha \in \Omega^{1,0}(\C^3)$  with $\partial \alpha = 0$ couples via
$$
A \mapsto \tfrac{1}{3} \int \alpha \op{Tr} (A (\partial A)^2 ).
$$
\item A field $\alpha \in t \Omega^{3,0}(\C^3)$ couples by
$$
A \mapsto \tfrac{1}{3} \int \alpha \op{Tr} (A^3). 
$$
\end{enumerate}
\subsection{$(1,0)$ BCOV theory}
Let $X$ be a Calabi-Yau $3$-fold. The fields of $(1,0)$ BCOV theory are the subset of those fields of the full BCOV theory consisting of the complex
$$
\Omega^{2,\ast}[1] \xto{\partial} t \Omega^{3,\ast}.
$$
The fields of ghost number zero are $\Omega^{2,1} \oplus t \Omega^{3,0}$, and these fields couple to those of holomorphic Chern-Simons theory for $\mf{sl}(N \mid N)$ by the formulae we wrote down above.

We can rewrite these fields in terms of holomorphic vector fields on $X$. Contracting with  the holomorphic volume form gives an isomorphism between the spaces
$$\Omega^{0,*}(X,T X) \iso \Omega^{2,*}(X), \quad  \Omega^{0,\ast}(X)\iso \Omega^{3,\ast}(X) .$$
It allows us to rewrite our space of fields as 
$$
\Omega^{0,\ast}(X, T X)[1] \oplus \Omega^{0,\ast}(X).
$$
The differential is a sum of the $\dbar$ operator and the holomorphic divergence map
$$
\partial = \op{Div} : \Omega^{0,k}(X, T X) \to \Omega^{0,k}(X).
$$
In this formulation, fields of cohomological degree zero are $\Omega^{0,1}(X,T X)$, describing deformations of complex structure; and $\Omega^{0,0}(X)$, describing changes in the holomorphic volume form.  The fields  of cohomological degree $-1$ (i.e.\ ghosts) are $\Omega^{0,0}(X, T X)$. These are ghosts for holomorphic changes of coordinates.

Let us describe the interaction for $(1,0)$ BCOV theory.  If $\alpha$ denotes the field in $\Omega^{0,\ast}(X, T X)[1]$ and $\phi \in \Omega^{0,\ast}(X)$, the interaction is
$$
I(\alpha,\phi)= \sum_{n \ge 0} \tfrac{1}{6} \int \phi^n \Omega \wedge (\alpha \vdash)^3 \Omega
$$
where $\Omega$ is the holomorphic volume form, and $(\alpha \vdash)^3$ indicates the operation of contracting with $\alpha$ $3$ times.

One can show that the equations of motion (including the linearized BRST operator we have discussed) describe the variations of $X$ as a Calabi-Yau manifold equipped with a holomorphic volume form.  This is in contrast to the full BCOV theory, in which the equations of motion describe a much larger space including non-commutative deformations of $X$.   

We conjecture that free $(1,0)$ BCOV theory is the holomorphic twist of the free $(1,0)$ tensor multiplet in $6$ dimensions.  The idea is that the propagating field of ghost number zero $(1,0)$ BCOV theory that lives in $\Omega^{2,1}$ corresponds to $9$ of the $10$ components of the self-dual $3$-form in the $(1,0)$ tensor multiplet.  
 
Our theorem regarding quantization of $(1,0)$ BCOV theory is the following.
\begin{theorem*}
There is a unique quantization of $(1,0)$ BCOV theory coupled to $\mf{sl}(k \mid k)$ holomorphic Chern-Simons theory, and compatible with certain natural symmetries, in the following situations:
\begin{enumerate}
\item On $\C^3$.
\item On a Calabi-Yau $X$ which is the total space of the canonical bundle of a complex surface.  
\item On the product of a $K3$ surface with an elliptic curve. 
\item At genus $0$, on any Calabi-Yau.  From the point of view of the holomorphic Chern-Simons gauge theory, genus $0$ means we only consider planar diagrams.  
\end{enumerate}
\end{theorem*}
 
\subsection{The interaction between general BCOV fields and holomorphic Chern-Simons}
We will sketch the cohomology cancellation argument which allows us to quantize open-closed BCOV theory. To understand this argument, it is useful to have some understanding of how a general field of BCOV theory couples to holomorphic Chern-Simons (although we will not give explicit formulae right now). 

Let's consider the most general possible single-trace Lagrangian we can write down which is a first order deformation of the holomorphic Chern-Simons interaction.  If we require our Lagrangian to be $\GL(N \mid N)$ invariant and compatible with the inclusions $\gl(N \mid N) \into \gl(N + k \mid N + k)$, then the general Lagrangian is a sum of terms of the form 
\begin{equation*}
\int_{\C^3} \prod \d z_i \prod \d \zbar_i  \alpha \op{Tr}\left( (D_0A_{r_0}) (D_1 A_{r_1}) \dots (D_n A_{r_n}) \right) \tag{$\dagger$} 
\end{equation*}
where $\alpha \in \cinfty(\C^3)$, $0 \le r_i \le 3$, $A_{r_i}$ indicates the component of the field $A$ in $\Omega^{0,r_i}$, and each $D_i$ is a constant-coefficient differential operator from $\Omega^{0,r_i}(\C^3)$ to $\cinfty(\C^3)$. Up to integration by parts, we can assume $D_0$ is a linear map without any derivatives.

The space of such Lagrangians is a cochain complex with a differential given by $\{S^{hCS},-\}$ where $\{-,-\}$ is the BV bracket.  Consistent first-order deformations of the holomorphic Chern-Simons action  are given by closed elements of this cochain complex, and cohomologous elements define equivalent first-order deformations.  
A result of Tsygan \cite{Tsy83} and Loday-Quillen \cite{LodQui84} allows us to identify this complex of Lagrangians with the (local) cyclic cochain complex of the dg algebra $\Omega^{0,\ast}(\C^3)$.  It is important here that we are working uniformly in the $N$ appearings in $\gl(N \mid N)$. Because of this, trace relations do not appear, and expressions like ($\dagger$) give the same Lagrangian if they are related by a cyclic permutation or a total derivative, which is why one finds the cyclic cochain complex.  

A classic result of Hochschild-Kostant-Rosenberg allows one to calculate this local cyclic complex, and one finds that it is quasi-isomorphic to the complex $\Omega^{-\ast,\ast}(\C^3)[[t]][-1]$, with differential $\dbar + t \del$. This, however, is the complex of fields of BCOV theory.

In other words, the fields of BCOV theory are the universal object which can couple to holomorphic Chern-Simons theory with single-trace operators.  

\subsection{The cohomology cancellation argument} 
Recall that we are aiming to couple our open-closed BCOV theory where on the open string sector we have the groups $\gl(N \mid N)$, and we work uniformly in $N$.  The fact that we are working uniformly in $N$ implies that Feynman diagrams for the open sector can be viewed as ribbon graphs, and that more generally every Feynman diagram for the open-closed theory gives us a topological type of a Riemann surface of some genus $g$, with some number $h$ of boundary components, some number $n$ of interior marked points, and some number of marked points on each boundary component. The closed-string fields are placed on the interior marked points and the open-string fields on the boundary marked points.  The fact that we are using the super Lie algebras $\gl(N \mid N)$ instead of $\gl(N)$ tells us that only surfaces each of whose boundary components has at least one marked points can appear. 

It is important to bear in mind that we introduce Riemann surfaces only as a combinatorial tool for describing topological types of Feynman diagrams.  We are strictly doing string field theory, and we will not use the world-sheet theory at all.

It makes sense to try to consider our theory up to genus $G$, by only considering diagrams of genus less than $G$. Also, given an integer $R$, it makes sense to consider our theory where we consider diagrams of genus $g < G$, of genus $G$ where $h+n < R$. Recall $h$ is the number of boundary components and $n$ the number of interior marked points. We will indicate this by saying that $(g,h+n) < (G,R)$ using the lexicographical ordering on pairs of integers.

Our construction of the open-closed theory is by induction.  Suppose we have constructed our theory for $(g,h,n)$ where $(g,h+n) < (G,R)$.  Then, we will construct it for $(g,h,n) < (G,R+1)$ by induction.  The obstruction-deformation group describing this problem is built from Lagrangians defined on the open-closed fields which, on the open sector, have $h$ traces, and which are homogeneous of degree $n$ as a function of the closed-string fields, where $h+n = R$.  

The main claim is that this complex vanishes.  Let us illustrate this point for the simple case when $R = 1$.

In this case, our complex has two terms: single-trace Lagrangians  of the open-string fields, and linear functionals of the closed-string fields.  The differential on this complex has three terms: the open-string BRST differential, which acts only on the open-string sector; the closed-string BRST differential, acting on the closed-string sector; and the BV bracket with the interaction $I^{1-disk}$.  This last term maps a linear functional of the closed-string fields to a single-trace functional of the open-string fields. 

As we have sketched above, single-trace first-order deformations of the open-string sector are described by the cyclic cohomology of the algebra $\Omega^{0,\ast}(\C^3)$, which is precisely the complex $\Omega^{-\ast,\ast}(\C^3)[[t]][-1]$ of fields of the closed-string sector.

The space of linear functionals of the closed string fields is $t^{-1} \Omega^{-\ast,\ast}(\C^3)[t^{-1}]$.  The pairing between an element $\phi \in t^{-1}\Omega^{-\ast,\ast}(\C^3)[t^{-1}][-1]$ and a closed-string field $\alpha \in \Omega_c^{-\ast,\ast}(\C^3)[[t]][-1]$ is 
$$
\phi(\alpha) = \sum_{k \ge 0} \int \phi_{-k-1}\wedge \alpha_k
$$
where $\phi_{l}$ and $\alpha_l$ indicate the coefficient of $t^l$.  The differential $\dbar + t \partial$ on $\Omega^{-\ast,\ast}(\C^3)[t^{-1}]$ is the dual to the differential $\dbar + t \partial$ on $\Omega^{-\ast,\ast,}(\C^3)[[t]]$.  We use the convention that $t \partial $ applied to $\phi_{-1} t^{-1}$ is zero.  

Thus, our obstruction-deformation group looks like
$$
\Omega^{-\ast,\ast}(\C^3)[[t]][-1] \oplus t^{-1}\Omega^{-,\ast,\ast}(\C^3)[t^{-1}][-1] 
$$
where the first summand comes from the open-string sector, and the second from the closed-string sector. Using the language of cyclic cohomology, we can say that the obstruction-deformation complex for the open string sector yields local cyclic cochain complex of $\Omega^{0,\ast}(\C^3)$, whereas that for the closed string sector yields \emph{negative} local cyclic complex.

Our obstruction-deformation complex is \emph{almost} the complex $\Omega^{-\ast,\ast}(\C^3)((t))[-1]$ with differential $\dbar + t \partial$.  The only issue is that we are missing the term in the differential which sends a closed--string Lagrangian $\phi t^{-1}$ to the open-string Lagrangian $\partial \phi t^0$.  In the cyclic cohomology language, our complex is almost the periodic cyclic cochain complex, except that we are missing the connecting map between negative cyclic cochains and cyclic cochains. 

There is a remaining term in the differential, which maps the closed-string to the open-string sectors, coming from the closed-string BV bracket with the interaction term $I^{1-disk}$.  

We calculate that this differential is \emph{precisely} the missing term mentioned above, so that the full obstruction-deformation complex is $\Omega^{-\ast,\ast}(\C^3)((t))[-1]$ with differential $\dbar + t \partial$ (that is, the period cyclic cochain complex).  This complex is simply a direct sum of infinitely many copies of the de Rham complex of $\C^3$, of the form
$$
\Omega^{0,\ast}(\C^3) t^k \to \Omega^{1,\ast}(\C^3) t^{k+1} \to \Omega^{2,\ast}(\C^3) t^{k+2} \to \Omega^{3,\ast}(\C^3) t^{k+3} 
$$
As such, the cohomology is simply $\C((t))[-1]$.  

What this shows is that, although if we concentrate on the purely open or purely closed string sectors there are lots of complicated possible Lagrangians, for the coupled theory every Lagrangian is equivalent to a sum of the simple ones corresponding to the elements in $t^k 1 \in \Omega^{0,0}(\C^3)((t))$. A similar calculation applies when $R > 1$, except that we find the space of possible Lagrangians is isomorphic to a symmetric power of $\C((t))[-1]$.  

The Lagrangians corresponding to $t^k 1$ have very simple descriptions. If $k \ge 0$, then this Lagrangian is a function of the open-string fields $A \in \Omega^{0,\ast}(\C^3, \gl(N \mid N))$ of the form
$$
A \mapsto \sum_{l_1 + l_2 + l_3 = 2k+1} c_{l_1,l_2,l_3}  \int \op{Tr} A^{l_1} (\partial A) A^{l_2} (\partial A) A^{l_3} (\partial A)
$$
where $c_{l_1,l_2,l_3}$ are certain combinatorial constants. 

If $k < 0$, then the Lagrangian corresponding to $t^k 1$ is a function of the closed-string field $\phi \in \Omega^{-\ast,\ast}(\C^3)[[t]][-1]$ of the form
$$
\phi = \sum t^k \phi_k \mapsto \int \phi_{-k-1}. 
$$
There are similar expressions when $R > 1$. 

In either case, we see that the Lagrangian is invariant under scaling of $\C^3$. In topological string theory, one expects quantities associated to surfaces of type $(g,h,n)$ to have weight\footnote{Whether the weight is $3(2g-2+h+n)$ or $3(2g-2+h)$ is convention dependent. If one thinks of the fields of BCOV theory as being polyvector fields, then the weight is $2g-2+h$, and if one thinks of them as forms the weight is $3(2g-2+h+n)$. The difference is accounted for by the fact that the isomorphism between forms and polyvector fields depends linearly on the holomorphic volume form.} $3(2-2g-h-n)$ when we scale $\C^3$,  and we can make this into a consistent axiom for quantum open-closed BCOV theory.  

It follows that the Lagrangians we have found do not scale correctly to contribute to the obstruction-deformation complex for quantizing open-closed BCOV theory (in a way compatible with scaling of $\C^3$), except when $2g-2+h+n = 0$. From this we see that once we have specified the quantum theory for all diagrams with $2g-2+h+n = 0$, the axioms specify the rest of the theory uniquely.

If $2g-2+h+n = 0$, then we have either a disc with one interior marked point, an annulus with no interior marked points, or a torus with no interior marked points.  The disc with one interior marked point has already been specified. The torus with no interior marked points does not appear in our story, as we do not consider ``vacuum'' Feynman diagrams with no external lines whatsoever. It follows that it remains to construct the theory at the annulus level.

A similar cohomology cancellation holds when we consider $(1,0)$ BCOV theory, allowing us to quantize this theory to all loops once we have quantized to the annulus level. 
\subsection{Annulus anomaly cancellation}
The final thing to check for our argument is that we can construct the theory at the annulus level.  The cohomology cancellation argument we sketched above does not apply at the annulus level. One reason, of course, is that the possible Lagrangians we found are of weight zero and so can contribute to annulus diagrams. There is a more subtle reason, however.  The cohomology cancellation was between various possible Lagrangians appearing at $(g,h,n)$ for fixed $g$ and fixed $h+n$. If we try to apply this at $g = 0$ and $h+n=2$, we find that we have already specified (as part of our classical data) what happens when $(g,h,n)$ is $(0,1,1)$ and $(g,h,n) = (0,0,2)$. (The latter should be thought as  kinetic term in the classical BCOV action, encoded above by the linearized BRST operator).  Lagrangians corresponding to these two types of surface can not appear in our obstruction-deformation complex, so that the complex is entirely built from Lagrangians of type $(0,2,0)$. The cancellation described above therefore can not take place. 

We can explicitly compute possible anomalies to quantization at the annulus level. There are two possible anomalies, corresponding to the functionals of the open string field $A \in \Omega^{0,\ast}(\C^3,\gl(N \mid N))$:
\begin{align*} 
\alpha_1 ( A)  & = \int_{\C^3} \op{Tr} A \partial A \op{Tr} \partial(A)^2 \\
\alpha_2(A) & = \int_{\C^3} \op{Tr} A \op{Tr} (\partial A)^3. 
\end{align*}
There are no ambiguities to quantization at the annulus level, so that if the anomaly vanishes the quantization is unique. 

If we try to quantize holomorphic Chern-Simons on its own, without coupling to BCOV theory, we find (by an explicit computation) that the anomaly is $\alpha_1 / 2 + \alpha_2$.  

There is another possible source of annulus anomalies, however, coming from the closed-string sector.  The Lagrangian $I^{1-disk}$ couples the open and closed string sectors, and corresponds to a disk with one interior marked point. Applying the closed-string BV anti-bracket  $\{I^{1-disk},I^{1-disk}\}_C$ also gives an annulus-level anomaly, which we call the closed-string anomaly.  
\begin{theorem*}
The open and closed string anomalies precisely cancel for BCOV theory on $\C^d$ (where $d$ is odd).  Therefore, there is a unique quantum open-closed BCOV theory to all orders.  
\end{theorem*}
The fact that the quantization at the annulus level gives one to all orders follows from the cohomology cancellation argument we sketched earlier.

If we use $\gl(N)$ instead of $\gl(N \mid N)$, there is an extra one-loop open-string anomaly which is the functional of the open-string field 
$$
A \mapsto \int_{\C^3} \op{Tr} A (\partial A)^3 \op{Tr} 1. 
$$
Of course, $\op{Tr} 1$ is $N$ if we use $\gl(N)$ but zero for $\gl(N \mid N)$.  This extra term corresponds to an annulus with no marked points on one of its two boundaries. This term does not cancel with the closed-string sector. 

We also have the following variant, for $(1,0)$ BCOV theory.
\begin{theorem*}
The annulus anomaly cancels if we couple $(1,0)$ BCOV theory to $\mf{sl}(N \mid N)$ holomorphic Chern-Simons on any Calabi-Yau manifold of dimension $3$. 

As a corollary, $(1,0)$ open-closed BCOV theory admits a unique quantization on $\C^3$, and indeed on any Calabi-Yau which is the total space of the canonical bundle over a complex surface. (The cohomology cancellation argument we discussed earlier works on this class of Calabi-Yaus).   
\end{theorem*}

The point here is the following. If we use $\mf{sl}(N \mid N)$ holomorphic Chern-Simons, the term in the anomaly of the form $\int \op{Tr} A \op{Tr}\left(  (\partial A)^3 \right)$ is zero, and the anomaly is simply $\int \op{Tr} (A \partial A) \op{Tr} \left( (\partial A)^2\right)$.  

In the fields of BCOV theory, $\Omega^{2, \ast}$ couples to $2$ copies of $A$, $\Omega^{1,\ast}$ to $3$ copies of $A$ and $\Omega^{3,\ast}$ to $1$ copy of $A$.  If $I^{1-disk}_k$ refers to the term in $I^{1-disk}$ where on the interior marked point one places a field of BCOV theory in $t^0 \Omega^{k,\ast}$, then one finds that 
\begin{align*} 
 \{I^{1-disk}_{3}, I^{1-disk}_{1}\}^{C} (A) &= \int \op{Tr} A \op{Tr}\left(  (\partial A)^3 \right)\\
  \{I^{1-disk}_{2}, I^{1-disk}_{2}\}^{C} (A) &=   \int \op{Tr} (A \partial A) \op{Tr} \left( (\partial A)^2\right).
\end{align*}
The only propagating fields in $(1,0)$ BCOV theory are in $\Omega^{2,\ast}$, so the closed-string anomaly here is $\{I^{1-disk}_{2}, I^{1-disk}_{2}\}^{C}$, which cancels the anomaly from the open-string sector using $\mf{sl}(N \mid N)$.  

This annulus anomaly cancellation is very similar to the way that the introduction of $(1,0)$ tensor multiplets can cancel a one-loop anomaly appearing in a $6$ dimensional gauge theory with $(1,0)$ supersymmetries. 
\subsection{The classical BCOV interaction}
One slightly odd feature of our construction of open-closed BCOV theory is that our initial data, from which we produce the entire quantum theory in a unique way, does not include the BCOV interaction. Recall that our initial data is $I^{disk}$ and $I^{1-disk}$, corresponding to disks with zero and one interior marked points.  From this, we generate by the arguments sketched above the full quantum open-closed BCOV theory, and in particular the closed-string interactions $I^{n-sphere}$ associated to a sphere with $n$ marked points. In this subsection we will describe a conjectural description for these closed-string interactions, and relate it to Kontsevich's formality theorem.

It is convenient to rewrite the fields of BCOV theory in terms of polyvector fields, as this makes it easier to describe our conjectural formula for the interaction.
 
Let
\begin{align*} 
 \PV^{\ast,\ast}(\C^3) &=  \Omega^{0,\ast}(\C^3, \wedge^\ast T \C^3)\\
&= \cinfty(\C^3)[\d \zbar_i, \partial_{z_j}] 
\end{align*}
be the Dolbeault complex of $\C^3$ with coefficients in the bundle of poly-vector fields, which is the exterior algebra of the tangent bundle.  The variables $\d \zbar_i$ and $\partial_{z_j}$ are odd, of cohomological degree $1$. 

Let $\PV^{i,j}(\C^3)$ denote $\Omega^{0,j}(\C^3, \wedge^i T \C^3)$.  There is an isomorphism
$$
\PV^{i,j} (\C^3) \iso \Omega^{3-i,j}(\C^3)
$$
coming from the isomorphism of holomorphic bundles
$$
\wedge^i T \C^3 \iso \wedge^{3-i} T^\ast \C^3
$$
given by contracting with the holomorphic volume form.

The space of polyvector fields has two operators $\dbar$ and $\partial$, which correspond via the isomorphism with the de Rham complex to the usual $\dbar$ and $\partial$ operators.  Thus, $\dbar$ maps $\PV^{i,j}$ to $\PV^{i,j+1}$ and $\partial$ maps $\PV^{i,j}$ to $\PV^{i
-1,j}$.

In the language of polyvector fields, the space of fields of BCOV theory is
$$
\PV^{\ast,\ast}(\C^3)[[t]][2] \iso \Omega^{-\ast,\ast}(\C^3)[[t]][-1]. 
$$
with differential $\dbar + t \partial$.  

There is an integration map
$$
\int : \PV^{\ast,\ast}_c(\C^3) \to \C
$$
which is zero except on $\PV^{3,3}_c$, and which sends $\alpha \in \PV^{3,3}_c(\C^3)$ to
$$
\int \alpha := \int_{\C^3} \Omega \wedge (\alpha \vdash \Omega). 
$$
Here $\Omega = \d z_1 \d z_2 \d z_3$, and  $\alpha \vdash \Omega \in \Omega^{0,3}$ is the form obtained by contracting $\alpha$ with $\Omega$.
 
In \cite{CosLi11} we described a classical interaction for our formulation of BCOV theory, which we now recall. Define functionals
$$
I_n : \PV^{\ast,\ast}_c(\C^3)[[t]][2] \to \C
$$ 
as follows. If $\alpha \in \PV^{\ast,\ast}_c(\C^3)[[t]][2]$, let $\alpha_k$ denote the coefficient of $t^k$. Then, we set
$$
I_n(\alpha) = \sum_{k_1, \dots, k_n \text{ with } \sum k_i = n-3}\frac{(n-3)!}{k_1 ! \dots k_n !} \int \alpha_{k_1} \wedge \dots \wedge \alpha_{k_n}. 
$$
We then define the interaction $I$ by saying that
$$
I(\alpha) = \sum_{n \ge 3} \tfrac{1}{n!} I_n(\alpha). 
$$
One can check \cite{CosLi11} that $I(\alpha)$ satisfies the classical master equation. 

In this paper, we show that there is a unique quantum open-closed BCOV theory. When restricted to the classical closed-string sector this gives us some new interaction for BCOV theory, which we call $\what{I}$. Again, we can expand
$$
\what{I} (\alpha) = \sum_{n \ge 3} \tfrac{1}{n!} \what{I}_n(\alpha)
$$
where $\what{I}_n(\alpha)$ is homogeneous of degree $n$ as a function of $\alpha$.  The functional $\what{I}$ automatically satisfies the classical master equation.  
\begin{conjecture}
The functionals $I$ and $\what{I}$ are equivalent solutions to the classical master equation. 
\end{conjecture}
This conjecture is certainly not obvious.  In fact, it implies a variant of Kontsevich's formality theorem: 
\begin{uproposition}
This conjecture implies the holomorphic analog of Willwacher-Calaque's cyclic refinement of Kontsevich's formality theorem.
\end{uproposition}
As stated, it implies the $3$-dimensional version of the cyclic formality theorem, but since our construction of open-closed BCOV theory works in any odd dimension the obvious refinement of the conjecture implies Willwacher-Calaque's theorem in any odd dimension.

Even though we can't currently prove our conjecture, we can show the following.
\begin{ulemma}
The cubic interactions $I_3$ and $\what{I}_3$ are equivalent solutions of the classical master equation modulo  quartic terms. 
\end{ulemma}
In other words, to leading order, the dynamically-generated classical interaction $\what{I}$ and our BCOV action $I$ do agree.   

For $(1,0)$ BCOV theory, we can obtain stronger results.
\begin{theorem*}
On any Calabi-Yau $3$-fold $X$, the dynamically-generated classical interaction $J$ is equivalent to the classical interaction $I$ which is the restriction of the interaction from the full BCOV theory to the fields of $(1,0)$ BCOV theory. (The interaction $I$ is the one we wrote down explicitly earlier). 
\end{theorem*}
\subsection{Acknowledgements}
K.C. is grateful to Davide Gaiotto, Jaume Gomes and Edward Witten for helpful conversations.  

Part of this work was done during visits of S.L. to the Perimeter Institute, which also supports the research of K.C.  Research at Perimeter Institute is supported by the Government of Canada through Industry Canada and by the Province of Ontario through the Ministry of Economic Development and Innovation.

\section{Holomorphic Chern-Simons theory}
Let us now start on the technical material in the paper. We will begin with a discussion of how to describe $\gl_N$ gauge theories uniformly in $N$. 
\subsection{Invariant functions for $\gl_N$ gauge theory}
Let $V$ be a graded topological vector space, $V^\vee$ be its continuous dual, $\gl_N$ be the Lie algebra of $N\times N$ matrices.  We are interested in holomorphic Chern-Simons type theory whose space of fields is represented by $V\otimes \gl_N$. 

Let us first introduce some notation for functions on $V\otimes \gl_N$. Let
$$
\Oo(V \otimes \gl_N)=\what{\Sym}^{> 0} (V^\vee \otimes \gl_N^\vee)
$$ 
denote the graded algebra of formal power series on $V \otimes \gl_N$ \emph{modulo constants}. The component $\Sym^k (V^\vee \otimes \gl_N^\vee)$ will be called homogeneous of degree $k$. The group $GL_N$ acts on $V\otimes \gl_N$ via its adjoint representation on $\gl_N$. It induces a natural $GL_N$ action on $\Oo(V \otimes \gl_N)$ via duality. 

For $M < N$, we have an embedding $V \otimes \gl_M \into V \otimes \gl_N$ induced by
$$
   \gl_M\into \gl_N, \quad A\to \begin{pmatrix} A & 0\\ 0 & 0 \end{pmatrix}.
$$
Let 
$$i^{N}_M : \Oo (V \otimes gl_N) \to \Oo(V \otimes \gl_M)$$
be the restriction map associated to this embedding.
\begin{definition}
\label{definition_admissible_function}
Let 
$$\{f_N \in \Oo(V \otimes \gl_N) \mid N \in \Z_{>0}\}$$
be a collection consisting of a function on $V \otimes \gl_N$ for each $N$.  We say this collection of functions is \emph{admissible} if 
\begin{enumerate}
\item Equivariance: each $f_N$ is $GL_N$-invariant.
\item Compatibility: $
i^N_M f_N = f_M.  
$
\end{enumerate}
We let $\Oo_{adm}(V \otimes \gl_\infty)$ denote the admissible collection of functionals, which is a bi-graded commutative non-unital algebra (the cohomology degree inherited from $V$ and the homogeneous degree).  
\end{definition}

\begin{eg} Let $V=\C$. Then the collection of functionals $f_N\in \Oo(\gl_N)$ defined by
$$
   f_N(A)=\Tr(A^k), \quad A\in \gl_N
$$
is admissible and homogeneous of degree $k$. 
\end{eg}

\begin{eg} Let $V=\C[1]$, so $V$ is concentrated on degree $-1$. Let $\epsilon \in V$ denote the generator. Then the collection of functionals $f_N\in \Oo(\gl_N[1])$ defined by
$$
   f_N(\epsilon A)=\Tr(A^k), \quad A\in \gl_N
$$
is admissible and homogeneous of degree $k$, which vanishes if $k$ is even. 
\end{eg}

More generally, if $V$ is a graded vector space, let $\sigma$ denote the cyclic permutation
$$
\sigma: V^{\otimes k}\to V^{\otimes k}, \quad \sigma(v_1\otimes \cdots \otimes v_k)=(-1)^{|v_k|(|v_1|+\cdots + |v_{k-1}|)}v_k\otimes v_1\otimes\cdots \otimes v_{k-1},
$$
where $|v|$ denotes the degree of $v\in V$.  Let 
$$
\cyc_k(V)=\bracket{V^{\otimes k}}^{\sigma}
$$ 
denote the space of cyclically invariant elements of $V^{\otimes k}$, and let 
$$
    \cyc(V)=\bigoplus_{k\geq 1} \cyc_k(V)
$$
be the graded vector space of cyclically invariant tensors on $V$. There is a natural map
$$
     \Sym^k(V\otimes \gl_N)\to \cyc_k(V), \quad 
$$
by assigning 
$$
    (v_0\otimes A_0)\otimes \cdots \otimes (v_{k-1}\otimes A_{k-1})\to \sum_{s\in S_{k-1}} \pm v_0\otimes (v_{s(1)}) \cdots \otimes (v_{s(k-1)}) \Tr( A_0A_{s(1)}\cdots A_{s(k-1)}),
$$
where the sign $\pm$ is given by permuting the graded objects $v_i\otimes A_i$. Dually, we obtain a natural map
$$
\op{Cyc}_k(V)^\vee  \mapsto \Oo(V \otimes \gl_\infty) 
$$
by a sequence of functionals $\phi_N \in \Oo(V \otimes \gl_N)$ from $\phi \in \op{Cyc}_k(V)^\vee$ such that for a decomposable tensor $v \otimes A \in V \otimes \gl_N$, 
$$
\phi_N( v \otimes A ) = {1\over k}\phi (v^{\otimes k} )\op{Tr} ( A^k ) .
$$
This map extends by linearity to a map 
\begin{equation*}
\op{Cyc} (V)^\vee = \prod_{k\geq 1} \Cyc_k(V)^\vee \to \Oo( V \otimes \gl_\infty). \tag{$\dagger$}
\end{equation*}

Let $\Oo(\Cyc(V))=\what{\Sym}^{> 0} (\Cyc(V)^\vee)$ denote the completed non-unital symmetric algebra on the dual $\Cyc(V)^\vee$, viewed as formal functions on $\Cyc(V)$ modulo constants. 
\begin{lemma}
\label{lemma_cyclic_admissible}
The map ($\dagger$) extends to an isomorphism of bi-graded algebras
$$
\Oo (\Cyc(V)) \to \Oo_{adm} (V \otimes \gl_\infty).
$$
\end{lemma}
\begin{proof} By classical invariant theory for $GL_N$, for $N$ sufficient large ($N\geq k$), 
\begin{align*}
   \Sym^k(V^\vee \otimes \gl_N^\vee)_{GL_N}&\simeq  \bracket{((V^\vee)^{\otimes k}\otimes (\gl_N^\vee)^{\otimes k})_{S_k}}_{GL_N}
   \simeq   \bracket{((V^\vee)^{\otimes k}\otimes (\gl_N^\vee)^{\otimes k})_{GL_N}}_{S_k}\\
   &{\simeq} \bracket{(V^{\vee})^{\otimes k} \otimes \C[S_k]}_{S_k} .
\end{align*}
It follows that 
$$
 \Oo_{adm} (V \otimes \gl_\infty)\simeq \prod_{k\geq 1} \bracket{(V^{\vee})^{\otimes k} \otimes \C[S_k]}_{S_k}.
$$
Let $\sigma_k\subset S_k$ be the conjugacy class of the cycle $(12\cdots k)$. Since any permutation is a product of cycles, the natural map
$$
   \what{\Sym}^{>0}\bracket{\prod_{k\geq 1} \bracket{(V^{\vee})^{\otimes k} \otimes \C[\sigma_k]}_{S_k} }\to  \prod_{k\geq 1} \bracket{(V^{\vee})^{\otimes k} \otimes \C[S_k]}_{S_k}
$$
is an isomorphism. The lemma follows from the observation  
$$
\cyc_k^\vee\simeq \bracket{(V^{\vee})^{\otimes k} \otimes \C[\sigma_k]}_{S_k}
$$
and the direct check that the above identifications lead to ($\dagger$) . 
\end{proof}

\begin{definition}\label{definition-weight}
An admissible function $f=\{f_N\}$ is said to have weight $k$ is $f\in \Sym^k(\cyc(V)^\vee)$ under the isomorphism in Lemma  \ref{lemma_cyclic_admissible}.
\end{definition}

\subsection{Invariant functions and the cyclic cochain complex}
Now let us consider the above situation in the case that $V = A[1]$ where $A$ is a differential graded algebra.  Then, the graded space $\Oo(A\otimes \gl_N[1])$ has a natural differential, which is the Chevalley-Eilenberg differential $\d_{CE}$ for the dg Lie algebra $A \otimes \gl_N$.  

If $\{f_N\} \in \Oo(A \otimes \gl_N[1])$ is a collection of admissible functions, then so is $\{\d_{CE} f_N\}$.  It follows that the space $\Oo_{adm} (A\otimes \gl_\infty[1])$ has a differential for which the map 
$$\Oo_{adm}(A \otimes \gl_\infty[1]) \to \Oo(A \otimes \gl_N[1]) = C^\ast(A \otimes \gl_N)$$
is a cochain map.  
\begin{theorem*}[Tsygan \cite{Tsy83}, Loday-Quillen \cite{LodQui84}]
There is an isomorphism of cochain complexes
$$
\Oo_{adm}(A \otimes \gl_\infty[1]) \iso \widehat{\Sym}^{>0} \left( CC^\ast(A)[1] \right).
$$
\end{theorem*}
Here $CC^\ast(A)$ is the cyclic cochain complex of $A$. In the notation above, $CC^\ast(A)$ is $\op{Cyc}(A)^\vee [-1]$ equipped with a Hochschild-type differential. 

\subsection{A supergroup generalization}
For our application, we need a variant of the above construction for the super Lie algebra $\gl(N \mid N)$, which is the graded Lie algebra of endomorphisms of the graded vector space $\C^{N\mid N}=\C^N \oplus \C^{N}[1]$.  

Given a graded vector space $V$, the admissible functions on $V\otimes \gl(N \mid N)$ is defined the same as before by a sequence of functions $f_N \in \Oo(V \otimes \gl(N \mid N))$ such that 
\begin{enumerate} 
 \item Each $f_N$ is $\gl(N \mid N)$-invariant.
\item There is a natural inclusion $\gl(N \mid N) \into \gl(N + k \mid N + k)$ where $\gl(N \mid N)$ consists of those matrices acting on $\C^{N \mid N} \oplus \C^{k \mid k}$ which preserve $\C^{N \mid N}$ and act by zero on $\C^{k \mid k}$.  We let $i^{N+k}_N$ denote the corresponding restriction map $\Oo(V \otimes \gl(N + k \mid N + k)) \to \Oo(V \otimes \gl(N \mid N)$.  We require that $i^{N+k}_N f_{N+k} = f_N$.  
\end{enumerate} 
We let $\Oo_{adm}(V \otimes \gl(\infty \mid \infty))$ denote the space of admissible sequences of functions.
\begin{lemma}\label{lem-super-invariants}
Exactly as in the $\gl(N)$ case, there is a natural isomorphism
$$
\Oo(\cyc(V)) \iso \Oo_{adm}(V \otimes \gl(\infty \mid \infty)).
$$
In particular, we find a natural isomorphism 
$$
\Oo_{adm}(V \otimes \gl(\infty))\iso \Oo_{adm}(V \otimes \gl(\infty \mid \infty)).
$$
\end{lemma}
\begin{proof} 
 
We will define a map 
$$
\Oo(\cyc(W)) \to \Oo_{adm}(W \otimes \gl(\infty \mid \infty)) 
$$ 
and the image of this will be the admissible functionals. 

The map is very simple. If $e_1,\dots,e_n : W \to \C$ then we will construct a functional
$$
f_{e_1 \otimes \dots \otimes e_n} : \left( W \otimes \gl(N \mid N)\right)^{\otimes n} \to \C
$$
by saying that
$$
f_{e_1 \otimes \dots \otimes e_n} (w_1 \otimes A_1 \otimes \dots \otimes w_n \otimes A_n) = \prod e_i(w_i) \op{Tr}(A_1 \dots A_n).
$$
Here we take the trace in the defining representation of $\gl(N \mid N)$.  Projecting onto the symmetric tensors on the dual of $W \otimes \gl(N \mid N)$ makes $f_{e_1 \otimes \dots \otimes e_n}$ into an element of $\left( \Sym^n W \otimes \gl(N \mid N)\right)^\vee$. 
 
Since this map does not change (except by the appropriate sign) if we cyclically permute the $e_i$, we thus get a map
$$
\prod_n \op{cyc}_n(W)^\vee \to \Oo(W \otimes \gl(N \mid N))
$$ 
and so a map of commutative algebras
$$
\Oo(\cyc(W)) \to \Oo(W \otimes \gl(N \mid N)).
$$
The image is certainly admissible, so we get the desired map from $\Oo(\cyc(W))$ to $\Oo_{adm}(W \otimes \gl(\infty \mid \infty))$. 

Next we need to check that this map is an isomorphism. The map is clearly injective, so it suffices to check that it is surjective onto the space of $\gl(N \mid N)$-invariant functions on $W \otimes \gl(N \mid N)$.  But, as in the $ \gl(N)$ case, this follows immediately from invariant theory for the super group \cite{Ser01}.  The statement we need (proved in \cite{Ser01}) is the following.  Let $V = \C^{N \mid N}$ denote the defining representation of $\gl(N \mid N)$. Then, the space of $\gl(N \mid N)$ invariants in a tensor power $V^{\otimes p} \otimes (V^{\otimes q})^\ast$ is non-zero unless $p = q$. If $p = q$, then the invariants have the following description.  Let $c \in V \otimes V^\ast$ denote the matrix for the identity map from $V$ to $V$.  For each permutation $\sigma \in S_p$ we have an element
$$
c_{\sigma} = \otimes_{i = 1}^{p} c_{i,\sigma(i)}
$$
where $c_{i,\sigma(i)}$ refers to $c$ placed on the $i$'th tensor factor of $V^{\otimes p}$ and the $\sigma(i)$ tensor factor of $(V^\ast)^{\otimes p}$.  Then, the elements $c_{\sigma}$ span the space of invariants.   (In the $\gl(N)$ case invariant tensors have exactly the same description).

The proof from the $\gl(N)$ case now applies.  
\end{proof}

\subsection{Holomorphic Chern-Simons theory} Let $X$ be a compact Calabi-Yau manifold of dimension $d$ where $d$ is odd, $\Omega_X$ be a holomorphic volume form on $X$. The space of fields for the $\gl_N$ holomorphic Chern-Simons theory on $X$ is given by
$$
  \A_N:=\Omega^{0,*}(X)[1]\otimes \gl_N,
$$
where $\Omega^{0,*}(X)$ is the smooth $(0,*)$-forms on $X$ and $[1]$ is the shifting operator such that $\Omega^{0,1}(X)$ is of degree $0$. 
$\A_N[-1]$ is a DGLA with differential $\dbar$ and Lie bracket induced from that of $\gl_N$. Moreover, $\A_N$ is equipped with a $(2-d)$-symplectic pairing 
$$
   \omega(\alpha, \beta)=\int_X \Tr(\alpha \beta)\wedge \Omega_X, \quad \alpha, \beta \in \A_N.
$$
The DGLA structure reflects the gauge transformation, which is also encoded into the following classical action functional 
$$
   HCS_N(A)={1\over 2}\omega(\dbar A, A)+\CS_N(A), \quad A\in \A_N
$$
where 
$$
  \CS_N(A)={1\over 3}\int_X \Tr(A^3)\wedge \Omega_X. 
$$
The cohomological degree of the holomorphic Chern-Simons functional is
$$
\deg HCS_N=3-d. 
$$ 

\begin{deflem} The collections $\{HCS_N\}$ and $\{\CS_N\}$ are admissible, i.e., belonging to $ \Oo_{adm} (\Omega^{0,*}(X)[1] \otimes \gl_\infty)$. We will denote the corresponding admissible functionals by $HCS_{\infty}$ and $\CS_\infty$. 
\end{deflem}

Under the isomorphism in Lemma \ref{lemma_cyclic_admissible}, we have $
     \CS_\infty \in \cyc_3(\Omega^{0,*}(X)[1])^\vee  
$ corresponding to the cyclically invariant map
$$
   A_1 \otimes A_2\otimes A_3 \to (-1)^{|A_2|} \int_X (A_1 A_2 A_3)\wedge \Omega_X, \quad A_i \in \Omega_X^{0,*}[1]. 
$$
where $|A|=i-1, A\in \Omega^{0,i}[1]$, denotes the cohomology degree. 

\begin{definition}\label{def-functional} Let 
$$
\Oo(\A_N)=\what{\Sym}^{>0} (\A_N^\vee)
$$ 
denote the space of functionals on $\A_N$ modulo constants. $\Ool(\A_N)$ will denote the subspace of $\Oo(\A_N)$ given by local functionals. $\Oo_{adm}(\A_\infty)$ will denote the admissible collections.  A collection $\{f_N\in \Oo(\A_N)\}$ is called local if each $f_N\in \Ool(\A_N)$. The space of admissible local functionals will be denoted by $\Oo_{adm, loc}(\A_\infty)$. 
\end{definition}

\begin{eg} $HCS_\infty, \CS_\infty \in \Oo_{adm, loc}(\A_\infty)$. 
\end{eg}

The $(2-d)$-symplectic structure $\omega$ defines  a Poisson bracket of degree $d-2$ on local functionals:
$$
  \{-,-\}^{O}: \Ool(\A_N)\otimes \Ool(\A_N) \to \Ool(\A_N) 
$$
where the superscript $O$ refers to open string. 

\begin{lemma}\label{bracket-admissible-local} The Poisson bracket extends to admissible local functionals 
$$
  \{-,-\}^{O}: \Oo_{adm,loc}(\A_\infty)\otimes \Oo_{adm,loc}(\A_\infty) \to \Oo_{adm,loc}(\A_\infty). 
$$
\end{lemma}
\begin{proof} By Lemma \ref{lemma_cyclic_admissible}, the combinatorial part on $\gl_N$ of every admissible functionals consists of traces only. The Poisson bracket has kernel given by 
$$
\delta_0 \sum_{i,j=1}^N e_{ij}\otimes e_{ji}
$$
where $\delta_0$ is the delta function distribution representing the identity operator on $\Omega^{0,*}(X)$, and $e_{ij}$ is the fundamental matrix with only nonzero entry $1$ at the $i$th row and $j$th column.  It follows that for any two collections of admissible functionals $\{f_N\}, \{f^\prime_N\}$, 
$$
   i^N_M(\{f_N, f^\prime_N\}^O)=\{i^N_M(f_N), i^N_M(f^\prime_N)\}^O. 
$$
Moreover, $\sum_{i,j=1}^N e_{ij}\otimes e_{ji}$ is obviously $GL_N$-invariant. Therefore the collection of functionals $\{f_N, f^\prime_N\}^O$ is admissible. 
\end{proof}

\begin{lemma}\label{HCS-CME} The holomorphic Chern-Simons functional satisfies the following classical master equation
$$
  \{HCS_\infty, HCS_\infty\}^O=0,
$$
or equivalently
$$
   \dbar \CS_\infty+{1\over 2}\{\CS_\infty, \CS_\infty\}^O=0. 
$$
Here the bracket $\{-,-\}$ is defined by Lemma \ref{bracket-admissible-local}.
\end{lemma}
\begin{proof} The functionals $\CS_N$ defines $
      \Sym^3(\A_N)\to \C
$ in terms of its components by
$$
   A_1\otimes A_2\otimes A_3\to \int_X \Tr(A_1[A_2, A_3])\wedge \Omega_X. 
$$
It follows that $\{HCS_\infty,-\}=\dbar +\{\CS_\infty\}$ represents the Chevalley-Eilenberg differential of the DGLA $\A_N$, which squares zero implying the Lemma. 

\end{proof}

Therefore $\dbar+\{\CS_\infty,-\}^O$ defines a differential on $\Oo_{adm,loc}(\A_\infty)$, generating the infinitesimal gauge transformation in the \BV formalism.

We have a graded vector space $\op{Cyc}(\Omega^{0,\ast}(X)[1])^\vee$ of collections of cyclically-invariant linear functions on tensor powers of $\Omega^{0,\ast}(X)[1]$.  We have seen that we can identify $\op{Cyc}(\Omega^{0,\ast}(X)[1])^\vee$ with the weight $1$ admissible functions on $\Omega^{0,\ast}(X)[1] \otimes \gl_\infty$. Let $\op{Cyc}^\ast_{loc}(\Omega^{0,\ast}(X)[1])$ refer to the subcomplex which corresponds to weight $1$ elements of $\Oo_{adm,loc}(\Omega^{0,\ast}(X)[1]\otimes \gl_\infty)$. Concretely, an element of $\op{Cyc}^\ast_{loc}(\Omega^{0,\ast}(X)[1])$ is a collection of cyclically-invariant linear functionals on tensor powers of $\Omega^{0,\ast}(X)$, where each such functional can be written as a finite sum of functionals of the form
$$
\alpha_1 \otimes \dots \otimes \alpha_n \mapsto \int_X \d \op{Vol} D_1 \alpha_1 \dots D_n \alpha_n
$$
where each $D_i$ is a differential operator from $\Omega^{0,\ast}(X)$ to $\cinfty(X)$. 
\begin{lemma}
There is an isomorphism of cochain complexes 
$$\Oo_{adm}(\Omega^{0,\ast}(X)[1] \otimes \gl_\infty) \iso \what{\Sym}^{>0} \op{Cyc}(\Omega^{0,\ast}(X)[1])^\vee$$
where on the left hand side the complex has the differential $\dbar + \{I_{CS},-\}$ and on the right hand side the differential on $\op{Cyc}\left( \Omega^{0,\ast}(X)[1]\right)^\vee$ makes it a shift by one of the cyclic cochain complex of $\Omega^{0,\ast}(X)$.  
\end{lemma}
\begin{proof}
This follows from the Tsygan-Loday-Quillen theorem \cite{Tsy83,LodQui84} and the fact that the differential $\dbar + \{I_{hCS},-\}$ on $\Oo(\Omega^{0,\ast}(X) \otimes \gl_N[1])$ is the Chevalley-Eilenberg differential associated to the dg Lie algebra $\Omega^{0,\ast}(X) \otimes \gl_N$. 
\end{proof}
It follows from this that the differential on $\Cyc^\ast_{loc}(\Omega^{0,\ast}(X)[1])$ given by $\dbar + \{I_{hCS},-\}$ is a local  version of the cyclic cochain differential. 

\section{Coupling with Kodaira-Spencer theory of gravity}
\subsection{Coupled system} We are interested in compatible quantizations of $\gl_N$ holomorphic Chern-Simons theory on $X$ coupled with Kodaira-Spencer gauge theory (BCOV theory).  We will write down an explicit formula for a classical action which couples the two systems.

Recall that the field content of BCOV theory is given by complex
$$
   (\PV(X)[[t]][2], Q)
$$
where 
$$
\PV(X)=\bigoplus_{i,j}\PV^{i,j}(X), \quad \PV^{i,j}(X)=\Omega^{0,j}(X, \wedge^i T_X)
$$
is the space of smooth polyvector fields on $X$. $\PV(X)$ is a bi-graded algebra equipped with two differentials 
$$
  \dbar: \PV^{i,j}(X)\to \PV^{i,j+1}(X), \quad \pa: \PV^{i,j}(X)\to \PV^{i-1,j}(X)
$$ 
where $\pa$ is the divergence operator associated with the holomorphic volume form $\Omega_X$. The differential $Q$ is given by
$$
   Q=\dbar+t\pa,
$$
where $t$ represents the gravitational descendant. Our convention for grading is that elements of $\PV(X)[[t]][2]$ in $\PV^{i,j}(X)t^k$ has degree $i+j+2k-2$. 

\begin{definition} The space of fields for the $\gl_N$ holomorphic Chern-Simons theory coupled to BCOV theory is the complex 
$$
    \E_N=\PV(X)[[t]][2]\oplus \bracket{\Omega^{0,*}(X)\otimes \gl_N[1]}
$$
with differential $Q_N=Q\oplus \dbar$. 
\end{definition}

The coupling between polyvector fields and holomorphic Chern-Simons theory in the classical theory is encoded by a natural map 
$$
\PV(X)[[t]][2] \to \op{Cyc}^\ast_{loc} (\Omega^{0,\ast}(X)[1]).
$$
The full tree-level coupling is encoded in a map like this which is a map of $L_\infty$ algebras. Willwacher-Calaque's \cite{WilCal08} cyclic formality theorem provides such a map. For our purposes, we need only to have a map of cochain complexes, so we only need the leading term of the $L_\infty$ map constructed in \cite{WilCal08}.

We will describe very explicitly such a cochain map.  The formula will involve certain integrals of differential forms on some auxiliary spaces. 

Let $C_m=\{\theta_i, 1\leq i\leq m| 0<\theta_1<\theta_2<\cdots <\theta_m<1\}$, which can be viewed as parametrizing $m+1$ points on the circle up to rotation. There exists a cyclic automorphism induced by cyclically permuting the points
$$
  \sigma: C_m\to C_m, \quad \sigma^{-1}(\theta_i)=\begin{cases} \theta_{i+1}-\theta_1 & \mbox{if}\ 1\leq i\leq m-1 \\
  1-\theta_1 & i=m \end{cases}
$$
We introduce a $1$-form on $C_m$ by
$$
  \eta=(\theta_2-\theta_1)d\theta_1+\cdots+ (\theta_m-\theta_{m-1})d\theta_{m-1}+(1-\theta_m)d\theta_m 
$$
and a $2$-form on $C_m$
$$
  \omega=-d\eta. 
$$
It is easy to see that under cyclic permutation
$$
  \sigma^*(\eta)=\eta-d\theta_1, \quad \sigma^*(\omega)=\omega. 
$$

Let $\Sigma_m=\{0,1,\cdots, m, \bullet\}$ be the set with $m+2$ elements, ordered by $0<1<\cdots<m<\bullet$. The additional element $\bullet$ refers to the origin of the disk as in \cite{WilCal08}. Given a subset $\Gamma=\{i_1, \cdots, i_k\}\subset \Sigma_m$ indexed with increasing order, we define a differential form on $C_m$ by
$$
     \omega_\Gamma=e_{i_1}\wedge \cdots \wedge e_{i_k}, \text{ where } \begin{cases} e_j=d\theta_j & 1\leq j\leq m\\
     e_j=0& j=0\\
     e_j=\eta & j=\bullet \end{cases}.
$$
The choice of $\Gamma$ also defines a map
\begin{align*} 
 \PV(X) \times \Omega^{0,\ast}(X, \gl_N)^{\otimes m+1} & \to \C\\ 
\mu \otimes A_0 \otimes \dots \otimes A_m &\mapsto \mu_\Gamma(A_0,\dots,A_m) \\
 \end{align*}
where $\mu_\Gamma$ is defined as follows.  Let $\Gamma = \{i_1,\dots,i_k\} \subset \Sigma_m$. We define an operator $\partial_{i,\Gamma}$ for $i = 0,\dots, m$ by $\partial_{i,\Gamma} = \partial$ if $i \in \Gamma$ and $\partial_{i,\Gamma} = \op{Id}$ if $i \not \in \Gamma$. Then,
$$ 
\mu_\Gamma(A_0, \cdots, A_m) = \begin{cases} \int_{X} \op{Tr}\left(  \partial_{0,\Gamma}A_0\wedge \dots \wedge \partial_{m,\Gamma} A_m\right) \wedge \left( \mu \lrcorner \Omega_X\right) & \text{ if } \bullet \not \in \Gamma \\
 \int_{X}\left( \op{Tr} \partial_{0,\Gamma}A_0\wedge \dots \wedge \partial_{m,\Gamma} A_m\right) \wedge \left( \partial \mu \lrcorner \Omega_X\right) & \text{ if } \bullet  \in \Gamma
\end{cases} 
$$
If $\mu \in \PV^{r,\ast}(X)$, then this expression is non-zero only when $r = k$.

\begin{lemma} 
Fixing $\mu\in \PV(X)$, the functional
$$
  \sum_{\Gamma\subset \Sigma_m}\pm \int_X \int_{C_m} \omega^k\wedge \omega_\Gamma \wedge \mu_\Gamma(A_0, A_1, \cdots, A_m)
$$
is cyclically invariant.
\end{lemma}

\begin{rmk} Here $\pm$ are some signs coming from permuting graded objects, which will not be used in the current paper. See \cite{WilCal08} for details. 

\end{rmk}

\begin{proof} Let $\sigma = (0,1,\dots,m)$ be the cyclic permutation of the set $\{0,1,\dots,m\}$.  We extend $\sigma$ to an isomorphism of sets $\Sigma_m \to \Sigma_m$ by declaring that $\sigma$ sends $\bullet$ to $\bullet$. Then, $\sigma$ acts on the subsets $\Gamma \subset \Sigma_m$.  Explicitly, if $\Gamma=\{i_1, \cdots, i_k\}$ doesn't contain $\bullet$, then 
$$
  \sigma(\Gamma)=\begin{cases} \{i_1+1,\cdots, i_k+1\} & i_k<m  \\
         (0, i_1+1, \cdots, i_{k-1}+1)& i_k=m \end{cases}
$$
If $\Gamma=\{\Gamma^\prime, \bullet\}$, then 
$$
  \sigma(\Gamma)=\{\sigma(\Gamma^\prime), \bullet\}. 
$$
The lemma amounts to showing that
$$
  \sum_{\Gamma\subset \Sigma_m}\pm \int_X \int_{C_m}\omega^k\wedge \omega_{\Gamma}\wedge \mu_{\sigma^{-1}(\Gamma)}=  \pm \sum_{\Gamma\subset \Sigma_m}\int_X \int_{C_m} \omega^k\wedge \omega_{\Gamma}\wedge \mu_{\Gamma}
$$
where we have suppressed the inputs $A_0, \cdots, A_m$ for simplicity. On the other hand
\begin{align*}
   \sum_{\Gamma\subset \Sigma_m}\pm \int_X \int_{C_m}\omega^k\wedge \omega_{\Gamma}\wedge \mu_{\sigma^{-1}(\Gamma)}
   = \sum_{\Gamma\subset \Sigma_m}\pm \int_X \int_{C_m}\omega^k\wedge \omega_{\sigma(\Gamma)}\wedge \mu_{\Gamma}
\end{align*}
while it is direct to check that 
\begin{align*}
 \omega_{\sigma(\Gamma)}=\sigma^* \omega_\Gamma+\sum_{j\in \sigma(\Gamma)}\pm d\theta_1\wedge\omega_{\sigma(\Gamma)\backslash\{j\}}.
\end{align*}
The first term gives 
$$
\sum_{\Gamma\subset \Sigma_m}\pm \int_X \int_{C_m}\sigma^*(\omega^k\wedge \omega_{\Gamma})\wedge \mu_{\Gamma}=
\sum_{\Gamma\subset \Sigma_m}\pm \int_X \int_{C_m}\omega^k\wedge \omega_{\Gamma}\wedge \mu_{\Gamma}.
$$
Therefore we need to show that 
$$
\sum_{\Gamma\subset \Sigma_m}\sum_{j\in \sigma(\Gamma)}\pm \int_X \int_{C_m}\omega^k\wedge d\theta_1\wedge\omega_{\sigma(\Gamma)\backslash\{j\}}\wedge \mu_{\Gamma}
=0.
$$
In fact, 
\begin{align*}
&\sum_{\Gamma\subset \Sigma_m}\sum_{j\in \sigma(\Gamma)}\pm \int_X \int_{C_m}\omega^k\wedge d\theta_1\wedge\omega_{\sigma(\Gamma)\backslash\{j\}}\wedge \mu_{\Gamma}\\
=&\sum_{0\notin \Gamma}\sum_{j\in \sigma(\Gamma)}\pm \int_X \int_{C_m}\omega^k\wedge d\theta_1\wedge\omega_{\sigma(\Gamma)\backslash\{j\}}\wedge \mu_{\Gamma}+\sum_{0\in \Gamma}\pm \int_X \int_{C_m}\omega^k\wedge d\theta_1\wedge\omega_{\sigma(\Gamma)\backslash\{1\}}\wedge \mu_{\Gamma}
\end{align*}
while integration by parts on $X$ implies that
\begin{align*}
  \sum_{0\in \Gamma}\pm \int_X \int_{C_m}\omega^k\wedge d\theta_1\wedge\omega_{\sigma(\Gamma)\backslash\{1\}}\wedge \mu_{\Gamma}&=-\sum_{0\notin \Gamma}\sum_{j\notin \Gamma}\pm \int_X \int_{C_m}\omega^k\wedge d\theta_1\wedge\omega_{\sigma(\Gamma)}\wedge \mu_{\Gamma+\{j\}}\\
  &=-\sum_{0\notin \Gamma}\sum_{j\in \sigma(\Gamma)}\pm \int_X \int_{C_m}\omega^k\wedge d\theta_1\wedge\omega_{\sigma(\Gamma)\backslash\{j\}}\wedge \mu_{\Gamma}.
\end{align*}
\end{proof}

\begin{definition}\label{defn-1-disk}
We define the disk coupling with one interior marked point by the local functional 
$$
  I^{1-disk}_N(t^k \mu, A)=\sum_{m\geq 0}\sum_{\Gamma\subset \Sigma_m}\pm \int_X \int _{C_m} \omega^k\wedge\omega_\Gamma\wedge  \mu_{\Gamma}(A,\cdots, A). 
$$
\end{definition}

If $\mu \in \PV^{r,*}(X)$, then the term in $I^{1-disk}_N$ involving the integral over $C_m$ vanishes unless $m=r+2k$. It follows that $ I^{1-disk}_N$ has cohomology degree
$$
  \deg  I^{1-disk}_N=3-d. 
$$

\begin{rmk} $C_m$ equivalently describes the configuration space on the disk with one interior point and $m+1$ boundary points. In \cite{WilCal08}, a coupled interaction with arbitrary number of interior and boundary points are described via a similar integration on the corresponding configuration space. In the current paper, we only need the leading term $I^{1-disk}_N$. The other disk interactions will be constructed from deformation theory. 

\end{rmk}

\subsection{Classical master equation}
Let $\Oo(\E_N), \Ool(\E_N), \Oo_{adm}(\E_\infty)$ and $\Oo_{adm, loc}(\E_\infty)$ be defined similarly as in Definition \ref{def-functional}. Recall that there exists a bracket $\{-,-\}^O$ on $\Ool(\A_N)$ which is well-defined on $\Oo_{adm,loc}(\A_N)$. It is easy to see that $\{-,-\}^O$ extends naturally to $\Ool(\E_N)$ and $\Oo_{adm,loc}(\E_\infty)$.  

The collection $\{I^{1-disk}_N\}$ defines an admissible function, which we will denote by $I^{1-disk}_\infty$. The differential $Q_N$ induces a differential on $\Oo(\E_N)$ which is well-defined on $\Oo_{adm}(\E_\infty)$ and we will denote it by $Q_\infty$. 

\begin{lemma}\label{lem-CME} The following classical master equation holds
\begin{align*}
  & Q_\infty \CS_\infty +{1\over 2}\{\CS_\infty, \CS_\infty\}^O=0\\
  & Q_\infty I^{1-disk}_\infty+\{\CS_\infty, I^{1-disk}_\infty\}^O=0. 
\end{align*}
\end{lemma}
\begin{proof} 
The first equation follows from Lemma \ref{HCS-CME}, since $Q_\infty \CS_\infty=\dbar \CS_\infty$. The second equation follows essentially from an easy part of \cite{WilCal08}. We sketch a proof here for completeness. It is easy to see that $\dbar I^{1-disk}_\infty=0$. Therefore it amounts to showing that
\begin{align*}
  &-\sum_{\Gamma\subset \Sigma_m}\pm \int_X \int_{C_m} \omega^{k+1}\wedge \omega_\Gamma \wedge (\pa\mu)_\Gamma(A_0, A_1, \cdots, A_m)\\
  =&  \sum_{\Gamma\subset \Sigma_{m-1}}\pm \int_X \int_{C_{m-1}} \omega^k\wedge \omega_\Gamma \wedge (\mu_\Gamma(A_0 A_1, \cdots, A_m)-\mu_{\Gamma}(A_0, A_1A_2, \cdots, A_m)+\\
  &\cdots \pm \mu_{\Gamma}(A_mA_0, A_1, \cdots, A_{m-1})). 
\end{align*}
The right hand side corresponds to the integration on the boundary of $C_m$
\begin{align*}
&\sum_{\Gamma\subset \Sigma_{m-1}}\pm \int_X \int_{C_{m-1}} \omega^k\wedge \omega_\Gamma \wedge (\mu_\Gamma(A_0 A_1, \cdots, A_m)-\mu_{\Gamma}(A_0, A_1A_2, \cdots, A_m)\\
&+\cdots \pm \mu_{\Gamma}(A_mA_0, A_1, \cdots, A_{m-1}))\\
=&\sum_{\Gamma\subset \Sigma_m}\pm \int_X \int_{\pa C_m}\omega^k \wedge \omega_\Gamma\wedge \mu_\Gamma(A_0, A_1, \cdots, A_m)\\
=&\sum_{\Gamma\subset \Sigma_m}\pm \int_X \int_{C_m} d(\omega^k\wedge \omega_\Gamma\wedge \mu_\Gamma(A_0, A_1, \cdots, A_m))\\
=&\sum_{\bullet \in \Gamma}\pm \int_X \int_{C_m} d(\omega^k\wedge \eta\wedge \omega_{\Gamma\backslash \{\bullet\}}\wedge (\pa \mu)_{\Gamma\backslash{\bullet}}(A_0, A_1, \cdots, A_m))\\
=&-\sum_{\Gamma\subset \Sigma_m}\pm \int_X \int_{C_m} \omega^{k+1}\wedge \omega_\Gamma \wedge (\pa\mu)_\Gamma(A_0, A_1, \cdots, A_m)
\end{align*}
as desired. Here we have used the fact that $ d(\omega^k\wedge \omega_\Gamma\wedge \mu_\Gamma(A_0, A_1, \cdots, A_m))=0$ if $\Gamma$ doesn't contain $\bullet$. 
\end{proof}

\begin{rmk}The classical master equation says that $ I^{1-disk}_\infty$ describes a first-order deformation of the holomorphic Chern-Simons functional $HCS_{\infty}$. As explained in the introduction, it gives in fact the universal first-order deformation. 

\end{rmk}

\subsection{An alternative approach to constructing the interaction $I^{1-disk}$.} 
We have presented the explicit formula for the interaction $I^{1-disk}$ which arises from a cochain map between $\PV(X)[[t]]$ and cyclic cochains of $\Omega^{0,\ast}(X)$.  One can, however, use a simple abstract argument to prove the existence of the desired $I^{1-disk}$.  In this section we will sketch this argument to give a conceptual understanding of $I^{1-disk}$.

By the Hochschild-Kostant-Rosenberg theorem \cite{Kon97} there is an isomorphism between $\PV(X)$ and the local Hochschild cochain complex $HC^\ast_{loc}(\Omega^{0,\ast}(X))$. This local Hochschild cochain complex is built from Hochschild cochains which are sequences of maps
$$
\Omega^{0,\ast}(X)^{\otimes n} \to \Omega^{0,\ast}(X)
$$
which are poly-differential operators. The integration map on $X$ allows us to identify local Hochschild cochains of $X$ with a local version of the linear dual of the Hochschild chains on $X$.  The latter consists of collections of functionals
$$
\Omega^{0,\ast}(X)^{\otimes n+1} \to \C
$$
which are local in the sense used before, and equipped with a differential which is dual to the Hochschild chain differential.  In this way, the Connes $B$-operator can be viewed as an operator on $HC^\ast_{loc}(\Omega^{0,\ast}(X))$.  

A simple check shows us that the HKR map
$$
\PV(X) \to HC^\ast_{loc}(\Omega^{0,\ast}(X))
$$
intertwines the operator $\partial$ with the Connes $B$-operator.  It follows that we get a HKR map
$$
\PV(X)[[t]] \to HC^\ast_{loc}(\Omega^{0,\ast}(X))[[t]]
$$
which intertwines the differential $\dbar + t \partial$ with the differential $\d + t B$, where $\d$ is the Hochschild differential and $B$ is the Connes operator.

The final step is to use the fact that there is a quasi-isomorphism $HC^\ast_{loc}(\Omega^{0,\ast}(X)) [[t]]$ and the local cyclic cochain complex of $\Omega^{0,\ast}(X)$.  This last fact is a local version of the standard theorem comparing two different models for cyclic cohomology: one based on a double complex, and one based on collections of cyclically-invariant multi-linear maps.   

\section{Quantum open-closed BCOV theory}
\label{quantum-BCOV}
In this section, $X$ will be a compact Calabi-Yau of dimension $d$ \emph{odd}.  We will define what it means to construct the quantum open-closed BCOV theory on $X$.  

\subsection{Regularization} Let us fix a K\"{a}hler metric on $X$, which gives rise to the adjoint $\dbar^*$ on both $\PV(X)$ and $\Omega^{0,*}(X)\otimes \gl_N[1]$. Associated to the corresponding Laplacian $\dbar \dbar^*+\dbar^*\dbar$, we have the heat kernels 
\begin{align*}
    K_L^{C}\in \PV(X)\otimes \PV(X)\\
    K_L^{O, N}\in (\Omega^{0,*}(X)[1]\otimes \gl_N)^{\otimes 2}
\end{align*}
where $L$ is the time parameter for the heat kernel. The superscript $C$ refers to ``closed string" sector while $O$ refers to ``open string" sector. They can be viewed as the regularization of the delta-function distribution which represents the ultraviolet singularity in general QFT. We will treat both as elements in $\Sym^2(\E_N)$. In terms of components, we have
$$
  K_L^C\in \bigoplus_{i,j} \PV^{i,j}(X)\otimes \PV^{d-i,d-j}(X), \quad  K_L^{O, N}\in \bigoplus_i (\Omega^{0,i}(X)\otimes \gl_N)\otimes (\Omega^{0,d-i}(X)\otimes \gl_N). 
$$
Then with respect to the cohomology degree on $\E_N$, we have
$$
  \deg (K_L^{C})=2d-4, \quad  \deg(K_L^{O, N})=d-2.
$$

\begin{definition} We define the effective propagators at scale $L$ by the smooth kernels
\begin{align*}
  P(\epsilon, L)^C=\int_\epsilon^L (\pa \dbar^*\otimes 1)K^C_t dt\\
  P(\epsilon, L)^{O,N}=\int_\epsilon^L (\dbar^*\otimes 1)K_t^{O,N}. 
\end{align*}
\end{definition}
Note that 
$$
\deg   P(\epsilon, L)^C=2d-6, \quad \deg P(\epsilon, L)^{O,N}=d-3. 
$$

In general, every symmetric element $\Phi$ of $\Sym^k(V)$ defines an order $k$ differential operator $\pa_{\Phi}$ on $\Oo(V)$. This is the unique order $k$ operator whose restriction to $\Sym^k(V^\vee)$ is given by pairing with $\Phi$. This is generalized to the setting for smooth kernels and functions being distributions. In particular, the effective propagator induces order two operators
$$
   \pa_{P(\epsilon, L)^C}, \pa_{P(\epsilon, L)^{O,N}}: \Oo(\E_N)\to \Oo(\E_N)
$$
of cohomology degree $2d-6, d-3$ respectively. 
\begin{definition} We define the effective BV operators at scale $L$ 
\begin{align*}
    \Delta_L^C=\pa_{(\pa\otimes 1)K_L^C}, \quad \Delta_L^{O, N}=\pa_{K_L^{O,N}}
\end{align*}
on $\Oo(\E_N)$ of cohomology degree $2d-5, d-2$ respectively. The corresponding effective BV operators are defined by
\begin{align*}
\{\Phi, \Psi\}_L^{C}&=\Delta_L^C(\Phi \Psi)-\Delta_L^C(\Phi)\Psi-(-1)^{|\Phi|}\Phi \Delta_L^C(\Psi)\\
\{\Phi, \Psi\}_L^{C,N}&=\Delta_L^{C,N}(\Phi \Psi)-\Delta_L^{C,N}(\Phi)\Psi-(-1)^{|\Phi|}\Phi \Delta_L^{C,N}(\Psi).
\end{align*}
\end{definition}

\begin{definition}
A collection of functionals $\{f_N| f_N\in \Oo(\E_N)\}$ is called admissible with $N$-degree $m$ if $\{f_N/N^m\}$ is an admissible collection. The space of such collections will be denoted by $\Oo^{(m)}_{adm}(\E_\infty)$.

A collection of such functionals is called weakly admissible if it is finite sum of functionals which are admissible of some $N$-degree.  We let 
$$\Oo^{(\ast)}_{adm}(\E_\infty) =  \bigoplus_{m \ge 0} \Oo^{(m)}_{adm}( \E_\infty)$$
denote the collection of weakly admissible functions. 
\end{definition}

Obviously our previous definition for $\Oo_{adm}(\E_\infty)$ refers to admissible functions of $N$-degree 0. There exists a natural identification 
$$
    \Oo_{adm}(\E_\infty)\iso  \Oo^{(m)}_{adm}(\E_\infty)
$$
which sends $\{f_N\}$ to $\{N^m f_N\}$ and morphisms 
$$
    \Oo^{(m)}_{adm}(\E_\infty) \otimes \Oo_{adm}^{(l)}(\E_\infty) \to \Oo^{(m+l)}_{adm}(\E_\infty). 
$$
In this way the space of weakly admissible functions forms an algebra. There is an algebra isomorphism
$$\Oo_{adm}^{(\ast)} \iso \Oo_{adm}(\E_\infty)[N].$$   
In fact this is an isomorphism of $\C[N]$-algebras, where $\C[N]$ acts on a weakly admissible collection of functions $\{f_N\}$ by multiplying each $f_N$ by a polynomial $P(N) \in \C[N]$. 

\begin{deflem}
Suppose that $\{f_N\}, \{g_N\}$ are weakly admissible collections of functions. Then so are   $\{\pa_{P(\epsilon, L)^C} f_N\}$, $\{\Delta_L^C f_N\}$, $\{\pa_{P(\epsilon, L)^{O,N}} f_N\}$,  $\{\Delta_L^{O,N} f_N\}$, $\{\{f_N,g_N\}^{C}_L\}$, $\{\{f_N,g_N\}^{O,N}_L\}$.  We will denote the corresponding operators on $\Oo^{(\ast)}_{adm}(\E_\infty)$ by  
$$
\pa_{P(\epsilon, L)^C}, \Delta_L^C, \pa_{P(\epsilon, L)^{O,\infty}}, \Delta_L^{O,\infty}, \{-,-\}_L^C, \{-,-\}^{O,\infty}_L.
$$
\end{deflem}

The operators $\pa_{P(\epsilon, L)^{O,\infty}}, \Delta_L^{O,\infty}, \{-,-\}^{O,\infty}_L$ may increase the N-degree by $1$ by ``closing a trace", while $\pa_{P(\epsilon, L)^C}, \Delta_L^C, \{-,-\}^{C}_L$ preserve the $N$-degree. 

\subsection{The renormalization group flow and the quantum master equation}

Now we introduce a formal parameter $\lambda$ of cohomology degree $3-d$ and consider  
$$\Oo^{(\ast)}_{adm}(\E_\infty)[[\lambda]]= \Oo_{adm}(\E_\infty)[N][[\lambda]]$$. 
\begin{definition} Let
$$
   \Oo^{(\ast)}_{adm}(\E_\infty)[[ \lambda]]_+\subset \Oo_{adm}^{(\ast)}(\E_\infty)[[\lambda]]
$$
be the space of functionals  of the form
$$
  I=\sum_{g\geq 0, h\geq 0}\lambda^{2g+h+f}I_{g,h;f}
$$
such that  $I_{g,h;f}$ has $N$-degree $f$ and weight $h$ (see Definition \ref{definition-weight}) and such that $I_{0,0;0}, I_{0,1;0}$ are at least cubic. 
\end{definition}
Under the isomorphism between $\Oo^{(\ast)}_{adm}(\E_\infty)$ and $\Oo_{adm}(\E_\infty)[N]$, a series $I = \sum \lambda^{2g+h+f} I_{g,h;f}$ becomes a series of the form $\sum \lambda^{2g+h+f} I_{g,h;f} N^f$.

Given $I$ as in the definition, we can Taylor expand each $I_{g,h;f}$ as a sum
$$
I_{g,h;f} = \sum I_{g,h,n;f}
$$
where $I_{g,h,n;f}$ is homogeneous of degree $n$ as a function of the closed fields.   

Heuristically, we should think of $I_{g,h,n;f}$ as corresponding to a surface of genus $g$ with $n$ interior marked points, $h$ boundary components with at least one marked point, and $f$ unmarked boundary components, where we have inputs from $\PV(X)[[z]][2]$ on interior marked points and inputs from $\Omega^{0,*}(X)\otimes \gl_N[1]$ on boundary marked points. 

\begin{deflem}
We have a well-defined \emph{renormalization group (RG) flow operator} 
$$
   W_\epsilon^L: \Oo^{(\ast)}_{adm}(\E_\infty)[[\lambda]]_+\to  \Oo_{adm}^{(\ast)}(\E_\infty)[[ \lambda]]_+
$$
defined by 
$$
W_\epsilon^L(I)=\lambda^2 \log \exp\bracket{\lambda^2\pa_{P(\epsilon, L)^C}+\lambda \pa_{P(\epsilon, L)^{O,\infty}}} \exp\bracket{I/\lambda^2}.
$$
\end{deflem}
\begin{proof} $W_\epsilon^L(I)$ is given by sums over connected Feynman graphs with two types of propagator $P(\epsilon, L)^C$ (closed string), $P(\epsilon, L)^{O, \infty}$ (open string) and vertices $I$ of two types of half edges from closed string and open string sector. We think about each term $I_{g,h;f}$ as a surface $\Sigma_{g,h;f}$of genus $g$, with $h+f$ boundaries, then each vertex contribute a factor $\lambda^{-\chi(\Sigma_{g,h;f})}$ where $\chi(\Sigma_{g,h;f})$ is the Euler characteristic. It is easy to see that a connected Feynman graph will produce a new surface $\Sigma_{g^\prime, h^\prime;f^\prime}$ of similar type such the Feynman graph integral produces a factor $\lambda^{-\chi(\Sigma_{g^\prime, h^\prime;f^\prime})}$. $W_\epsilon^L(I)$ being well-defined now follows from the fact that $\chi(\Sigma_{g^\prime, h^\prime;f^\prime})\leq 2$ and that only finite graphs contribute given a fixed number of field inputs. 
\end{proof}
\begin{remark}
Suppose we apply this RG flow operator $W_\eps^L(I)$ to, for example, the holomorphic Chern-Simons interaction $I^{CS}_{\infty}$.  The result can be described as a sum over trivalent ribbon graphs with external lines, where each vertex of the ribbon graph has the interaction $I^{CS}_{\infty}$ and each edge has the open-string propagator $P(\eps,L)^{O,\infty}$.  Each ribbon graph $\gamma$ is weighted by $\lambda^{1-\chi(\gamma)}$ and by $N$ to the number of faces of the ribbon graph which have no external lines.  
\end{remark}
\begin{definition}
As before, let $\E_N$ denote the space of fields at finite $N$.  Let $\Oo(\E_N)[[\lambda]]_+$ denote the series of the form $I = \sum I_{r,n_O,n_C} \lambda^{r}$ where $I_{r,n_O,n_C}$ is homogeneous of degree $n_O$ (respectively, $n_C$) as a function of the open-string fields (respectively, closed-string fields).  We require that $I_{0,n_O,n_C} = 0$ unless $n_O = 0$ and $n_C \ge 3$, and that $I_{1,n_O,n_C} = 0$ unless $n_O + 2 n_C \ge 3$.  

Then, the formula
$$
W_{\eps}^L(I) = \lambda^2 \log \exp\bracket{\lambda^2\pa_{P(\epsilon, L)^C}+\lambda \pa_{P(\epsilon, L)^{O,N}}} \exp\bracket{I/\lambda^2}.
$$
gives a well-defined renormalization group flow operator on $\Oo(\E_N)[[\lambda]]_+$. 
\end{definition}
If $I \in \Oo^{(\ast)}(\E_\infty)$, let $\{I^N \in \Oo(\E_N)\}$ denote the collection of functions making up the weakly admissible function $I$.   
\begin{lemma}
The map
\begin{align*} 
 \Oo^{(\ast)}(\E_\infty)[[\lambda]]_+ & \to \Oo(\E_N)[[\lambda]]_+ \\
\sum \lambda^{2g+h+f}I_{g,h;f} & \mapsto \sum \lambda^{2g+h+f} I_{g,h;f}^N 
\end{align*}
intertwines the renormalization group flow maps.  
\end{lemma}
\begin{proof}
By definition, the operators $\partial_{P(\eps,L)^C}$ and $\partial_{P(\eps,L)^{O,\infty}}$ act on a weakly admissible collection $\{I^N\}$  via the action of $\partial_{P(\eps,L)^C}$ and $\partial_{P(\eps,L)^{O,N}}$ on each $I^N$.  Further the map $\Oo^{(\ast)}(\E_\infty) \to \Oo(\E_N)$ is an algebra homomorphism. The result follows immediately.\end{proof}

\begin{definition} $I\in \Oo^{(\ast)}_{adm}(\E_\infty)[[ \lambda]]_+$ is said to satisfy \emph{quantum master equation} (QME) at scale $L$ if
$$
   Q_\infty I+\lambda^2 \Delta^C_L I+\lambda \Delta_L^{O, \infty}I+{1\over 2}\{I, I\}_L^C+{1\over 2 \lambda}\{I, I\}_L^{O, \infty}=0. 
$$
\end{definition}

It is also convenient to write the quantum master equation formally as
$$
   \bracket{Q_\infty+\lambda^2 \Delta^C_L+\lambda \Delta_L^O} e^{I/\lambda^2}=0. 
$$
\begin{deflem}
A function $I \in \Oo(\E_N)[[\lambda]]$ satisfies the finite-$N$ scale $L$ quantum master equation if 
$$
   Q_\infty I+\lambda^2 \Delta^C_L I+\lambda \Delta_L^{O, N}+{1\over 2}\{I, I\}_L^C+{1\over 2 \lambda}\{I, I\}_L^{O, N}=0. 
$$

The restriction map
$$
\Oo^{(\ast)}(\E_\infty)[[\lambda]]_+ \to \Oo(\E_N)[[\lambda]]_+
$$
takes a solution to the scale $L$ quantum master equation to a solution to the scale $L$ finite-$N$ quantum master equation. 

Further, a functional $I \in \Oo^{(\ast)}(\E_\infty)[[\lambda]]_+$ satisfies the scale $L$ quantum master equation if and only if all of its finite $N$ restrictions $I^N$ satisfy the scale $L$ quantum master equation.  
\end{deflem} 
\begin{proof}
The fact that the restriction map is compatible with the QME is immediate from the fact that we defined the scale $L$ BV Laplacians and odd Poisson brackets at infinite $N$ in terms of their finite $N$ counterparts.  The fact an infinity $N$ functional $I$ satisfies the QME if and only if each of its finite $N$ restrictions $I^N$ does follows from the fact that the product of all the restriction maps
$$
\Oo^{(\ast)}(\E_\infty) \to \prod_N \Oo(\E_N)
$$  
is injective.
\end{proof} 
\begin{lemma} RG flow and QME are compatible in the following sense: if $I$ satisfies QME at scale $\epsilon$, then $W_\epsilon^L(I)$ satisfies QME at scale $L$. 
\end{lemma}
\begin{proof} This is proved similarly as in \cite{Cos11}.  The results of \cite{Cos11} are directly applicable to the finite-$N$ version of this statement. The fact that the RG flow and QME are compatible with the restriction maps from $\Oo^{(\ast)}_{adm}(\E_\infty) \to \Oo(\E_N)$, and the fact that to verify the quantum master equation for a functional $I \in \Oo^{(\ast)}_{adm}(\E_\infty)[[\lambda]]_+$ it suffices to verify it for all of its restrictions $I^N$, implies the result for infinite $N$. 
\end{proof}

\subsection{The definition of quantum theory} We now have enough notations to define the notion of quantum open-closed BCOV theory. 
\begin{definition}
\label{quantum_BCOV_definition}A quantum open-closed BCOV theory on a Calabi-Yau $X$ of odd dimension $d$ is a collection of functionals
$$
   I[L]\in \Oo^{(\ast)}_{adm}(\E_\infty)[[ \lambda]]_+
$$
one for each $L>0$, satisfying the following axioms.
\begin{enumerate}
\item The renormalization group flow: 
$$
   I[L]=W_\epsilon^L(I[\epsilon]), \quad \forall L, \epsilon>0.
$$
\item The quantum master equation
$$
   Q_\infty I[L]+\lambda^2 \Delta^C_L I[L]+\lambda \Delta_L^{O, \infty}I[L]+{1\over 2}\{I[L], I[L]\}_L^C+{1\over 2 \lambda}\{I[L], I[L]\}_L^{O, \infty}=0. 
   $$
\item
   The locality axiom: $I[L]$ has a small $L$ asymptotic expansion in terms of admissible local functionals,
$$
I[L] \simeq \sum \Phi_i(L) J_i
$$ 
where $\Phi_i(L)$ are some functions of $L$, and $J_i \in \Oo_{adm,loc} ( \E_\infty)[[\lambda, N]]$. 
\item The degree axiom. With respect to the expansion 
$$
  I[L]=\sum_{g,h,n,f\geq 0}I_{g,h,n;f}[L]\lambda^{2g+h+f}N^f, 
  $$
where $I_{g,h,n;f}\in \Sym^n(\PV(X)[[t]][2])^\vee \otimes \Sym^h(\cyc(\Omega^{0,*}(X)[1])),
$
$I_{g,h,n;f}$ has cohomology degree $(d-3)(2g-2+h+f)$. Equivalently 
$
\deg(\lambda^{-2}I[L])=0.
$ 
\item The classical limit. We require that
$$
   \lim_{L\to 0} I_{0,1,0}[L]=\CS_\infty, \quad \lim_{L\to 0} I_{0,1,1}=I^{1-disk}_\infty. 
$$
\end{enumerate}
\end{definition}
\begin{definition}
A quantum theory coupling the $\gl_N$ holomorphic Chern-Simons theory with BCOV theory is specified by a collection of functionals $I^N[L] \in \Oo(\E_N)[[\lambda]]_+$ satisfying the same properties as listed above, except that the classical limit axiom is changed to the following:
\begin{enumerate} 
 \item If we expand $I^N[L] = \sum_{r \ge 0} I^N_r[L] \lambda^r$, then $I^N_0[L]$ is a functional only of the closed-string fields and is at least cubic.  
\item The limit
$$
\lim_{L \to 0} I_1[L]
$$  
exists and is equal to $\CS_N + I^{1-disk}_N$ when the homogenous degree of closed inputs $\leq 1$. 
\end{enumerate}
\end{definition}
\begin{lemma}
If $I[L] \in \Oo^{(\ast)} (\E_\infty)[[\lambda]]_+$ is a collection of functionals defining a quantum open-closed BCOV theory, then the restrictions $I[L]^{N} \in \Oo(\E_N)[[\lambda]]_+$ define a quantum theory coupling $\gl_N$ holomorphic Chern-Simons with BCOV theory.  

Conversely, suppose one has a collection $I[L]^{N} \in \Oo(\E_N)[[\lambda]]_+$ of functionals defining a coupled BCOV theory and $\gl_N$ holomorphic Chern-Simons. Suppose that the collection $\{I[L]^{N}\}$ defines a weakly admissible collection of functionals.  Let us expand 
$$
I[L]^{N} = \sum_{r \ge 0} I_{r,h;f}^{N}[L] \lambda^{r}
$$
where the collections of admissible functionals $I_{r,h;f}[L]^{N}$ are of $N$-degree $f$ and weight $h$.   Suppose further that $I_{r,h;f}[L]^{N} = 0$ unless $r - h - f\geq 0$ and is even.  Then, the admissible collection $I[L]^{N}$ defines a quantum open-closed BCOV theory. 
\end{lemma}
\begin{proof}
The statement that the restriction to finite $N$ of a quantum open-closed BCOV theory gives a quantization of the coupled $\gl_N$ holomorphic Chern-Simons and  BCOV theory is immediate.  For the converse, suppose we have functionals $I[L]^{N}$ as stated in the lemma. Then, they define an admissible function $I[L] \in \Oo^{(\ast)}_{adm}(\E_\infty)[[\lambda]]$. The RG flow, quantum master equation, classical limit axiom, and locality axiom are automatically satisfied, as they can be tested at finite $N$.   Then only thing that needs to be verified is that $I[L]$ has an expansion of the form $\sum I_{g,h;f} [L] \lambda^{2g+h+f}$ where $I_{g,h;f}$ is of weight $h$ and $m$-degree $f$.  This, however, follows from the final criterion stated in the lemma. 
\end{proof}
Thus, we see that constructing the open-closed theory is essentially the same as constructing the coupling of the closed string sector to the $\gl_N$ gauge theory for all $N$, in a compatible way. 

For our purpose, we will also need a version of quantum open-closed theory without free boundaries. This can be obtained via a mod-N reduction from the above definition. 

\begin{definition}\label{definition-modulo-N}
A quantum open-closed BCOV theory on a Calabi-Yau $X$ of odd dimension $d$ without free boundaries is a collection of functionals $I[L]\in \Oo_{adm}(\E_\infty)[[\lambda]]_+$ for each $L>0$, satisfying all conditions (1)(2)(3)(4)(5) as above modulo $N$, i.e., we disregard all components whose $N$-degree $>0$.
\end{definition}

We leave it to the reader to check that this mod-N version is well-defined. This definition amounts to ignoring all the terms containing free boundaries without open string inputs that appears in the renormalization group flow and quantum master equation.  We will see shortly that a quantum open-closed BCOV theory modulo $N$ is the same data as a quantization of the open-closed theory where on the open sector we use the super Lie algebras $\gl(N \mid N)$ instead of $\gl(N)$.

\subsection{Supergroup holomorphic Chern-Simons}\label{Supergroup version}
Everything we have discussed works if we consider holomorphic Chern-Simons with gauge Lie algebra $\gl(N\mid N)$. The space of fields is $\Omega^{0,\ast}(X)[1]\otimes \gl(N \mid N)$.  We let $\E_{N \mid N}$ denote the corresponding space of fields for the open-closed theory
$$
\E_{N \mid N} =\PV(X)[[t]][2] \oplus \Omega^{0,\ast}(X) \otimes \gl(N \mid N)[1].
$$ 
The formulae we wrote for the Chern-Simons action and the one-disk interaction $I^{1-disk}_N$ apply in this situation too, except that as well as taking the trace in $\gl_N$ we take the trace in the defining representation of $\gl(N \mid N)$. 

Admissible functions $\Oo_{adm}(\E_{\infty \mid \infty})$ on $\E_{N\mid N}$ are defined similarly, from which we can define the RG flow and quantum master equation for open-closed BCOV theory in the supergroup case. Before we do so, we need a small detour about differential operators on the space $\Oo(W \otimes \gl(\infty \mid \infty))$. 

Let $W$ be any graded vector space.  As before, let $\Oo_{adm}(W \otimes \gl(\infty))$ and $\Oo_{adm}(W \otimes \gl(\infty \mid \infty))$ denote sequences of admissible functions on $W \otimes \gl(N)$ and on $W \otimes \gl(N \mid N)$ respectively.  There is a natural isomorphism of algebras 
$$\Oo_{adm}(W \otimes \gl(\infty)) \iso \Oo_{adm}(W \otimes \gl(\infty \mid \infty))$$
since both can be identified with the symmetric algebra on the dual of $\op{Cyc}(W)$.

Let $\Oo^{(\ast)}_{adm}(W \otimes \gl(\infty))$ denote the space of weakly admissible sequences of functions on $W \otimes \gl(N)$. There is an isomorphism 
$$
\Oo^{(\ast)}_{adm}(W \otimes \gl(\infty)) \iso \Oo_{adm}(W \otimes \gl(\infty)) [N]
$$
and so a map
$$
\Oo^{(\ast)}_{adm}(W \otimes \gl(\infty)) \to \Oo(W \otimes \gl(\infty \mid \infty)) 
$$
obtained by setting $N = 0$.

For example, in the case that $W = \C$, the algebra of admissible functions is generated by the functions $A \mapsto \Tr A^n$ on $\gl(N)$ or on $\gl(N \mid N)$, for $n  > 0$. The algebra of weakly admissible functions includes the case $n = 0$; the function $\op{Tr} \op{Id}$ takes value $N$ on $\gl(N)$ and zero on $\gl(N \mid N)$.  In this case, the map from weakly admissible functions on $\gl(N)$ to admissible functions on $\gl(N \mid N)$ sends the function $A \mapsto \op{Tr} A^n$ on $\gl(N)$ to the same function on $\gl(N \mid N)$. 

\begin{lemma}
\label{lemma_supergroup_operator}
Let $W$ be a graded vector space and let $P \in W \otimes W$ be a symmetric element.  Define operators $\delta_{P \otimes c}^{N\mid M}$ on $\Oo(W \otimes \gl(N \mid M))$ associated to the elements $P \otimes c \in W^{\otimes 2} \otimes \gl(N \mid M)^{\otimes 2}$ where $c \in \gl(N \mid M)^{\otimes 2}$ is the Casimir.  (The Casimir is the inverse to the invariant pairing on $\gl(N \mid M)$ coming from $\op{Tr}(X Y)$ in the defining representation).  In the case that $M = 0$ we will refer to this operator as $\delta_{P \otimes c}^{N}$. 

Then, $\delta_{P \otimes c}^{N}$ preserves the class of weakly admissible functions on $W \otimes \gl(N)$, and so gives an operator $\delta_{P \otimes c}^{\infty}$ on $\Oo_{adm}^{(\ast)}(W \otimes \gl(\infty))$.  Further, $\delta_{P \otimes c}^{N \mid N}$ preserves the class of admissible functions on $W \otimes \gl(N \mid N)$, and so gives an operator $\delta_{P \otimes c}^{\infty \mid \infty}$ on $\Oo_{adm}(W \otimes \gl(\infty \mid \infty))$.

The natural map
$$
\Oo^{\ast}_{adm}(W \otimes \gl(\infty)) \to \Oo_{adm}(W \otimes \gl(\infty \mid \infty))
$$
intertwines the operators $\delta_{P \otimes c}^{\infty}$ and $\delta_{P \otimes c}^{\infty \mid \infty}$. 
\end{lemma}
\begin{proof}
This is a simple combinatorial check. The upshot that $\delta_{P \otimes c}^{\infty \mid \infty}$ being well-defined on $\Oo_{adm}(W \otimes \gl(\infty \mid \infty))$ as opposed to the $\gl(N)$-case follows from the fact that any increase of $N$-degree comes with a factor $\op{Tr} (1)$ which is zero for $\gl(N\mid N)$.
\end{proof}

Let $\Oo_{adm}(\E_{\infty \mid \infty})[[\lambda]]$ denote the space of sequences of admissible functions on the spaces $\E_{N \mid N}$ with a formal parameter $\lambda$. As before, we can define the open-closed BV operators and RG flow operators on each space $\Oo(\E_{N \mid N})[[\lambda]]$.   There is something special, however, that happens in the supergroup case.  
\begin{deflem}
Suppose that $\{f_N\}, \{g_N\}$ are admissible sequences of functions in $\Oo(\E_{N \mid N})$.  Then so are   $\{\pa_{P(\epsilon, L)^C} f_N\}$, $\{\Delta_L^C f_N\}$, $\{\pa_{P(\epsilon, L)^{O,N\mid N}} f_N\}$,  $\{\Delta_L^{O,N\mid N} f_N\}$, $\{\{f_N,g_N\}^{C}_L\}$, $\{\{f_N,g_N\}^{O,N\mid N}_L\}$. 

 We will denote the corresponding operators on $\Oo^{(\ast)}_{adm}(\E_{\infty\mid \infty})$ by  
$$
\pa_{P(\epsilon, L)^C}, \Delta_L^C, \pa_{P(\epsilon, L)^{O,\infty\mid \infty}}, \Delta_L^{O,\infty\mid \infty}, \{-,-\}_L^C, \{-,-\}^{O,\infty\mid \infty}_L.
$$
\end{deflem}
The key difference with the $\gl(N)$ case is that, in the $\gl(N)$ case, the open-string BV operator and RG flow operator can increase the $N$-degree. (Recall that a sequence of functionals $\{f_N\}$ is weakly admissible with $N$-degree $k$ if the sequence $N^{-k} f_{N}$ is admissible).  In the $\gl(N \mid N)$ case, these operators preserve the functionals of $N$-degree zero, and so the class of admissible as opposed to weakly admissible sequences.   
\begin{proof}
It is clear that the operators coming from the closed-string sector preserve the class of admissible sequences of functionals.  We need to check the statement for the operators coming from the open-string sector. Since the BV bracket $\{-,-\}^{O,N\mid N}_{L}$ is determined by the BV Laplacian $\tr_{L}^{O, N \mid N}$, it suffices to check the desired property for the operators $\tr_{L}^{O, N \mid N}$ and $\pa_{P(\eps,L)^{O,N \mid N}}$.   

Both cases follow from lemma \ref{lemma_supergroup_operator}.
\end{proof}

Let $\Oo_{adm}(\E_{\infty \mid \infty})[[\lambda]]_+$ denote the subspace of $\Oo_{adm}(E_{\infty \mid \infty})[[\lambda]]$ consisting of series 
$$I = \sum I_{g,h} \lambda^{2g-2+h}$$
where $I_{g,h}$ has weight $h$ (see Definition \ref{definition-weight}) and $I_{0,0}$ and $I_{0,1}$ are at least cubic. 
\begin{definition}\label{quantum_BCOV_definition_supergroup}
A quantum open-closed BCOV theory for the supergroups $\gl(N \mid N)$ is a collection of functionals $I[L] \in \Oo_{adm}(\E_{\infty \mid \infty})$ satisfying the axioms listed in definition \ref{quantum_BCOV_definition}. 
\end{definition} 
The main result of this section is the following.
\begin{proposition}
A quantum open-closed BCOV theory for the supergroups $\gl(N \mid N)$ is the same thing as a quantum open-closed BCOV theory for the ordinary groups $\gl(N)$ where we work modulo $N$ as in definition \ref{definition-modulo-N}.
\end{proposition}
\begin{proof}
As we have seen, there is an isomorphism between the space $\Oo_{adm}(\E_{\infty \mid \infty})$ of sequences of admissible functions on $\E_{N \mid N}$ and the space $\Oo_{adm}(\E_{\infty})$ of sequences of admissible functions on $\E_{N}$. This implies that a functional $I[L] \in \Oo_{adm}(\E_{\infty})[[\lambda]]_+$ defining a modulo $N$ quantum theory can be viewed as a functional in $\Oo_{adm}(\E_{\infty \mid \infty})[[\lambda]]_+$ defining a supergroup theory. It remains to check that this isomorphism respects all the axioms.  The compatibility of the RG flow axiom and the quantum master equation axiom with this isomorphism follows from lemma \ref{lemma_supergroup_operator}. 
\end{proof}

\section{Existence and uniqueness-Local theory}
One of the main properties of quantum field theory is \emph{locality}. This allows us to build up the whole quantum theory from local data on the underlying manifold. Our main result in this first paper of this series is to prove that there exists a unique quantization on the local piece of a Calabi-Yau manifold, i.e., on $\C^n$, satisfying obvious local symmetries. Since the local theory illustrates the main structures of open-closed BCOV theory, we will describe the construction in detail. The main tool of the construction is the obstruction theory for perturbative quantization developed in \cite{Cos11}. 

We will construct the theory \emph{modulo $N$}, in the sense explained in definition \ref{definition-modulo-N}. Equivalently, this means that for the open string sector we use the supergroups $\gl(N \mid N)$ instead of $\gl(N)$.  In the language of Riemann surfaces, it means that we will construct the theory without free boundaries.

The main theorems we will prove in the rest of the paper are the following.
\begin{theoremA}
Let $X$ be a Calabi-Yau manifold with an action of $\C^\times$ which rescales the volume form (we call such an $X$ \emph{conical}).  Suppose we have constructed a quantum BCOV theory which only includes $I_{0,h,n;f}$ when $n+h + f \le 2$. In other words, we ask that we have only disk amplitudes with at most one closed-string marked point, and annulus amplitudes with no closed-string marked points, satisfying all our axioms modulo higher-order terms.  The disk amplitudes are specified by the classical limit of our theory. 

Then, there exists a \emph{unique} quantization of the open-closed BCOV theory on $X$ in a way compatible with the given $\C^\times$ action. 
\end{theoremA}
What this theorem shows is that, once we can specify certain disk and annulus amplitudes, the rest of the theory at all genera can be constructed uniquely.  The quantization problem is controlled, using results from \cite{Cos11}, by a certain obstruction-deformation cohomology group.  The main point is that there is a cancellation between the obstruction/deformation group for the open theory and the closed theory that appears once we go beyond the annulus level. 

This cohomology calculation does not apply at the annulus level. It turns out that there is a cohomology class giving a potential obstruction (or anomaly) to constructing the theory at the annulus level. 
\begin{theoremB}
On $\C^d$ when we work modulo $N$  (that is, use the supergroups $\gl(N \mid N)$ as our gauge groups) then there is a unique quantization at the annulus level. That is, the annulus anomaly vanishes (so the quantization exists) and there is no ambiguity in quantizing to the annulus level.  
\end{theoremB}

Theorem B will be proved by explicitly computing the annulus obstruction, which again vanishes by a marvelous cancellation between open string and closed string sectors. This completes the local theory of quantum open-closed BCOV theory, which says that the whole package is completely determined via deformation theory by the open string sector ($HCS_\infty$) and its first-order coupling with closed string sector ($I^{1-disk}_\infty$).

\section{Obstruction theory} Let $X$ be a conical Calabi-Yau.  Thus, $X$ has an action of $\C^\times$ which scales the volume form with weight $w \neq 0$.  Let us choose a K\"ahler metric on $X$ which is invariant under the $S^1$ inside this $\C^\times$.   In this section we will prove theorem $A$.  

The $S^1$ action on $X$ induces an action on the space of fields of BCOV theory, by naturality.  It therefore induces an action on all spaces of functionals on the space of fields. We will denote this action by $\rho$.  We view $S^1$ as a subset of $\C^\times$ with coordinate $\mu$ where $\abs{\mu} = 1$.  We will extend this $S^1$-action to the $\lambda$ variable by
$$
 \rho_{\mu}\lambda  = \mu^{w} \lambda 
$$
The fact that we have chosen an $S^1$-invariant K\"ahler metric on $X$ implies that 
$$
  \rho_{\mu}(K_L^C)=\mu^{-2w} K_L^C, \quad \rho_\mu(K_L^O)= \mu^{-w} K_L^O, 
$$
while $\dbar, \pa, \dbar^*$ all commute with $\rho_\mu$. The follow lemma is then an easy consequence. 

\begin{lemma}\label{lemma-QME-RG-S}
The quantum master equation and RG flow equation is compatible with $S^1$-action in the sense that
$$
  \bbracket{ \bracket{Q_\infty+\lambda^2 \Delta^C_L+\lambda \Delta_L^O}, \rho_\mu}=0, \quad \bbracket{e^{\lambda^2\pa_{P(\epsilon, L)^C}+\lambda \pa_{P(\epsilon, L)^{O,\infty}}}, \rho_\mu}=0. 
$$
\end{lemma}

\begin{definition} A quantum open-closed BCOV theory $I[L]=\sum_{g,h,f}I_{g,h;f}[L]\lambda^{2g+h}\in \Oo^{(\ast)}_{adm}[[\lambda]]_+$ is $S^1$-invariant if $I_{g,h;f}$ is of weight $-(2g-2+h+f)w$ under the $S^1$-action
$$
 \rho_\mu(I_{g,h;f}[L])=\mu^{-(2g-2+h+f)w} I_{g,h;f}[L], 
$$
or equivalently $I[L]/\lambda^2$ is $S^1$-invariant. Similarly, if we work with an $S^1$-invariant quantum open-closed BCOV theory modulo $N$ for $I[L]=\sum_{g,h}I_{g,h}[L]\lambda^{2g+h}\in \Oo_{adm}[[\lambda]]_+$, we require that 
$$
 \rho_\mu(I_{g,h}[L])=\mu^{-(2g-2+h)w} I_{g,h}[L].
$$
\end{definition}

By Lemma \ref{lemma-QME-RG-S}, the definition of $S^1$-invariant quantum open-closed BCOV theory is compatible with quantum master equation and RG flow equation.  We will analyze the obstruction problem of constructing $S^1$-equivariant open-closed quantum BCOV theory on $X$.   

\begin{definition} 
 Let $\T(X)$ (or just $\T)$ be the simplicial set of $S^1$-invariant quantum open-closed BCOV theories on $X$, defined modulo $N$. The zero-simplices of this simplicial set are quantizations in the sense discussed above. The $n$-simplices are families of quantizations over forms on the $n$-simplex, in the sense explained in \cite{Cos11}. 
 \end{definition}

We will use an inductive argument to analyze possible quantizations.
\begin{definition} We define an ordering on $\N\times \N$ by 
  $$
    (g_1, R_1)< (g_2, R_2), \quad  \text{if}\ g_1<g_2, \text{or}\  g_1=g_2, R_1<R_2. 
  $$
\end{definition}

\begin{definition} We define $\T^{\leq(G, R)}(X)$ to be the simplicial set of $S^1$-invariant quantum open-closed theories where the functionals $I_{g,h,n}[L]$ are defined for $(g, h+n)\leq (G, R)$.  Let $\T^{< (G,R)}(X)$ denote the simplicial set where the functionals are defined for $(g,h+n) < (G,R)$. (Recall that $n$ is the number of closed-string inputs and $I_{g,h}[L]=\sum_n I_{g,h,n}[L]$). \label{definition_theories_(G,R)} 
\end{definition}
This definition makes sense, since the RG flow equation and quantum master equation for the functionals $I_{g,h,n}[L]$ where $(g,h+n) \le (G,R)$ do not depend on the functionals $I_{g,h,n}[L]$ where $(g,h+n) \not \le (G,R)$.  

The results of \cite{Cos11} imply that the simplicial sets $\T^{\le(G,R)}(X)$ are Kan complexes and the maps $\T^{\le (G,R)}(X) \to \T^{<(G,R)}(X)$ are Kan fibrations.

The main theorem of this section is the following. 
\begin{theorem}\label{theorem-vanishing}
\label{theorem_cohomology_cancellation}For $(g, R)$ with $g>0$ or $g=0, R\geq 3$, the map 
$$
   \T^{\leq (g,R)}(X)\to \T^{< (g, R-1)}(X)
$$
is a weak equivalence. 
\end{theorem}
\begin{corollary}
The map $\T(X) \to \T^{\leq(0,2)}$ is a weak equivalence. 
\end{corollary}
The corollary is immediate: $\T(X)$ is a limit of an infinite sequence of fibrations each of which is a weak equivalence. 

The proof of this theorem will take up the rest of this section. 

The first thing we need to know is a result from \cite{Cos11}, which describes the obstruction-deformation complex for lifting from $\T^{<(G,R)}$ to $\T^{(G,R)}$ in terms of a complex of local functionals. 

Since $X$ is non-compact, it is useful to take our fields to have compact support. Let $\E_{c,N}$ denote the space of fields with compact support, where on the open sector we use the gauge group $\gl(N)$.  Let $\Oo_{adm,loc}(\E_{c, \infty})$ denote the space of admissible local functionals.   

It is useful to shift the cohomological degree and the $S^1$ action on this space of admissible local functionals.
\begin{definition}
Let $\Oo'_{adm,loc}(\E_{c,\infty})$ denote the space of admissible local functionals but with a cohomological shift, whereby the functionals of weight $h$ are shifted by $[h(d-3)]$.  The graded vector space $\Oo'_{adm,loc}(\E_{c,\infty})$ acquires a differential given by combining the linear differential on the space of fields with bracketing with our classical action functionals $I_{0,1,0}$ and $I_{0,1,1}$.

The space of admissible local functionals has an $S^1$ action coming from that on the space of fields by Lie derivative.  Let us give a slightly different $S^1$ action on $\Oo'_{adm,loc}(\E_{c,\infty})$ by saying that $S^1$ acts on the open-string fields in the usual way, by Lie derivative, but that $S^1$ acts on the closed-string fields $\PV(X)[[t]]$ via the Lie derivative on forms on $X$, using the isomorphism $\PV(X) \iso \Omega^{\ast,\ast}(X)$ coming from the holomorphic volume form on $X$.   This shifts the weight of the $S^1$ action on $\PV(X)$ by $w$.  
\end{definition}

Let 
$$
  V_{h,n}\subset \Oo'_{adm,loc}(\E_{c, \infty})
$$
be the subcomplex of those local functionals $I$ which are homogeneous of weight $n$ as a functional on $\PV_c(X)[[t]][2]$ and admissible of of weight $h$ as a functional on $\Omega^{0,*}_c(X)[1]\otimes \gl_N$.  Let
$$
   V_{R}=\bigoplus_{h+n=R} V_{h,n}. 
$$
 $V_{R}$ is a subcomplex of $\Oo'_{adm,loc}(\E_{c,\infty})$. The differential, explicitly, is 
$$
Q_\infty+\{I_{0,1,0},-\}^O+\{I_{0,1,1},-\}^C, \quad I_{0,1,0}=\CS_\infty, I_{0,1,1}=I_\infty^{1-disk}. 
$$
This differential commutes with the $S^1$ action on $V_{R}$ (where we use the modified $S^1$ action as discussed in the definition of $\Oo'_{adm,loc}(\E_{c,\infty})$). Let $V_{R}^{(k)}$ denote the subspace of weight $k$ under this $S^1$ action. 

The main result of \cite{Cos11}, in this situation, says the following.
\begin{proposition*}
The obstructions to lifting a theory in $\T^{<(g, R)}(X)$ to $\T^{\leq (g,R)}(X)$ lies in  
$$
   H^{(d-3)(2g-2)+1}(V^{(2-2g - R)w}_{R}). 
$$
The cohomological shift by $(d-3)(2g-2)$ arises because the functionals $I_{g,h,n}[L]$ have cohomological degree $(d-3)(2g-2+h)$, and the $(d-3)h$ shift is already account for in the definition of $\Oo'_{adm,loc}(\E_{c,\infty})$. The weight under the $S^1$ action reflects the weight required to have an $S^1$-invariant quantization. 

In addition, if the obstruction vanishes, the set of possible lifts up to homotopy is a torsor for $H^{(d-3)(2g-2)}(V^{(2-2g-R)w}_{R})$. 

More generally, the simplicial set of possible lifts is a torsor for the Dold-Kan simplicial Abelian group associated to the cochain complex $V_{R}^{(2-2g-R)w}[(d-3)(2g-2)]$. 
\end{proposition*}
The theorem will follow from the following proposition.
\begin{proposition}
 The complex $V^{(k)}_{R}$ has zero cohomology unless $k = 0$, and 
 $$H^\ast (V_R^{(0)} ) = H^*_{dR}(X)[2d]\otimes \Sym^R \left( \C((t))[2-d]^\vee \right).$$ 
\end{proposition}
This immediately provides a proof of the theorem.  The point is that the obstruction-deformation complex for lifting from $\T^{<(g,R)}$ to $\T^{\le(g,R)}$ is $V^{(2-2g-R)}[(d-3)(2g-2)]$, and according to the proposition this is zero if $2-2g-R \neq 0$, showing that there is a unique (up to contractible choice) lift to the next order.

\subsection{Proof of the proposition}
The proof of this proposition will occupy the rest of this section.

The results of \cite{Cos11} allow us to describe the cohomology of the complexes of local functionals  in terms of jets of the space of fields.  The space 
$$\E_N = \PV(X)[[t]][2] \oplus \Omega^{0,\ast}(X) \otimes \gl_N[1] $$ of fields is the space of smooth sections of a graded vector bundle on $X$ which we denote $E_N$.  Let $J(E_N)$ denote the bundle of jets of $E_N$, and for $x \in X$, let $J_x(E_N)$ denote the fibre of this bundle at $x$. Concretely, 
$$
J_x(E_N) = J_x(\PV)[[t]][2] \oplus J_x(\Omega^{0,\ast}) \otimes \gl_N[1] 
$$
where, if we choose coordinates near $x$, we can write the spaces of jets of polyvector fields and of jets of the Dolbeault algebra as 
\begin{align*}
J_x(\PV) &= \C[[z_i, \pa_i,  \bar z_j, \d \bar z_j  ]]  \\
J_x(\Omega^{0,\ast}) &= \C[[z_i,   \bar z_j, \d \bar z_j  ]].
\end{align*}
Here the generators $\d \bar z_j$ and $\pa_i$ are of cohomological degree $1$. 

Let $\Oo(J(E_N))$ denote the bundle on $X$ of algebras of functionals modulo constants on the bundle of jets. The fibre of $\Oo(J(E_N))$ at $x$ is defined to be the completed non-unital symmetric algebra of the topological dual of $J_x(E_N)$. 

Recall that sections of a jet bundle are always a $D$-module (or, equivalently, an infinite-rank vector bundle with a flat connection). Here $D$ is the sheaf of smooth differential operators.  If we choose local coordinates on a neighbourhood $U$ of $x \in X$, we can trivialize the jet bundles near $x$ (as above).  The sections of the jet bundle on $U$ can be described as
\begin{align*} 
 \Gamma(U, J(E_N)) &= \cinfty(U) \what{\otimes} J_x(E_N)\\
 &= \cinfty(U)[[z_i,\pa_i,\bar z_j, \d \bar z_j,t]] \oplus \cinfty(U)[[z_i, \bar z_j, \d \bar z_j]]\otimes \gl_N.
 \end{align*}
 In local coordinates, the algebra $D(U)$ of differential operators on $U$ is
 $$
D(U) = \cinfty(U)\left[\dpa{z_i}, \dpa{\bar z_j}\right].
 $$
The sub-algebra $\cinfty(U)$ of $D(U)$ acts on $\Gamma(U,J(E_N))$ in the evident way.  Thus, the $D$-module structure, in coordinates, is determined by the action of the commutative algebra $\C[\dpa{z_i}, \dpa{\bar z_j}]$ of constant-coefficient differential operators.   By the Leibniz rule, it suffices to say how these constant-coefficient operators act on constant sections of the jet bundle in our chosen trivialization.   In the concrete expression given above for the space $J_x(E_N)$ of constant sections,  the action of $\dpa{z_i}$ and $\dpa{\bar z_j}$ on $J_0(\PV)$ and $J_0(\Omega^{0,\ast})$ is given by minus the evident derivations.

With these conventions, one checks that a flat section of $J(E_N)$ -- that is one annihilated by all the $\dpa{z_i}$ and $\dpa{\bar z_j}$ -- is precisely the jet of an actual smooth section of the bundle $E_N$.

By naturality, the bundle of algebras $\Oo(J(E_N))$ is a $D$-module as well.   Then, there is an isomorphism of graded vector spaces
$$
\Ool(\E_{c,N}) \iso \Gamma(X, \Omega^{d,d} \otimes_{D}\Oo( J(E_N) )).
$$
Here $\Omega^{d,d}$ refers to the sheaf of top forms on $X$, which is a right $D$-module.  We are tensoring the sheaf of right modules $\Omega^{d,d}$  with the sheaf of left modules $\Oo(J(E_N))$ over the sheaf of differential operators $D$, and then taking global sections.

Let $\Oo_{adm}(J(E_\infty))$ be the bundle of admissible sequences of functions (in the same sense as before, and recall that it is defined modulo constants) on the sequence of bundles $J(E_N)$. Let $\Oo'_{adm}(J(E_\infty))$ be the same thing but with the shift in cohomological degree described above, whereby functions of weight $h$ are shifted by $[(d-3)h]$.  Then, the space of admissible functionals $\Oo'_{adm,loc}(\E_{c,\infty})$ has a similar description,
$$
\Oo'_{adm,loc}(\E_{c,\infty}) \iso \Gamma(X, \Omega^{d,d} \otimes_{D}\Oo'_{adm}(J(E_\infty))).
$$

It turns out (see \cite{Cos11}) that the $D$-module of functions on jets is always flat, so that we can replace the actual tensor product by the derived tensor product to find a quasi-isomorphism
$$
\Oo'_{adm,loc}(\E_{c,\infty}) \simeq \Gamma(X, \Omega^{d,d} \otimes_{D}^{\mbb{L}} \Oo'_{adm}(J(E_\infty))). 
$$
Finally, using the canonical Spencer resolution of the right $D$-module $\Omega^{d,d}$ by flat $D$-modules one finds a quasi-isomorphism
$$
\Oo'_{adm,loc}(\E_{c,\infty}) \simeq  \Omega^{\ast,\ast}(X, \Oo'_{adm}(J(E_\infty))[2d]. 
$$
On the right hand side we have the de Rham complex of $X$ with coefficients in the $D$-module $\Oo_{adm}(J(E_\infty))$, with a shift by $2d$.

We let
\begin{equation*}
W_{R} = \bigoplus_{i+j = R} \Sym^i \left( J (\PV[[t]]) [2] ^\vee \right) \otimes \Sym^j \left( \op{Cyc} \left( J (\Omega^{0,\ast})^\vee[-1]  \right)[d-3]\right) \subset \Oo'_{adm}(J(E_\infty)). \tag{$\dagger$}
\end{equation*}
Then, $W_R$ is the $D$-module version of the  complex $V_R \subset \Oo'_{adm,loc}(\E_{c,\infty})$, and we have a quasi-isomorphism
$$
V_R \iso \Omega^{\ast,\ast}(X, W_R)[2d]. 
$$

If we want to compute the cohomology of $V_R$, we can use a spectral sequence argument.  Since $\Omega^{\ast,\ast}(X, W_R)$ is a double complex, with one differential coming from that on the dg $D$-module $W_R$ and the other being the de Rham differential, we can apply the spectral sequence for a double complex. This yields a convergent spectral sequence of the form 
$$
H^q_{dR}(X, H^p(W_R)) \Rightarrow H^{p+q-2d} (V_R).
$$

To compute the first page of this spectral sequence, we need to compute the cohomology of the dg $D$-module $W_R$. We find the following.
\begin{proposition}\label{proposition-D-module-vanishing}
The cohomology of $W_R$ is the trivial $D$-module with fibre the space $\Sym^R( \C((t))[2-d]^\vee)$. 
\end{proposition}
\begin{proof}[Proof of the proposition]
We will compute the cohomology of $W_R$ locally on $X$, in local coordinates. So let us assume for now that $X = \C^d$, and work near the origin $0 \in \C^d$. Let $J_0(E_\infty)$ denote the fibre of the jet bundle at the origin, and consider the space $\Oo'_{adm}(J_0(E_\infty))$ of admissible functionals. Inside this we have the fibre $W_{R,0}$ of $W_R$ at the origin. We will compute the cohomology of $W_{R,0}$. 

We can write 
\begin{align*} 
 \Oo'_{adm}(J_0(E_\infty)) &= \Oo( J_0(\PV[[t]]) [2] ) \otimes \what{\Sym} \left(\op{Cyc} \left( J_0 (\Omega^{0,\ast})^\vee[-1]  \right) [d-3]\right)
 \end{align*}
modulo constants and 
\begin{align*}
W_{R,0} &= \bigoplus_{i+j = R} \Sym^i \left( J_0 (\PV[[t]]) [2] ^\vee \right) \otimes \Sym^j \left( \op{Cyc} \left( J_0 (\Omega^{0,\ast})^\vee[-1]  \right)[d-3]\right) \subset \Oo'_{adm}(J_0(E_\infty)). 
\end{align*}
We will prove that the cohomology of $W_{R,0}$ is isomorphic to $\Sym^R \left( \C((t))[-d] \right)$. The differential on the whole complex preserves the decomposition into the product of the spaces $W_{R,0}$, so this will prove the result.  
 
Let us filter  $W_{R,0}$ by saying that $F^k$ is the subspace which can be expressed in this way where $i \le k$.  

The differential on our complex has four terms. Two of these terms come from the linear differentials $\dbar + t \partial$ on $\PV[[t]]$ and $\dbar$ on $\Omega^{0,\ast}$;  the other two come from bracketing with the classical action functionals $I_{0,1,0}$ and $I_{0,1,1}$.

All these terms preserve this filtration, and on the associated graded we forget the term coming from $I_{0,1,1}$. 

In the first page of the spectral sequence associated to the filtration $F^k$, we thus need to compute the cohomology of the complexes ($\dagger$) with the differential arising from the linear differentials on $\PV[[t]]$ and $\Omega^{0,\ast}$, and from the classical action functional $I_{0,1,0}$. 

The key fact we need is that the term in the differential corresponding to $I_{0,1,0}$ gives precisely the differential on the cyclic cochain complex of the commutative algebra 
$$J_0 \Omega^{0,\ast} = \C[[z_i, \bar z_i, \d \bar z_i]].
$$
This follows from Tsygan and Loday-Quillen's theorem on cyclic homology \cite{Tsy83,LodQui84}.  As, $I_{0,1,0}$ comes from the holomorphic Chern-Simons interaction on $\gl_N$, and thus gives the Chevalley-Eilenberg differential on $C^\ast ( J_0 \Omega^{0,\ast} \otimes \gl_N)$.  Tsygan and Loday-Quillen show that as $N \to \infty$ this becomes the cyclic differential. 

So, we see that the cohomology of the associated graded is
$$
\Sym^i H^\ast\left( J_0 PV[[t]][2]^\vee \right) \otimes \Sym^j \left( HC^\ast ( J_0 \Omega^{0,\ast} ) [d-4] \right)   
$$
where $HC^\ast$ refers to cyclic cohomology. 

Now, we apply the Hochschild-Kostant-Rosenberg theorem to compute cyclic cohomology of $J_0 \Omega^{0,\ast}$.  $J_0 \Omega^{0,\ast}$ is quasi-isomorphic to the power series algebra $\C[[z_i]]$.  Thus, we find that cyclic homology of $J_0 \Omega^{0,\ast}$ is
$$
HC_\ast ( J_0 \Omega^{0,\ast} ) = \left(t^{-1}\C[[z_i, \d z_i ]] [t^{-1}][-2]  , t \d_{dR} \right).
$$
Here, $\d z_i$ is of cohomological degree $-1$.  Thus, cyclic homology is built from forms on the formal disc $\Spec \C[[z_i]]$.  

Now, because we are working on a Calabi-Yau, we can use the volume form to identify forms with polyvector fields.  Let $\PV(\what{D}_d)$ denote holomorphic polyvector fields on the  $d$-dimensional formal disc $\Spec \C[[z_i]]$.  Then, we can rewrite cyclic homology as
$$
HC_\ast ( J_0 \Omega^{0,\ast} ) = t^{-1}\PV(\what{D}_d) [t^{-1} ] [d-2] 
$$
with differential $t \del$.  Taking into account the degree shifting of $h(d-3)$ which we've been carrying along, we have shown that the fiber of associated graded of our complex is quasi-isomorphic to 
$$
\Omega_0^{*,*}[2d]\otimes_{\C}\Sym^R \left( \PV(\what{D}_n)[[t]]  [2] \oplus t^{-1}\PV( \what{D}_n) [t^{-1}] [2] \right)^\vee. 
$$

The differential is $t \partial$ on each summand, but it does not map $t^{-1}\PV(\what{D}_n)[t^{-1}]  [2]$ to $ \PV( \what{D}_n) [[t]] [2]$. 
\begin{lemma}
The differential on the next page of the spectral sequence arises from the map
$$
c t \partial : t^{-1}\PV(\what{D}_n) [2] \to \PV(\what{D}_n) 
$$
where $c$ is a non-zero constant.  
\end{lemma}
\begin{remark} The constant $c$ is in fact determined by the annulus anomaly cancellation condition. See Theorem \ref{annulus-cancellation}. 
\end{remark}
\begin{proof} The differential on the next page of the spectral sequence is given by $\{I^{1-disk},-\}^C$. The closed string BV bracket $\{-,-\}^C$ is represented by the kernel of the operator $\pa: \PV(X)\to \PV(X)$, which is $(\partial \otimes 1) \delta_{Diag}$ where $\delta_{Diag}$ is the $\delta$-current on the diagonal of $X \times X$, viewed as a polyvector field via the isomorphism between forms and polyvector fields. In particular, it only involves the component $t^0\PV(X)$ of the full closed string fields $\PV(X)[[t]]$. For $\mu \in \PV(X)$, 
$$
   I^{1-disk}(\mu, A)=c \int_X \Tr (A\wedge \pa A\wedge\cdots \wedge \pa A)\wedge (\mu\vdash \Omega_X)
$$
which represents the natural pairing between $\PV(X)$ and the hochschild chains of $\Omega^{0,*}(X)$. It follows that $\{I^{1-disk},-\}^C$ is the dual of 
\begin{align*}
    \eta:  \op{Cyc}(\Omega^{0,*}(X))&\to \PV(X)
 \end{align*}
where $
      \eta(A_1\otimes \cdots \otimes A_k)$ is identified with $c(\pa A_1\wedge \cdots \wedge \pa A_k)
$ under the natural isomorphism between $\PV(X)$ and $\Omega^{*,*}(X)$ with respect to $\Omega_X$. The lemma is a local statement of this at a fiber of the jets. 

\end{proof}

This lemma shows that, after taking cohomology on this page of the spectral sequence, we find 
$$
W_{R,0} \simeq \Sym^R \left( \C \cdot (\partial_1 \wedge \dots \wedge \partial_d)((t)) [2] \right)^\vee. 
$$
as desired.  It is easy to verify that the $D$-module structure on the bundle with fibre $W_{R,0}$ is trivial, thus completing the proof of the proposition.  
\end{proof}
This completes the proof of the theorem that the map $\T^{\le (g,R)} \to \T^{< (g,R)}$ is a weak homotopy equivalence as long as $2g-2+R > 0$.  That is, once we have constructed $I_{0,2,0}[L]$, then we obtain a unique (up to homotopy) quantization $I_{g,h,n}[L]$ for all $(g,h,n)$. 

\subsection{Theorem A on $\C^d$} In this section, $X$ will denote the Calabi-Yau manifold $\C^d$, with linear coordinate $z_1, \cdots, z_d$. The holomorphic volume form is chosen to be 
$$
\Omega_X=\d z_1\wedge \cdots \wedge \d z_d.
$$
Let us explain in this section how to modify Theorem A proved above, stating that there is a quantization beyond the annulus level, to the case of $\C^d$, where we in addition assume that our theory is translation invariant. 

We choose the standard K\"{a}hler metric ${i\over 2}\sum_{i=1}^d \d z_i\wedge \d\bar z_i$ to define our gauge fixing.

A quantum BCOV theory on $\C^d$ is translation invariant if each $I_{g,h,n}[L]$ is invariant under the action of the Abelian group $\C^d$ acting on everything.  We need a stronger notion of translation invariance called holomorphic translation invariance.

Notice that our spaces of fields on $\C^d$ are of the form 
$$\E_N = \Omega^{0,\ast}(\C^d) \otimes V_N$$
where $V_N$ is the graded vector space
$$
V_N = \gl_N[1] \oplus \C[\partial_i][[t]][2]
$$
where the variables $\partial_i$, for $i = 1,\dots, d$, have degree $+1$.

The operator $\iota_{\partial_{\zbar_i}}$ of contracting with the vector field $\partial_{\zbar_i}$ maps $\Omega^{0,k}(\C^d)$ to $\Omega^{0,k-1}(\C^d)$ by removing one copy of $\d \zbar_i$ with the appropriate sign. The Cartan homotopy formula tells us that 
$$
\mscr{L}_{\partial_{\zbar_i}} = \left[ \dbar, \iota_{\partial_{\zbar_i}} \right]. 
$$
The operator $\iota_{\partial_{\zbar_i}}$ extends to an operator on the space of fields $\Omega^{0,\ast}(\C^d) \otimes V_N$ by tensoring with the identity on $V_N$, and satisfies the commutation relation
$$
\mscr{L}_{\partial_{\zbar_i}} = \left[ Q_N, \iota_{\partial_{\zbar_i}} \right] 
$$
where $Q_N$ is the linear differential on the space of fields. 

By naturality, the operators $\iota_{\partial_{\zbar_i}}$ extend to derivations which we denote by same notation on the spaces of functionals on the space of fields.  One can check that the classical action functionals $I_{0,1,0}$ and $I_{0,1,1}$ satisfy
$$
\iota_{\partial_{\zbar_i}} I_{0,1,0} = 0 \qquad \iota_{\partial_{\zbar_i}} I_{0,1,1} = 0. 
$$
\begin{lemma}
The operators $\iota_{\partial_{\zbar_i}}$ commute with the RG flow operator from scale $\eps$ to scale $L$ and with the BV Laplacian of scale $L$ on the space $\Oo(\E_{c,N})[[\lambda]]_+$.  Further, $\iota_{\partial_{\zbar_i}}$ is a derivation for the scale $L$ BV Poisson bracket. 
\end{lemma}
\begin{proof}
This follows from the fact that $\iota_{\partial_{\zbar_i}}$ commutes with the order two differential operators $\tr^O_L$, $\tr_L^C$,  $\partial_{P(\eps,L)^O}$ and $\partial_{P(\eps,L)^C}$. 
\end{proof}
\begin{definition}
A quantum BCOV theory on $\C^d$ (working, as usual, modulo $N$) is \emph{holomorphically translation invariant} if it is translation invariant and if each $I[L]$ is in the kernel of the operators $\iota_{\partial_{\zbar_i}}$.   Let $\T(\C^d)$ denote the simplicial set of holomorphically translation-invariant and $U(d)$-invariant quantum BCOV theories on $\C^d$, working modulo $N$. As in definition \ref{definition_theories_(G,R)} let $\T^{<(G,R)}(\C^d)$ and $\T^{\le (G,R)}(\C^d)$ denote the corresponding simplicial sets where we work up to order $<(G,R)$ or $\le (G,R)$.   
\end{definition}
\begin{remark}
By $U(d)$ invariance we mean a theory that is $SU(d)$ invariant and where under the remaining $S^1$ in $U(d)$, which scales the volume form on $\C^d$ with weight $d$, each $I_{g,h,n}[L]$ is of weight $-d (2g-2+h)$. 
\end{remark}
We will prove that there is a unique holomorphically translation-invariant BCOV theory on any $\C^d$ (where $d$ is odd).   The first step is to prove, just as we did for theories without the holomorphic translation-invariant condition, that once we construct a quantization at the level of annuli we automatically get a full quantization. This follows from the next theorem, which is the analog in this context of theorem \ref{theorem-vanishing}. 

\begin{theorem}\label{theorem-vanishing-flat-space}
For $(g, R)$ with $g>0$ or $g=0, R\geq 3$, the map 
$$
   \T^{\leq (g,R)}(\C^d)\to \T^{< (g, R-1)}(\C^d)
$$
is a weak equivalence. 
\end{theorem}

The proof is along the lines of the proof of theorem \ref{theorem-vanishing}, and will take the rest of this section.  We will analyze the obstruction-deformation complex for quantizing the holomorphically translation invariant theory.  As above let
$$
V_N = \gl_N[1] \oplus \C[\partial_i][[t]][2].
$$
In the notation of the proof of \ref{theorem-vanishing} we defined
\begin{align*} 
 E_N &= V_N [\d \zbar_i] = \gl_N[ \d \zbar_i][1] + \C[\partial_{j}, \d \zbar_i] [[t]] [2]
 \end{align*}
so that our space of fields is the space of sections of the trivial bundle on $\C^d$ with fibre $E_N$.  Here $\d \zbar_i$ and $\partial_i$ have cohomological degree $1$, and the indices $i,j$ run from $1$ to $d$. 

Let 
\begin{align*} 
 J_0(E_N) &= E_N[[z_i, \zbar_j]] \\
 &= V_N[[z_i, \zbar_j, \d \zbar_k]]
 \end{align*}
be the fibre at zero of the jet bundle associated to the trivial bundle with fibre $E_N$.  In the proof of theorem \ref{theorem-vanishing} we saw how to express the obstruction-deformation complex for quantizing BCOV theory in terms of a $D$-module built from the complex of admissible functions on $J_0(E_N)$.   We will do something similar for the holomorphically translation-invariant case.  

Note that $J_0(E_N)$ is acted on by the Abelian Lie algebra $\C^{2d}$ spanned by $\dpa{z_i}$ and $\dpa{\zbar_j}$, by differentiating formal power series in the evident way.  This is the Lie algebra of complexified translations on $\C^d$.    It is also acted on by the larger Abelian graded Lie algebra $\C^{2d} \oplus \C^d[1]$ whose odd elements are spanned by $\dpa{\d \zbar_i}$. The operator $\dpa{\d \zbar_i}$ is the same as the operator of contracting with $\dpa{\zbar_i}$; we use the notation $\dpa{\d \zbar_i}$ to avoid confusion at this stage.

As in the proof of theorem \ref{theorem-vanishing} we let $\Oo(J_0(E_N))$ denote the functionals modulo constants on $E_N$, which is the non-unital completed symmetric algebra of the dual of $J_0(E_N)$.  By naturality, $\Oo(J_0(E_N))$ is acted on by the Abelian graded Lie algebra $\C^{2d} \oplus \C^d[1]$ which acts on $J_0(E_N)$.   Let 
$$
U_{\C^{2d}}=\C\left[{\pa \over \pa z_i}, {\pa \over \pa \bar z_j}\right], \quad   U_{\C^{2d|d}}=\C\left[{\pa \over \pa z_i}, {\pa \over \pa \bar z_j}, {\pa \over \pa d\bar z_k}\right]
$$ 
denote the universal enveloping algebra for the Abelian Lie algebra $\C^{2d}$ and $\C^{2d} \oplus \C^d[1]$ respectively. 

$\Hom(J_0(E_N), \C)$ is the space of translation-invariant differential operators from the space $\E_N$ of fields to the space $\cinfty(\C^d)$.  Any sequence of elements
$$
D_1,\dots,D_n \in \Hom(J_0(E_N),\C)
$$
of such differential operators gives rise to a translation-invariant local functional on $\E_{c,N}$ by the formula
$$
e \mapsto \int_{\C^d} \prod \d z_i \prod \d \zbar_i D_1(e) \dots D_n(e)
$$
for $e \in \E_{c,N}$.

This expression gives rise to an identification
$$
\Omega^{d,d}_0 \otimes_{U_{\C^{2d}} } \Oo(J_0(E_N)) \iso \Oo_{loc}(\E_{c,N})^{\C^d} 
$$
between translation-invariant local functionals and the tensor product of the left module $\Oo(J_0(E_N))$ over $U_{\C^{2d}}$ with the trivial right module of rank one $\Omega^{d,d}_0$. (We think of $\Omega^{d,d}_0$ as the space of translation invariant $(d,d)$-forms on $\C^d$).  Tensoring over $U_{\C^{2d}}$ has the effect of removing those Lagrangians which are total derivatives. 

There is a similar description of holomorphically translation invariant local functionals using $U_{\C^{2d|d}}$ instead. Since the Lie algebra $\C^{2d} \oplus \C^d[1]$ acts on $\Omega^{0,\ast}(\C^d)$, the universal enveloping algebra $U_{\C^{2d|d}}$ of this Lie algebra does as well, and by differential operators.  The algebra $U_{\C^{2d|d}}$ is precisely the algebra of differential operators from $\Omega^{0,\ast}(\C^d)$ to itself which commute with the action of the extended translation Lie algebra $\C^{2d} \oplus \C^d[1]$ (whose odd elements act on $\Omega^{0,\ast}(\C^d)$ by contraction with the vector fields $\dpa{\zbar_i}$).   We can thus think of $A$ as being the algebra of constant-coefficient differential operators on the graded manifold whose functions are $\Omega^{0,\ast}(\C^d)$.

In a similar way, we can consider the differential operators
$$
\E_{N} = \Omega^{0,\ast}(\C^d,V_N) \to \Omega^{0,\ast}(\C^d)
$$
which are constant-coefficient, in the sense that they commute with the action of the graded Lie algebra $\C^{2d} \oplus \C^d[1]$. The collection of such differential operators will be denoted by $ \op{Diff}^{\C^{2d} \oplus \C^d[1]} (\E_N,\Omega^{0,\ast}(\C^d))$. Every such constant-coefficient operator is obtained by combining a linear map $V_N \to \C$ with an element of $U_{\C^{2d|d}}$, which differentiates the variables $z_i$, $\zbar_j$ and $\d \zbar_k$.  

Thus, we have natural isomorphisms
\begin{align*} 
 \op{Diff}^{\C^{2d} \oplus \C^d[1]} (\E_N,\Omega^{0,\ast}(\C^d)) &= \op{Hom}(V_N, U_{\C^{2d|d}}) \\
 &= \op{Hom}(V_N[[z_i,\zbar_j, \d \zbar_k]], \C) \\
 &= J_0(E_N)^\vee.
 \end{align*}
Here we are using the fact that we can identify $U_{\C^{2d|d}}$ with the linear dual of the power series algebra $\C[[z_i, \zbar_j, \d \zbar_k]]$.

Thus, $J_0(E_N)^\vee$ is the space of differential operators from $\E_N \to \Omega^{0,\ast}$ which commute with the action of the graded Lie algebra $\C^{2d} \oplus \C^d[1]$.

If $D_1,\dots,D_n \in J_0(E_N)^\vee$ are such differential operators, we can construct a holomorphically translation-invariant local functional on $\E_N$ by the formula
$$
e \mapsto \int_{\C^d} \prod \d z_i D_1(e) \wedge \dots \wedge D_n(e).
$$
On the right hand side we are wedging the elements $D_i(e) \in \Omega^{0,\ast}(\C^d)$ and integrating against the holomorphic volume form.  Since the integration map in the Dolbeault complex is of cohomological degree $-d$, this map is of degree $-d$.  

This gives us a map
$$
\Omega^{d,0}_0 \otimes \prod_{n > 0} \Sym^n( J_0(E_N)^\vee)[d] \to \Ool (\E_N)^{\C^{2d} \oplus \C^d[1]}
$$
(on the right hand side, we are using the fact that holomorphically translation-invariant local functionals are precisely those invariant under the graded Lie algebra $\C^{2d} \oplus \C^d[1])$.

There is some redundancy in this description of holomorphically translation-invariant local functionals, because of total derivatives.  If we take care of this, we find an isomorphism
$$
\Omega^{d,0}_0 \otimes_{U_{\C^{2d|d}}} \Oo(J_0(E_N))[d] \iso \Ool(\E_N)^{\C^{2d} \oplus \C^d[1]}. 
$$
Here the line $\Omega^{d,0}_0$ of translation-invariant $(d,0)$-forms is made into a right $U_{\C^{2d|d}}$-module in the trivial way.  

We will use this description of holomorphically translation-invariant local functionals to compute the obstruction-deformation complex for quantizing the theory on $\C^d$. 

\subsection{Admissible functions}
As discussed earlier, we let $\Oo_{adm}(J_0(E_{\infty})$ denote sequences of admissible functions on $J_0(E_N)$, and we let $\Oo_{adm,loc}(\E_{c,\infty})$ denote admissible sequences of local functions on $\E_{c,N}$.  We let $\Oo'_{adm}(J_0(E_{\infty})$ and $\Oo'_{adm,loc}(\E_{c,\infty})$ denote the same objects, except where the cohomological degree of the space of functions of weight $h$ is shifted by $[h(d-3)]$.  Everything in sight has an action of $U(d) = SU(d) \times S^1$.  We change the $S^1$ action, as in the proof of theorem \ref{theorem-vanishing}, on $\Oo'_{adm}(J_0(E_{\infty})$ and $\Oo'_{adm,loc}(\E_{c,\infty})$ so that the $S^1$ action on polyvector fields is the natural action on forms by pull-back, where we identify polyvector fields and differential forms using a chosen holomorphic volume form.  

The graded vector space $\Oo'_{adm}J_0(E_N))$ has a differential, as in the proof of theorem \ref{theorem-vanishing}, arising from the linear differential $Q$ on the space $\E_N$ of fields together with bracketing with $I_{0,1,0}$ and $I_{0,1,1}$.  Similarly the space $\Oo'_{adm,loc}(\E_{c,N})$ of local functionals has a differential.  The Lie algebra $\C^{2d} \oplus \C^d[1]$ of extended translations has a differential too, which sends $\dpa{\d \zbar_i}$ to $\dpa{\zbar_i}$.  The Cartan homotopy formula (see above) tells us that the dg Lie algebra $\C^{2d} \oplus \C^{d}[1]$ acts on the cochain complexes $\Oo_{adm,loc}(\E_{c,N})$ and $\Oo(J_0(E_N))$ in a way compatible with the differentials. It follows that the sub-space of invariants under the Lie algebra $\C^{2d} \oplus \C^d[1]$ is a sub-complex of $\Ool(\E_{c,N})$ of local functionals.  The discussion above now tells us that we have an isomorphism
$$
\Omega^{d,0}_{0} \otimes_{U_{\C^{2d\d}}} \Oo'_{adm}(J_0(E_\infty))[d] \iso \Oo'_{adm,loc}(\E_{c,\infty})^{\C^{2d} \oplus \C^d[1]}
$$
of cochain complexes.  Here we must treat $U_{\C^{2d\d}}$ as a dg commutative algebra, the universal enveloping algebra of the dg Lie algebra $\C^{2d} \oplus \C^d[1]$. 

Now, $\Oo'_{adm}(J_0(E_\infty))$ is a flat $U_{\C^{2d\d}}$-module so we can replace the tensor product by a derived tensor product.  This gives us a quasi-isomorphism
$$
\Omega^{d,0}_0 \otimes^{\mbb L}_{U_{\C^{2d|d}}} \Oo'_{adm}(J_0(E_\infty))[d] \simeq \Oo'_{adm,loc}(\E_{c,\infty})^{\C^{2d} \oplus \C^d[1]}.  
$$

The dg algebra $U_{\C^{2d|d}}$ is quasi-isomorphic to the sub-algebra $\C[\dpa{z_i}]$ of constant-coefficient holomorphic differential operators.  This implies that we have a quasi-isomorphism
$$
\Omega^{d,0}_0 \otimes^{\mbb L}_{\C[\dpa{z_i}]} \Oo'_{adm}(J_0(E_N))[d] \simeq \Oo'_{adm,loc}(\E_{c,\infty})^{\C^{2d} \oplus \C^d[1]}. 
$$
Using a standard resolution of $\Omega^{d,0}_0$ as a right $\C[\dpa{z_i}]$-module leads us to a quasi-isomorphism
$$
\Oo'_{adm}(J_0(E_N)) [\d z_i] [2d] \simeq \Oo'_{adm,loc}(\E_{c,\infty})^{\C^{2d} \oplus \C^d[1]} .
$$
On the left hand side of this quasi-isomorphism the complex $\Oo'_{adm}(J_0(E_\infty)) [\d z_i]$ is equipped with the de Rham differential as well as the internal differential on $\Oo'_{adm}(J_0(E_\infty))$. In other words, the left hand side of this equation is the translation-invariant subcomplex of the  holomorphic de Rham complex of $\C^d$ with coefficients in the dg $D$-module $\Oo'_{adm}(J_0(E_\infty))$.

\subsection{Completion of the proof}
Let 
$$
V_{h,n} \subset \Oo'_{adm,loc}(\E_{c,\infty})^{\C^{2d} \oplus \C^d[1]}
$$
be the subcomplex of holomorphically translation-invariant local functionals on $\C^d$ spanned by functionals of weight $h$ and homogeneous of degree $n$ as a function of the closed string inputs.  Let
$$
V_R = \bigoplus_{h+n=R} V_{h,n}.
$$
$V_R$ is a subcomplex of $\Oo'_{adm,loc}(\E_{c,\infty})^{\C^{2d} \oplus \C^d[1]}$ and has an $S^1$ action induced from that on the larger complex. Let 
$$
W_{h,n} \subset \Oo'_{adm}(J_0(E_{\infty})
$$
be defined analogously to $V_{h,n}$, and let
$$
W_R = \bigoplus_{h+n=R} W_{h,n}.
$$
Then, there is a quasi-isomorphism 
$$
V_R \simeq \left( W_R \otimes \C[\d z_i] [2d], \d_{W_R} + \d_{dR} \right). 
$$

Let $W_{R}^{(k)}$ and $V_{R}^{(k)}$ denote the subspaces of $W_R$ and $V_R$ where the $S^1$inside $U(d)$ acts by weight $k$. 

The obstruction theory developed in \cite{Cos11} tells us that the obstruction to lifting a theory in $\T^{<(g,R)}(\C^d)$ to one in $\T^{\le (g,R)}(\C^d)$ is an element in the $SU(d)$-invariants in  
$$
H^{(d-3)(2g-2)+1}  ( V_R^{-(2g-2+R)d} ) 
$$
and that if there is no obstruction, the set of lifts is a torsor for the $SU(d)$-invariants in $H^{(d-3)(2g-2)}(V_{R}^{-(2g-2+R)d})$.

We will show that $H^\ast(V_{R}^{(k)}) = 0$ if $k < 0$, thus completing the proof of the theorem.  

There is a spectral sequence which allows us to compute the cohomology of $V_R$ in terms of that of $W_R$, of the form
$$
H^\ast(W_R)[\d z_i] [2d] \Rightarrow H^\ast(V_R). 
$$
We know from \ref{proposition-D-module-vanishing} that
$$
H^\ast(W_R) = \Sym^R( \C((t))[2-d]^\vee) 
$$
and that the $U(d)$ action on this is trivial.  It follows  immediately that the cohomology of $V_R$ is concentrated in non-negative $S^1$-weights, as desired.

This completes the proof  of the theorem.

\section{Annulus anomaly cancellation}
To complete the proof of the existence and uniqueness of the quantum BCOV theory on $X = \C^d$, it remains to prove that we can construct the theory at the annulus level.  We will do this in this section.  

Since we need to analyze the theory on $\C^d$ very explicitly to prove this result, it will be helpful to write out explicitly the regularized kernels for BV operators and propagators. We choose the standard K\"{a}hler metric ${i\over 2}\sum_{i=1}^d \d z_i\wedge \d\bar z_i$ to define our gauge fixing. Then, we have
$$
K_L^C={1\over (4\pi L)^d} e^{-|z-w|^2/4L}\prod_{i=1}^d(\d\bar z_i-\d\bar w_i) \prod_{i=1}^d (\pa_{z_i}-\pa_{w_i}) \in \PV(X\times X),
$$
where $\{z_i\}, \{w_i\}$ denote linear coordinates on $X\times X$, and $|z-w|=\sqrt{\sum_{i=1}^d|z_i-w_i|^2}$. The closed string BV operator $\Delta_L^O$ is the second order operator associated to the kernel
$$
(\pa\otimes 1)K_L^C={1\over (4\pi L)^d} \sum_{i=1}^d \bracket{ \bar z_i-\bar w_i\over 4L} e^{-|z-w|^2/4L}\prod_{j=1}^d(\d\bar z_j-\d\bar w_j) \prod_{j\neq i} (\pa_{z_j}-\pa_{w_j}),
$$
and the closed string propagator $P(\epsilon, L)^C$ is the kernel
$$
\int_\epsilon^L \d t (\dbar^*\pa\otimes 1)K_t^C=\int_\epsilon^L \d t {1\over (4\pi t)^d} \sum_{i,j=1}^d \bracket{ \bar z_i-\bar w_i\over 4t} \bracket{ \bar z_j-\bar w_j\over 4t} e^{-|z-w|^2/4t}\prod_{k\neq i}(\d\bar z_k-\d\bar w_k) \prod_{k\neq j} (\pa_{z_k}-\pa_{w_k}).
$$
Similarly, the BV operator for the open string $\Delta_L^O$ is associated to the kernel
$$
   K_L^O={1\over (4\pi L)^d} e^{-|z-w|^2/4L}\prod_{i=1}^d(\d\bar z_i-\d\bar w_i)\otimes Id_{\gl_N}  \in \Omega^{0,*}(X\times X)\otimes (\gl_N\otimes \gl_N),
$$
where $Id_{\gl_N}=\sum_{i,j=1}^N e_{ij}\otimes e_{ji}\in \gl_N\otimes \gl_N$, $e_{ij}$ the matrix with entry $1$ at the $i$-th row and $j$-th column and zero otherwise. The open string propagator $P(\epsilon, L)^O$ is the kernel
$$
 \int_\epsilon^L \d t (\dbar^*\otimes 1) K_L^O=\int_\epsilon^L \d t {1\over (4\pi t)^d} \sum_{i=1}^d \bracket{ \bar z_i-\bar w_i\over 4t} e^{-|z-w|^2/4t}\prod_{j\neq i}(\d\bar z_j-\d\bar w_j)\otimes Id_{\gl_N}.
$$

In the previous section, we have reduced the construction of quantization of open-closed BCOV theory to $I_{0,2,0}[L]$, the annulus functional with two boundaries and without closed string inputs.    We need to do two things: first, show that any possible anomalies to constructing $I_{0,2,0}$ (satisfying the master equation) vanish. And secondly, show that there is no ambiguity to constructing $I_{0,2,0}[L]$. 

The first thing we need to do is to compute the cohomology groups controlling quantization to the annulus level.  Let $V_{g,h,n}$ refer to the obstruction-deformation group controlling possible terms in the action functional of type $(g,h,n)$.  Thus, $V_{g,h,n}$ is built from certain holomorphically translation-invariant functionals of weight $h$, with $n$ closed string inputs, which are $SU(d)$-invariant and have the correct scaling weight under $S^1 \subset U(d)$. Since $I_{0,2,0}$ has cohomology degree $0$, we see that at the annulus level, $H^1(V_{0,2,0})$ describes obstructions and $H^0(V_{0,2,0})$ describes deformations.  
\begin{proposition}
\label{proposition_annulus_cohomology}
We have $H^i(V_{0,2,0}) = 0 $ if $i \le 0$, and $H^1( V_{0,2,0})$ is of dimension $(d-1)/2$.   The possible anomalies are the functionals of $A \in \Omega^{0,\ast}_c(\C^d,\gl_N)$ 
$$
A \mapsto \int \op{Tr} \left(A (\pa{A})^{r} \right) \op{Tr} \left( (\pa{A})^{d-r} \right) 
$$
for $1 \le r \le (d-1)/2$. 
\end{proposition}
\begin{proof}
  The relevant obstruction-deformation complex is the holomorphically translation-invariant local functionals just of the open-string fields, which are admissible of weight $2$. The differential is $\dbar+\{I_{0,1,0},-\}^O$. 
  
Let us first describe the complex of holomorphically translation-invariant local functionals for the open string sector.  The space of fields in this case is just $\Omega^{0,\ast}(\C^d,\gl_N)[1]$.  The space of jets at the origin of fields is 
$$
J_0 (\Omega^{0,\ast}(\C^d,\gl_N) ) [1] = \C[[z_i, \zbar_j, \d \zbar_k]] \otimes \gl_N[1]. 
$$
Using the discussion in the proof of theorem \ref{theorem-vanishing-flat-space}, we find that for the $\gl_N$ holomorphic Chern-Simons gauge theory the complex of holomorphically translation-invariant local functionals is 
\begin{equation*}
\Omega^{d,0}_0 \otimes_{\C\left[\partial_{z_i} \right]}^{\mbb{L}} C^\ast_{red} (\C[[z_i, \zbar_j, \d \zbar_k]] \otimes \gl_N )[d] \tag{$\dagger$} 
\end{equation*}
where $C^\ast_{red}$ indicates the reduced Chevalley-Eilenberg cochain complex of the dg Lie algebra $\C[[z_i,\zbar_j, \d \zbar_k]] \otimes \gl_N$. The differential of the dg Lie algebra is the jet of the $\dbar$ operator which sends $\zbar_j \to \d \zbar_j$.

Since $\C[[z_i]] \otimes \gl_N$ is a quasi-isomorphic sub dg Lie algebra which is preserved by the operators $\dpa{z_i}$, we find that a complex quasi-isomorphic to ($\dagger$) is
\begin{equation*}
\Omega^{d,0}_0 \otimes_{\C\left[\partial_{z_i} \right]}^{\mbb{L}} C^\ast_{red}( \C[[z_i]] \otimes \gl_N )[d] \tag{$\ddagger$}.  
\end{equation*}
Since we are looking at admissible functionals (as $N$ varies), the Tsygan-Loday-Quillen theorem tells us that we can replace Lie algebra cochains by the symmetric algebra of cyclic cochains of $\C[[z_i]]$.  Let $\Omega^{-\ast}_{\what{D}^d}$ denote the algebra of differential forms on the formal $d$-disc, graded so that $i$-forms are in degree $-i$.    The Hochschild-Kostant-Rosenberg theorem tells us that the cyclic cohomology of $\C[[z_i]]$ is the dual of the complex $\Omega^{-\ast}(\what{D}^d)[t^{-1}]$, with differential $t \d_{dR}$. Here $t$ is a parameter of degree $2$. 

We thus find that the complex of all admissible holomorphically translation-invariant functionals is 
\begin{equation*}
\Omega^{d,0}_0 \otimes_{\C\left[\partial_{z_i} \right]}^{\mbb{L}} \Oo( \Omega^{-\ast}_{\what{D}^d}[t^{-1}][1])[d] 
\end{equation*}
where $\Oo$ indicates functionals modulo constants.  

Our obstruction-deformation complex is the subcomplex of weight $2$, namely 
$$
\Omega^{d,0}_0 \otimes_{\C\left[\partial_{z_i} \right]}^{\mbb{L}} \Sym^2 \left(\Omega^{-\ast}_{\what{D}^d} [t^{-1}][1]  \right)^\vee[d].   
$$

The linear dual of the space $\C[[z_i]]$ is the algebra $D_0 = \C[\dpa{z_i}]$ of constant-coefficient holomorphic differential operators.  It follows that the linear dual of $\Omega^{-\ast}_{\what{D}^d}$ is $D_0[\alpha_i]$ where the variables $\alpha_i$ have weight $+1$. Therefore we can write our obstruction-deformation complex as
$$
\Omega^{d,0}_0 \otimes_{D_0}^{\mbb L} \Sym^2 \left( D_0[[\alpha_i,t]][-1] \right)[d]. 
$$
Since tensor powers of $D_0$ are flat as $D_0$-modules we can replace the derived tensor product by an actual tensor product.  Let us also replace the symmetric square by the tensor square for now; we will take the invariants for the group $S_2$ at the end. 

There is a natural isomorphism
$$
\Omega^{d,0}_0 \otimes_{D_0} \left(D_0 \otimes D_0\right) \iso \Omega^{d,0}_0 \otimes_{\C} D_0. 
$$
Thus, we find that our obstruction complex is the $S_2$-invariants in the complex
\begin{equation}
\Omega^{d,0}_0 \otimes_{\C} \left(\C[[\alpha_i,t]] \otimes_{\C} D_0[[\alpha_i,t]] \right)[d-2]. \label{equation_sym2}
\end{equation}
This complex has a differential arising from that on the $D_0$-module $D_0[[\alpha_i,t]]$ dual to $\Omega^{-\ast}_{\what{D}^d} [t^{-1}]$. We will ignore this differential for now, and just describe the elements which are $U(d)$-invariant (which, of course, is what we're interested in). 

Under the $U(d)$-action, the odd variables $\alpha_i$ transform like $\dpa{z_i}$, so that the exterior algebra $\C[\alpha_i]$ is the constant poly-vector fields.  Thus, under the $S^1 \subset U(d)$ which scales $\C^d$, the variables $\alpha_i$ have weight $-1$ and the variables $\dpa{z_i}$ generating $D_0 = \C[\dpa{z_i}]$ also have weight $-1$.  The basis $\d z_1  \dots \d z_d$ of $\Omega^{d,0}_0$ has weight $d$. Therefore, the total number of $\alpha_i$ and $\dpa{z_j}$ that must appear is $d$.  

Invariant theory for $SU(d)$ tells us that the space of invariants in the $d$th tensor power of the defining representation $\C^d$ is of rank one, and spanned by the top exterior power.  Therefore, to be invariant under $SU(d)$, an element in $\C[[\alpha_i]]^{\otimes 2} \otimes \C[\dpa{z_i}]$ of weight $d$ must be totally antisymmetric. This means that at most one $\dpa{z_i}$ can appear, because the representation $\C[\dpa{z_i}]$ of $SU(d)$ is $\Sym^\ast \C^d$ and $SU(d)$ invariants in $\Sym^k \C^d \otimes (\C^d)^{d-k}$ is zero if $2 \le k \le d$.

We can now describe the $U(d)$ invariants in the complex in equation \ref{equation_sym2}. For any subset $I$ of the set $\{1,\dots,d\}$ let $\alpha_I$ denote the product $\prod_{i \in I} \alpha_i$, and let $\alpha_{I^c}$ denote the corresponding object associated to the complement $I^c$ of $I$.  A basis for the $U(d)$ invariants is given by the elements  
\begin{align*} 
\eta_{r,k_1,k_2} &= \sum_{I \subset \{1,\dots,d\}, \text{ with }  \abs{I} = r} \pm \d Vol \otimes \alpha_I t^{k_1}  \otimes \alpha_{I^c} t^{k^2}\\ 
\mu_{r,k_1,k_2} &= \sum_{\substack{I \subset \{1,\dots,d\} \text{ with } \abs{I} = r\\ i\in I}} \pm \d Vol \otimes \alpha_{I \setminus i} t^{k_1} \otimes \alpha_{I^c} t^{k_2} \dpa{z_i}.
 \end{align*}
The elements $\eta_{r,k_1,k_2}$ are of degree $2+2k_1 + 2k_2$, and the elements $\mu_{r_1,k_1,k_2}$ are of degree $1+2k_1 + 2k_2$.  The differential is of the form
$$
\d ( \eta_{r,k_1,k_2} ) = \mu_{r,k_1,k_2+1} \pm \mu_{r,k_1+1,k_2}. 
$$
We are only interested in the cohomology in degrees $\le 1$. In degrees $\le 0$, there are no $U(d)$-invariant elements, and in degree $1$ the $U(d)$-invariant elements are $\mu_{r,0,0}$ for $0 \le r \le d$.  The elements $\mu_{r,0,0}$ are closed, and the $S_2$-symmetry (which we have been neglecting) sends $\mu_{r,0,0} \mapsto \pm \mu_{d-r,0,0}$. It follows that the $\Z/2$ and $U(d)$-invariant cohomology in degree $1$ is of dimension $(d-1)/2$, as desired. It is direct to check that these elements corresponds to the functionals described in the proposition. 
\end{proof}

\subsection{Cancelling of the anomaly}
What we have  shown so far is that there is a possible anomaly to quantizing the theory to the annulus level, but that if the anomaly vanishes the quantization is unique.  

 In this subsection, we compute directly that the obstruction to constructing  $I_{0,2,0}[L]$ satisfying the relevant master equations vanishes.

We need to recall some details of the construction of quantum field theories in \cite{Cos11}, and some combinatorial facts about how to write the RG flow in terms of graphs.

We say a \emph{labelled graph} is a graph $\Gamma$ with the following extra structures:
\begin{enumerate}
\item Every edge (internal or external) is labelled by $c$ or $o$, indicating whether it is a closed-string or open-string graph.
\item Every vertex $v$ is labelled by numbers $(g,h,n)$. Also, $n$ is the number of germs of closed-string edges incident to the vertex $v$. The set of germs of open-string edges incident to $v$ is partitioned into $h$ non-empty subsets, and each of these subsets is given a cyclic order. 
\end{enumerate}
We should imagine the vertex $v$ as being associated to the topological type of Riemann surface of genus $g$ with $h$ boundary components and $n$ marked points in the interior. Each of the boundary components has a certain number of marked points. By gluing together these Riemann surfaces according to how they are marked, one finds a new topological type of a Riemann surface which in turn has some genus, number of boundary components, and number of interior marked points. This topological data is labelled by numbers $g(\Gamma)$, $h(\Gamma)$ and $n(\Gamma)$ which are determined by the combinatorics of $\Gamma$. We call $(g(\Gamma),h(\Gamma),n(\Gamma))$ the \emph{type} of $\Gamma$. 

For instance, if all vertices of $\Gamma$ are discs and there are no closed-string edges, then $\Gamma$ is a ribbon graph and $g(\Gamma)$ and $h(\Gamma)$ are obtained by the standard combinatorial algorithm for thickening a ribbon graph into a topological surface.  

Suppose that
$$
I \in \Oo_{adm}(\E_{c,\infty})[[\lambda]]_+
$$
is a series 
$$
I = \sum I_{g,h,n} \lambda^{2g+h}
$$
where $I_{g,h,n}$ is admissible of weight $h$ as a function of the open string fields and homogeneous of weight $n$ as a function of closed string fields.

Then for every labelled graph $\Gamma$ we can construct
$$
W_{\Gamma}(P(\eps,L)^C + P(\eps,L)^{O,\infty},I) \in \Oo_{adm}(\E_{c,\infty})[[\lambda]]_+
$$
by placing, on each internal open-string edge, the open-string propagator $P^0(\eps,L)$; on each internal closed-string edge the closed string propagator $P^C(\eps,L)$; and on each vertex of type $(g,h,n)$ the interaction $I_{g,h,n}$.  The external lines are labelled with open or closed string fields to produce an admissible function of the fields.   

Also, as explained in section \ref{quantum-BCOV}, we can apply the open-closed RG flow
$$
W_\eps^L( I ) = \lambda^2 \log\left\{ \exp \left(\lambda \pa_{P^O(\eps,L)} + \lambda^2 \pa_{P^C(\eps,L)}\right) \exp \left( I/\lambda^2 \right)\right\}. 
$$
This RG flow map can be expanded as a sum of graphs:
$$
W_\eps^L( I) = \sum_{\Gamma \text{ connected labelled graphs} }\lambda^{2 g(\Gamma) + h(\Gamma)}  \frac{1}{\abs{\Aut(\Gamma)} } W_{\Gamma}(P(\eps,L)^C + P(\eps,L)^{O,\infty}, I). 
$$

\subsection{Constructing the QFT explicitly}
Let us now recall a little about how to construct QFTs using the technology of \cite{Cos11}.  We will explain how one constructs our BCOV theory up to the annulus level. We start with our classical action, which in this case is
$$I^{classical} = \lambda I_{0,1,0} + \lambda I_{0,1,1}.$$
One constructs $I_{0,1,0}[L]$ and $I_{0,1,1}[L]$ by
$$
I_{0,1,0}[L] = \lim_{\eps \to 0} \sum_{\Gamma \text{ of type } (0,1,0)} W_{\Gamma} (P(\eps,L)^C + P(\eps,L)^{O,\infty}, I^{classical} ). 
$$
Since every graph appearing in this sum is a tree, this limit exists.  One constructs $I_{0,1,1}[L]$ by the same kind of formula, and again every graph in the sum is a tree.  It is automatic that the quantum master equation holds for $I_{0,1,0}[L]$ and $I_{0,1,1}[L]$ modulo higher-order terms.

Next, we can try to construct $I_{0,2,0}[L]$.  The idea is that we first define
$$
  I^{naive}_{0,2,0}[L]=\lim_{\epsilon\to 0}\sum_{\Gamma\text{ of type }(0,2,0)} W_\Gamma(P(\epsilon, L)^C+P(\epsilon, L)^{O, \infty}, \lambda I_{0,1,0}+ \lambda I_{0,1,1} +\lambda^2 I_{0,2,0}^{CT}(\epsilon)), 
$$
where $I^{CT}_{0,2,0}(\epsilon)$ is a counter-term introduced to make the $\eps \to 0$ limit exist. As explained in \cite{Cos11}, the functional $I^{CT}_{0,2,0}(\eps)$ is local, and is uniquely determined by the requirements that the $\eps \to 0$ limit exists and that, as a functional of $\eps$, $I_{0,2,0}^{CT}(\eps)$ is ''purely singular''. (There is, generally, a choice involved in what it means for a function of $\eps$ to be purely singular, this choice is called a renormalization scheme in \cite{Cos11}).

One then checks whether or not $I^{naive}_{0,2,0}[L]$ satisfies the quantum master equation modulo higher order terms.  If it does, then we are done. If not, the failure to satisfy the QME will be given by some obstruction $O_{0,2,0}[L]$. The limit $\lim_{L \to 0} O_{0,2,0}[L]$ exists, and is a local functional which we call $O_{0,2,0}$. The element $O_{0,2,0}$ will be a closed element of degree $1$ in the obstruction-deformation complex (built from local functionals) controlling quantizations of our theory. 

If the cohomology class of $O_{0,2,0}$ is zero, so that $O_{0,2,0} = \d J_{0,2,0}$ for some local functional $J_{0,2,0}$, then we set
$$
I_{0,2,0}[L] =\lim_{\epsilon\to 0}\sum_{\Gamma\text{ of type }(0,2,0)} W_\Gamma(P(\epsilon, L)^C+P(\epsilon, L)^{O, \infty}, \lambda I_{0,1,0}+ \lambda I_{0,1,1} +\lambda^2 I_{0,2,0}^{CT}(\epsilon)- \lambda^2 J_{0,2,0}). 
$$
The fact that $\d J_{0,2,0} = O_{0,2,0}$ implies that $I_{0,2,0}[L]$ now satisfies the master equation, and we can proceed. 

On the other hand, if the cohomology class of $O_{0,2,0}$ is non-zero, then we can not proceed any further. 

Our goal in this section is to prove the following.
\begin{theorem*}
The counter-terms $I_{0,2,0}^{CT}(\eps)$ vanish, and $I^{naive}_{0,2,0}[L]$ satisfies the quantum master equation. 
\end{theorem*}
This will tell us that we can construct the theory uniquely on $\C^d$ at the annulus level, and so proceed to construct the full theory uniquely by the cohomology cancellation argument.

In fact, what we find is that if we just consider the purely open theory, there is an annulus anomaly, meaning that pure holomorphic Chern-Simons theory does not exist. However, the open string anomaly is cancelled by a contribution from the closed string sector.

\subsubsection{Vanishing of counter term}
\begin{lemma} \label{lemma-finite}The following limit exists 
$$
  \lim_{\epsilon\to 0}\sum_{\Gamma: \text{annulus}} W_\Gamma(P(\epsilon, L)^C+P(\epsilon, L)^{O, \infty}, I_{0,1,0}+I_{0,1,1}).
$$
Therefore we can choose $I^{CT}(\epsilon)=0$.
\end{lemma}

\begin{proof} There are two types of annulus diagrams. The first one contains a closed string propagator connecting two $I_{0,1,1}$

\includegraphics[width=3in]{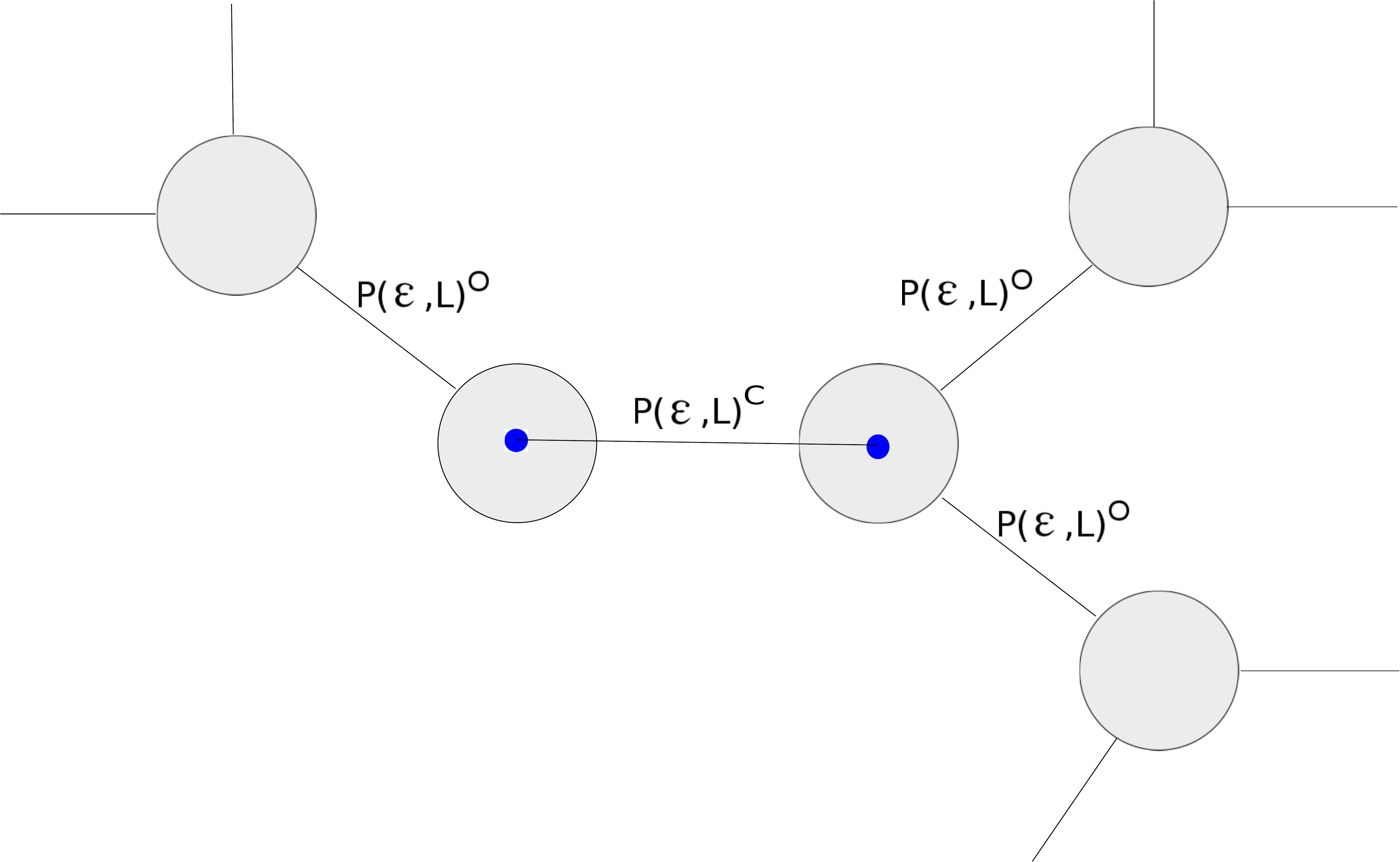}

\noindent Since there is no loop in the diagram, the limit under $\epsilon\to 0$ exists.

The second type is the one-loop diagram from open string sector.

\includegraphics[width=4in]{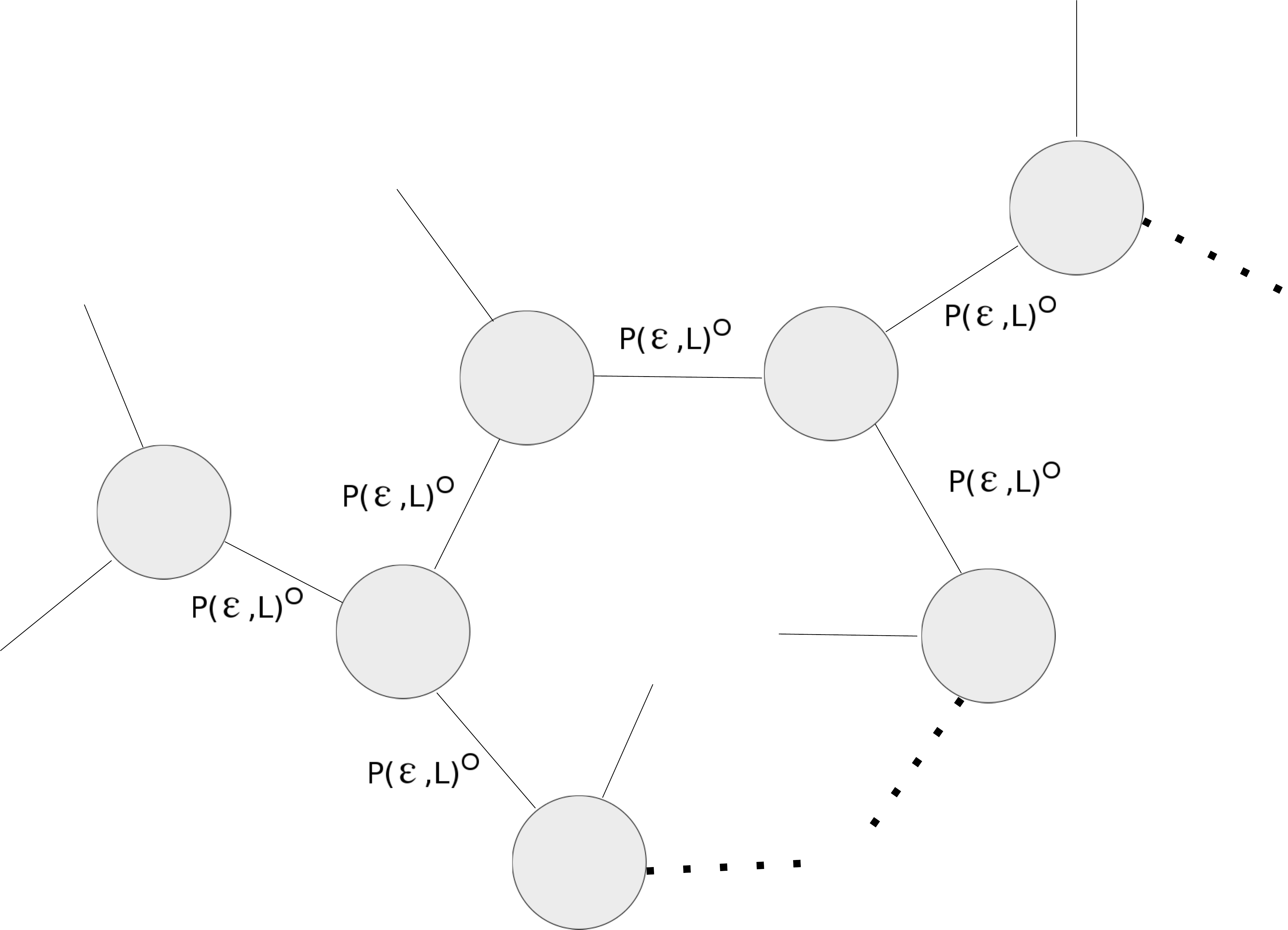}

\noindent To show that the counter-terms from the open-string sector vanish, it suffices to show this for diagrams which are wheels, because every one-loop diagram is a wheel with some trees grafted on and trees can not contribute any singularities. Recall also that our Feynman diagrams are ribbon graphs.  However, to show that the counter-terms vanish, the ribbon structure on the graph does not play any role, so we will ignore this for now. 

Consider the wheel diagram $\Gamma_m$ with $m+1$ vertices (and with any of ribbon graph structures)

\includegraphics[width=3in]{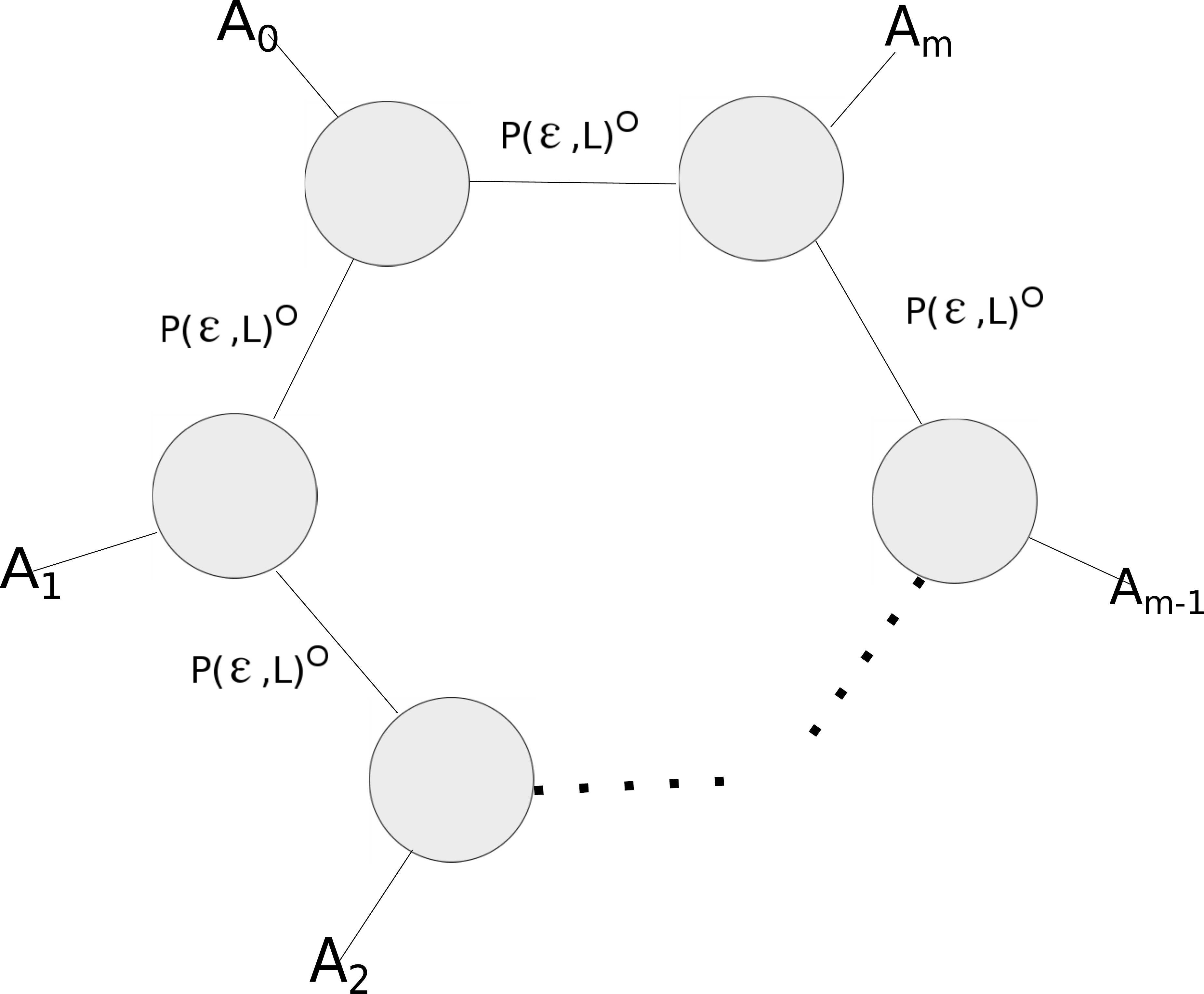}

Its graph integral can be represented by
$$
   W_{\Gamma_m}(P(\epsilon, L)^O, A_0, \cdots, A_m)={ \int_{\mathbb C^d\times\cdots\times \mathbb C^d} \prod_{\alpha=0}^mA_\alpha\wedge \d^dz^{(\alpha)}\prod_{\alpha=0}^m P(\epsilon, L)^O_{\alpha, \alpha+1}}
$$
where $z^{(\alpha)}=\{z^{(\alpha)}_i\}_{1\leq i\leq d}$ are copies of holomorphic coordinates on $\mathbb C^d$, $\d^d z^{(\alpha)}=\prod\limits_{i=1}^d \d z_i^{(\alpha)}$ is the corresponding holomorphic $n$-form. $A_\alpha\in \Omega_c^{0,*}(\mathbb C^n)$ lies on the copy of $\mathbb C^d$ with coordinate $z^{(\alpha)}$.  $P(\epsilon, L)^O_{\alpha, \alpha+1}$ is the regularized open sector propagator connecting the vertices $z^{(\alpha)}, z^{(\alpha+1)}$, where we have identified $z^{(m+1)}$ with $z^{(0)}$.

\begin{claim}
  $W_{\Gamma_m}(P(\epsilon, L)^O, A_0, \cdots, A_m)=0$ if $m<d$.
\end{claim}
In fact,  $\prod\limits_{\alpha=0}^m P(\epsilon, L)^O_{\alpha, \alpha+1}=0$ if $m<d-1$ for degree reasons.

To see this, consider the following factor in $\prod\limits_{\alpha=0}^m P(\epsilon, L)^O_{\alpha, \alpha+1}$
$$
    \prod_{\alpha=0}^m \bracket{\sum_{i=1}^d \bracket{\bar z^{(\alpha)}_i-\bar z^{(\alpha+1)}_i}\prod_{j\neq i}\bracket{\d\bar z^{(\alpha)}_j-\d\bar z^{(\alpha+1)}_j}}.
$$
If we change coordinates $w^{(0)}=z^{(0)}, w^{(\alpha)}=z^{(\alpha)}-z^{(\alpha-1)}, 1\leq \alpha\leq m$, then up to a constant, it becomes
$$
   \bracket{\bracket{\sum_{\alpha=1}^m V^{(\alpha)}}\vdash\prod_{i=1}^d\bracket{\sum_{\alpha=1}^m \d\bar w_i^{(\alpha)}}}\prod_{\alpha=1}^m\bracket{V^{(\alpha)}\vdash \prod_{i=1}^d \d\bar w_i^{(\alpha)}}
$$
where $V^{(\alpha)}=\sum\limits_{i=1}^d \bar w_i^{(\alpha)}{\pa\over \pa \bar w^{(\alpha)}_i}$, and $\vdash$ is the contraction. This contains $(m+1)(d-1)$ $\d\bar w_i^{(\alpha)}$'s, which vanishes for $m<d-1$ because it is a form of type $(0,(m+1)(d-1))$ on $\C^{dm}$ and $(m+1)(d-1) > dm$. 

If $m=d-1$, it is a top $(0,d(d-1))$-form. We can move $V^{(\alpha)}$ to the left and
\begin{align*}
  &\bracket{\bracket{\sum_{\alpha=1}^m V^{(\alpha)}}\vdash\prod_{i=1}^d\bracket{\sum_{\alpha=1}^m \d\bar w_i^{(\alpha)}}}\prod_{\alpha=1}^m\bracket{V^{(\alpha)}\vdash \prod_{i=1}^d \d\bar w_i^{(\alpha)}}\\
  =&  \pm\bracket{V^{(1)}\vdash V^{(2)}\vdash\cdots \vdash V^{(d-1)}\vdash \bracket{\sum_{\alpha=1}^{d-1} V^{(\alpha)}}\vdash\prod_{i=1}^d\bracket{\sum_{\alpha=1}^{d-1} \d\bar w_i^{(\alpha)}}}\prod_{\alpha=1}^{d-1}\bracket{\prod_{i=1}^d \d\bar w_i^{(\alpha)}}
  =0.
\end{align*}
This proves the claim.

\begin{claim} $\lim\limits_{\epsilon\to 0}W_{\Gamma_m}(P(\epsilon, L)^O, A_0, \cdots, A_m)$ exists if $m\geq d$.
\end{claim}
We use the coordinates $w^{(\alpha)}$ above. Let $\d^{2d}w^{(\alpha)}=\prod_{i=1}^d \d w^{(\alpha)}_i \d\bar w^{(\alpha)}_i$. The graph integral can be written as
\begin{align*}
&W_{\Gamma_m}(P(\epsilon, L)^O, A_0, \cdots, A_m)\\
=&\int_{\mathbb C^d}\d^{2d}w^{(0)} \int_{\bracket{\mathbb C^d}^m}\bracket{\prod_{\alpha=1}^{m}\d^{2d}w^{(\alpha)}}\prod_{\alpha=1}^{m+1} \int_\epsilon^L {\d t_\alpha\over (4\pi t_\alpha)^d} \\& \sum_{i_1, \cdots, i_m=1}^d {\bar w_{i_1}^{(1)}\over t_1} \cdots {\bar w^{(m)}_{i_m}\over t_m}{\sum\limits_{\alpha=1}^m \bar w^{(\alpha)}_{i_{m+1}}\over t_{m+1}}e^{-\sum\limits_{\alpha=1}^m|w^{(\alpha)}|^2/4t_\alpha- |\sum\limits_{\alpha=1}^m w^{(\alpha)}|^2/4t_{m+1}}\Phi^{i_1\cdots i_{m+1}}
\end{align*}
where $\Phi^{i_1\cdots i_{m+1}}$'s are compactly supported functions on $\bracket{\mathbb C^d}^{m}$. Using integration by parts (or Wick's theorem), the ${\bar w^{(\alpha)}\over t_\alpha}$ factor becomes holomorphic derivative acting on $\Phi$ whose coefficients are uniformly bounded for $t_\alpha$'s (see \cite{Li11b}). For example,
$$
   \Phi^{i_1\cdots i_{m+1}} {\bar w^{(1)}_{i_1}\over 4t_1}\to\bracket{\pa_{w^{(1)}_{i_1}}-\sum\limits_{\alpha=1}^m{t_{\alpha}\pa_{w^{(\alpha)}_{i_1}}\over t_1+\cdots +t_{m+1}}} \Phi^{i_1\cdots i_{m+1}}.
$$
Therefore to prove the existence of limit $\epsilon\to 0$, we only need to prove the convergence of the integral
$$
\int_{\bracket{\mathbb C^d}^m}\bracket{\prod_{\alpha=1}^{m}\d^{2d}w^{(\alpha)}}\prod_{\alpha=1}^{m+1} \int_0^L {\d t_\alpha\over (4\pi t_\alpha)^d} e^{-\sum\limits_{\alpha=1}^m|w^{(\alpha)}|^2/4t_\alpha- |\sum\limits_{\alpha=1}^m w^{(\alpha)}|^2/4t_{m+1}}.
$$
Performing the Gaussian integral on $w^{(\alpha)}$'s, we find (up to a constant)
\begin{align*}
      &\prod_{\alpha=1}^{m+1}\int_0^L {\d t_\alpha\over (4\pi t_\alpha)^d}{1\over \bracket{t_1^{-1}\cdots t_{m+1}^{-1}\sum\limits_{\alpha=1}^{m+1}{ t_{\alpha}}}^d}\\
      =&\prod_{\alpha=1}^{m+1}\int_0^L {\d t_\alpha\over \bracket{4\pi \sum\limits_{\alpha=1}^{m+1}t_\alpha}^d}\leq C \prod_{\alpha=1}^{m+1}\int_0^L \d t_\alpha {1\over \bracket{t_1\cdots t_{m+1}}^{d/(m+1)}}<\infty
\end{align*}
since ${d/(m+1)}<1$. This proves the claim.

The lemma follows from the above two claims.
\end{proof}

\subsubsection{Quantum master equation and anomaly cancellation}
The vanishing of the counter terms implies that we can define the naive quantization for $(g,h,n)=(0,2,0)$ by
$$
    I_{0,2,0}[L]=\lim_{\epsilon\to 0} \sum_{\Gamma: \text{annulus}}W_\Gamma\bracket{P(\epsilon,L)^C+P(\epsilon, L)^O, I_{0,1,0}+I_{0,1,1}}
$$

Let
$$
   O[L]=Q I_{0,2,0}[L]+\fbracket{I_{0,1,0}[L], I_{0,2,0}[L]}_{L}^O+{1\over 2}\fbracket{I_{0,1,1}[L], I_{0,1,1}[L]}_{L}^C+ \Delta_L^{O, \infty}I_{0,1,0}[L].
$$
By \cite{Cos11}, $O[L]$ satisfies a version of classical renormalization group flow equation, and
$$
    O=\lim_{L\to 0}O[L]
$$
exists as a local functional. This is the annulus anomaly for the quantum master equation at $(g,h,n)=(0,2,0)$. The main result in this section is the following
\begin{theorem}\label{annulus-cancellation}
Under a suitable rescaling of $I_{0,1,1}$ by a nonzero constant, the annulus anomaly $O$ vanishes.
\end{theorem}
If follows that $O[L]=0$ and the quantum master equation at $I_{0,2,0}[L]$ holds.

\begin{lemma}
$\lim\limits_{L\to 0}\fbracket{I_{0,1,0}[L], I_{0,2,0}[L]}_{L}^O=0$
\end{lemma}
\begin{proof} This follows from Lemma \ref{lemma-finite} that $\lim\limits_{L\to 0}I_{0,2,0}[L]=0$.
\end{proof}

We will let
$$
 O_{c}[L]={1\over 2}\fbracket{I_{0,1,1}[L], I_{0,1,1}[L]}_{L}^C, \quad O_{o}[L]=QI_{0,2,0}[L]+\Delta_L^{O, \infty}I_{0,1,0}[L]
$$

\begin{lemma}\label{lemma_closed_string_annulus_anomaly}
$O_c=\lim\limits_{L=0}O_{c}[L]$ exists, and is given by
$$
   O_c={1\over 2}\fbracket{I_{0,1,1}, I_{0,1,1}}^C
$$
Explicitly, we have
$$
    O_c(A_0, A_1, \cdots, A_k; B_{k+1},\cdots, B_d)=\pm C_1\int_{\C^d}\bracket{A_0\wedge \pa A_1\wedge\cdots \pa A_k}\wedge \bracket{\pa B_{k+1}\wedge\cdots\wedge \pa B_d}
$$
for two cyclic tensors $\{A_i\}, \{B_j\}$ of $\Omega^{0,*}_c(\mathbb C^d)$. Here $C_1$ is a nonzero constant which only depends on the dimension $d$.
\end{lemma}
\begin{proof} The closed string BV bracket $\{-,-\}^C$ is represented by the kernel of the operator $\pa: \PV(X)\to \PV(X)$, which is $(\partial \otimes 1) \delta_{Diag}$ where $\delta_{Diag}$ is the $\delta$-current on the diagonal of $X \times X$, viewed as a polyvector field via the isomorphism between forms and polyvector fields. In particular, it only involves the component $t^0\PV(X)$ of the full closed string fields $\PV(X)[[t]]$. For $\mu \in \PV(X)$, 
$$
   I_{0,1,1}(\mu, A)=c \int_X \Tr (A\wedge \pa A\wedge\cdots \wedge \pa A)\wedge (\mu\vdash \Omega_X).
$$
It follows that $\{I_{0,1,1}, I_{0,1,1}\}^C$ represents the functional
\begin{align*}
   A\to c\int_X \Tr  (A\wedge \pa A\wedge\cdots \wedge \pa A)\wedge  \Tr(\pa A\wedge\cdots \wedge \pa A).
\end{align*}

The lemma follows immediately. 
\end{proof}

We are left to analyze the local functional
$
    O_{o}=\lim\limits_{L\to 0} O_{o,[L]}.
$
Using the relation $[Q, P(\epsilon, L)^O]=K_\epsilon^O-K_L^O$, we find that $O_o[L]$ is given by the following two types of Feynman diagrams

\includegraphics[width=5in]{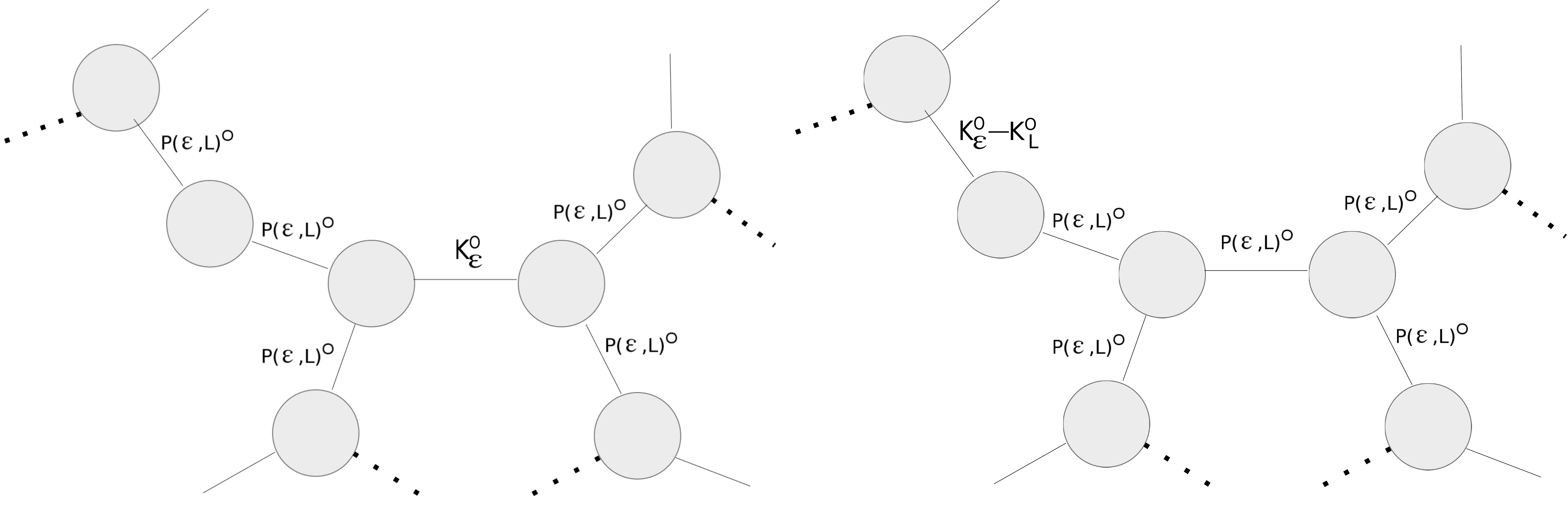}

The second graph vanishes under $L\to 0$. The nontrivial contribution to the local obstruction comes from the following wheel diagram with $m+1$ vertices ($m\geq d$ by degree reasons. See the proof of Lemma \ref{lemma-finite}).

\includegraphics[width=3in]{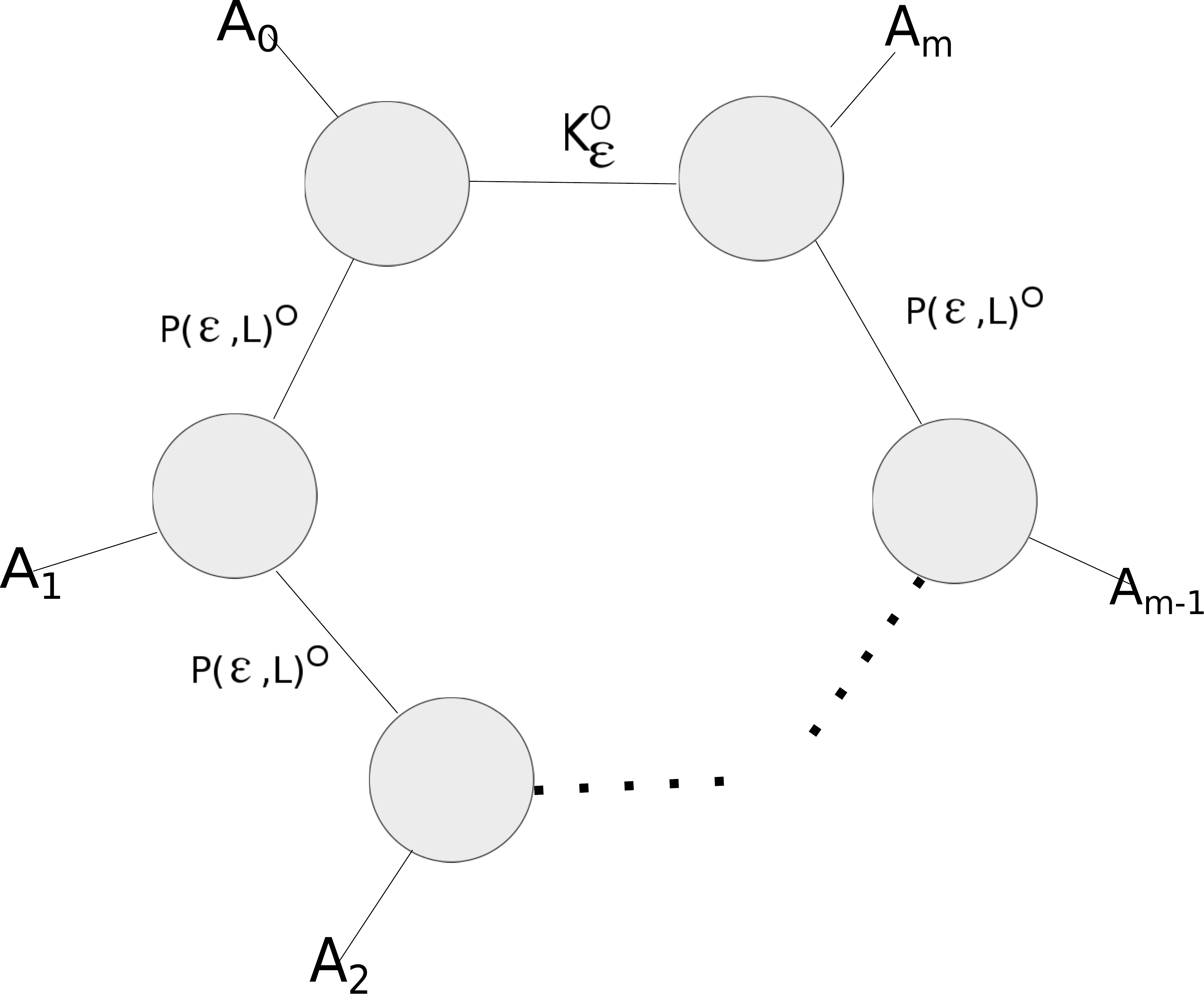}

We put $K_\epsilon^O$ on one edge and $P(\epsilon, L)^O$ on the other edges, and consider the following graph integral
$$
W_{\Gamma_m}(K_\epsilon^O, P(\epsilon, L)^O, A_0, \cdots, A_m)={ \int_{\mathbb (C^d)^{m+1}} \bracket{\prod_{\alpha=0}^mA_\alpha\wedge \d^d z^{(\alpha)}}K_{\epsilon,0,m}^O\prod_{\alpha=1}^m P(\epsilon, L)^O_{\alpha-1, \alpha}}
$$
where we have kept the same notations as in the proof of Lemma \ref{lemma-finite}, and $K_{\epsilon,0,m}^O$ denotes the corresponding BV kernel connecting the vertices $z^{(0)}$ and $z^{(m)}$.

\begin{lemma}
Let $\Phi$ be a smooth function on $(\mathbb C^d)^m$ with compact support, where $m>d$. Let $W(\epsilon, L)$ be the integral
\begin{align*}
   W(\epsilon, L)=&{1\over (4\pi \epsilon)^d}\bracket{\prod_{\alpha=0}^m \int_{\mathbb C^d} \d^{2d}z^{(\alpha)}} \bracket{\prod_{\alpha=1}^m \int_\epsilon^L {\d t_{\alpha}\over (4\pi t_\alpha)^{d}}} \Phi \\ & \prod_{\alpha=1}^m \bracket{{\bar z_{i_\alpha}^{(\alpha)}-\bar z_{i_\alpha}^{(\alpha-1)}\over 4t_\alpha}} e^{-\sum\limits_{\alpha=1}^m |z^{(\alpha)}-z^{(\alpha-1)}|^2/4t_\alpha-|z^{(m)}-z^{(0)}|^2/4\epsilon}
\end{align*}
where $1\leq i_1,\cdots, i_m\leq n$. Then $W(0, L)\equiv \lim\limits_{\epsilon\to 0}W(\epsilon, L)$ exists, and
$$
\lim\limits_{L\to 0}W(0,L)=0.
$$
\end{lemma}
\begin{proof} We consider the change of variables
$$
   w^{(0)}=z^{(0)},  w^{(\alpha)}=z^{(\alpha+1)}-z^{(\alpha)}, 1\leq \alpha\leq m.
$$
Then
\begin{align*}
   W(\epsilon, L)=&{1\over (4\pi \epsilon)^d}\bracket{\prod_{\alpha=0}^m \int_{\mathbb C^d} \d^{2d}w^{(\alpha)}} \bracket{\prod_{\alpha=1}^m \int_\epsilon^L {\d t_{\alpha}\over (4\pi t_\alpha)^{d}}} {\Phi}  \bracket{\prod_{\alpha=1}^m { \bar w^{(\alpha)}_{i_\alpha}\over 4t_\alpha}} e^{-\sum\limits_{\alpha=1}^m |w^{(\alpha)}|^2/4t_\alpha-|w^{(1)}+\cdots+ w^{(m)}|^2/4\epsilon}.
\end{align*}
We define the holomorphic derivative
$$
   \pa^*_{w^{(\alpha)}_{i_\alpha}}\equiv \pa_{w^{(\alpha)}_{i_\alpha}}-{\sum\limits_{\beta=1}^m t^{(\beta)}\pa_{w^{(\beta)}_{i_\beta}}\over \epsilon+t_1+\cdots+t_m}.
$$
Then integration parts gives
\begin{align*}
   W(\epsilon, L)=&{1\over (4\pi \epsilon)^d}\bracket{\prod_{\alpha=0}^m \int_{\mathbb C^d} \d^{2d}w^{(\alpha)}} \bracket{\prod_{\alpha=1}^m \int_\epsilon^L {\d t_{\alpha}\over (4\pi t_\alpha)^{d}}} \bracket{\prod_{\alpha=1}^m \pa^*_{w^{(\alpha)_{i_\alpha}}}\Phi} e^{-\sum\limits_{\alpha=1}^m |w^{(\alpha)}|^2/4t_\alpha-|w^{(1)}+\cdots+ w^{(m)}|^2/4\epsilon}
\end{align*}

 Since the dependence of ${\prod\limits_{\alpha=1}^m \pa^*_{w^{(\alpha)_{i_\alpha}}}\Phi}$ on $t_\alpha$ is uniformly bounded, we have
\begin{align*}
   |W(\epsilon, L)|\leq {C \over (4\pi \epsilon)^d} \bracket{\prod_{\alpha=1}^m \int_{\mathbb C^d} \d^{2d}w^{(\alpha)}} \bracket{\prod_{\alpha=1}^m \int_\epsilon^L {\d t_{\alpha}\over (4\pi t_\alpha)^{d}}} e^{-\sum\limits_{\alpha=1}^m |w^{(\alpha)}|^2/4t_\alpha-|w^{(1)}+\cdots+ w^{(m)}|^2/4\epsilon}
\end{align*}
where $C$ is a constant which depends only on $\Phi$. Performing the Gaussian integral on $w^{(\alpha)}$'s, we find
\begin{align*}
  |W(\epsilon, L)&|\leq {C \over (4\pi \epsilon)^d}\bracket{\prod_{\alpha=1}^m \int_\epsilon^L {\d t_{\alpha}\over (4\pi t_\alpha)^{d}}} {1\over \bracket{{\epsilon+t_1+\cdots+t_m\over \epsilon t_1\cdot t_m}}^d}\\
  &=C \int_{[\epsilon, L]^m} {\d t_1\cdots \d t_m \over (4\pi)^d \bracket{\epsilon+t_1+\cdots+t_m}^d}\\
  &\leq C \int_{[\epsilon, L]^m}{\d t_1\cdots \d t_m \over (4\pi)^n \bracket{t_1\cdots t_m}^{d/m}}
\end{align*}
It follows immediately that $W(0,L)$ exists and
$$
   \lim_{L\to 0}W(0,L)=0
$$
\end{proof}

\begin{corollary}\label{local m>n} If $m>d$, then
$$
  \lim\limits_{L\to 0}\lim\limits_{\epsilon\to 0}W_{\Gamma_m}(K_\epsilon^O, P(\epsilon, L)^O, A_0, \cdots, A_m)=0
$$
\end{corollary}

This implies that the local functional $O_o$ has precisely $d+1$ inputs. Now we compute the case for $m=d$.

\begin{lemma}\label{local n=m}
When $m=d$, we have
$$
 \lim_{\epsilon\to 0}W_{\Gamma_m}(K_\epsilon^O, P(\epsilon, L)^O, A_0, \cdots, A_m)=C\int_{\mathbb C^d} A_0\wedge \pa A_1\wedge\cdots\wedge \pa A_d 
$$
where $C$ is a non-zero constant which only depends on $d$.
In particular, it doesn't depend on $L$ and survives in the $L\to 0$ limit.
\end{lemma}
\begin{proof} Recall that
$$
W_{\Gamma_m}(K_\epsilon^O, P(\epsilon, L)^O, A_0, \cdots, A_d)={ \int_{(\mathbb C^d)^{d+1}} \bracket{\prod_{\alpha=0}^d A_\alpha\wedge \d^dz^{(\alpha)}}K_{\epsilon,0,d}^O\prod_{\alpha=1}^d P(\epsilon, L)^O_{\alpha-1, \alpha}}.
$$
Again, we change variables
$$
   w^{(0)}=z^{(0)},  w^{(\alpha)}=z^{(\alpha+1)}-z^{(\alpha)}, 1\leq \alpha\leq d.
$$
Then it is easy to see that
\begin{align*}
K_{\epsilon,0,d}^O\prod_{\alpha=1}^d P(\epsilon, L)^O_{\alpha-1, \alpha}=& \pm {1\over (4\pi\epsilon)^d}\prod_{\alpha=1}^d \int_{\epsilon}^L {\d t_\alpha\over (4\pi t_\alpha)^d}
\\ &\sum_{i_1,\cdots,i_d=1}^n \epsilon_{i_1,\cdots,i_d}\bracket{\prod_{\alpha=1}^d {\bar w^{(\alpha)}_{i_\alpha}\over 4t_\alpha}}e^{-\sum\limits_{\alpha=1}^d |w^{(\alpha)}|^2/4t_\alpha-|w^{(1)}+\cdots+ w^{(d)}|^2/4\epsilon} \prod_{i,\alpha=1}^d \d^d\bar w^{(\alpha)}_i
\end{align*}
where $\epsilon_{i_1,\cdots,i_n}$ is the totally anti-symmetric tensor. In particular, it is a top anti-holomorphic form in $w_1,\cdots, w_n$, and we can replace the form index in $A_{\alpha}$ to $w_0$
$$
  A_\alpha\equiv A_\alpha(z_\alpha, \bar z_\alpha)_I \bracket{\d \bar z^{(\alpha)}}^I\to A_\alpha(z_\alpha, \bar z_\alpha)_I \bracket{\d\bar w^{(0)}}^I.
$$
 $W_{\Gamma_m}(K_\epsilon^O, P(\epsilon, L)^O, A_0, \cdots, A_d)$ can be written in the form (simply denoted by $W(\epsilon, L)$)
\begin{align*}
   W(\epsilon, L)=&{1\over (4\pi\epsilon)^d}\bracket{\prod_{\alpha=0}^m \int_{\mathbb C^d} \d^{2d}w^{(\alpha)}} \bracket{\prod_{\alpha=1}^m \int_\epsilon^L {\d t_{\alpha}\over (4\pi t_\alpha)^{d}}} \Phi \\ &\sum_{i_1,\cdots,i_d}\epsilon_{i_1,\cdots,i_d} \bracket{\prod_{\alpha=1}^d { \bar w^{(\alpha)}_{i_\alpha}\over 4t_\alpha}} e^{-\sum\limits_{\alpha=1}^d |w^{(\alpha)}|^2/4t_\alpha-|w^{(1)}+\cdots+ w^{(d)}|^2/4\epsilon}
\end{align*}
where $\Phi$ is a compactly supported smooth function on $\bracket{\mathbb C^d}^{d+1}$ arising from $A_\alpha$'s. Let us introduce the matrix $M_{\alpha\beta}$
$$
   \sum\limits_{\alpha\beta=1}^d M_{\alpha\beta}w^{(\alpha)}\cdot\bar w^{(\beta)}=\sum\limits_{\alpha=1}^d |w^{(\alpha)}|^2/t_\alpha+|w^{(1)}+\cdots+ w^{(n)}|^2/\epsilon
$$
where $w^{(\alpha)}\cdot\bar w^{(\beta)}=\sum\limits_{i=1}^n w^{(\alpha)}_i\bar w^{(\beta)}_i$. $M$ depends on $t_\alpha$ and $\epsilon$. We still denote the holomorphic derivative
$$
   \pa^*_{w^{(\alpha)}_{i_\alpha}}\equiv {\sum\limits_{\beta=1}^d M^{-1}_{\alpha\beta}\pa_{w^{(\beta)}_{i_\beta}}\over t_\alpha} =\pa_{w^{(\alpha)}_{i_\alpha}}-{\sum\limits_{\beta=1}^d t^{(\beta)}\pa_{w^{(\beta)}_{i_\beta}}\over \epsilon+t_1+\cdots+t_d}.
$$
where $M^{-1}$ is the inverse matrix of $M$. Then integration by parts gives
\begin{align*}
   W(\epsilon, L)=&{1\over (4\pi\epsilon)^d}\bracket{\prod_{\alpha=0}^m \int_{\mathbb C^d} \d^{2d}w^{(\alpha)}} \bracket{\prod_{\alpha=1}^m \int_\epsilon^L {\d t_{\alpha}\over (4\pi t_\alpha)^{d}}} \\ &\sum_{i_1,\cdots,i_d}\epsilon_{i_1,\cdots,i_d} \bracket{\prod_{\alpha=1}^d \pa^*_{w^{(\alpha)}_{i_\alpha}} \Phi} e^{-\sum\limits_{\alpha=1}^d |w^{(\alpha)}|^2/4t_\alpha-|w^{(1)}+\cdots+ w^{(n)}|^2/4\epsilon}.
\end{align*}

To compute the $\epsilon\to 0$ limit, we introduce a cut-off function
$$
    \chi_\delta= \rho\bracket{\sum\limits_{\alpha=1}^n |w^{(\alpha)}|^2}
$$
where $\rho(x)$ is a smooth non-negative function for $x\geq 0$ such that
$$
  \rho(x)=1, \mbox{if}\ 0\leq x\leq \delta, \quad \rho(x)=0, \mbox{if}\ x\geq 2\delta
$$
where $\delta$ is a small enough positive number. It is easy to see that
\begin{align*}
\lim_{L\to 0}   \lim_{\epsilon\to 0}W(\epsilon, L)=&\lim_{L\to 0}\lim_{\epsilon\to 0}{1\over (4\pi\epsilon)^d}\bracket{\prod_{\alpha=0}^d \int_{\mathbb C^d} \d^{2d}w^{(\alpha)}} \bracket{\prod_{\alpha=1}^d \int_\epsilon^L {\d t_{\alpha}\over (4\pi t_\alpha)^{d}}} \chi_\delta \\ &\sum_{i_1,\cdots,i_d}\epsilon_{i_1,\cdots,i_d} \bracket{\prod_{\alpha=1}^d \pa^*_{w^{(\alpha)}_{i_\alpha}} \Phi} e^{-\sum\limits_{\alpha=1}^d |w^{(\alpha)}|^2/4t_\alpha-|w^{(1)}+\cdots+ w^{(d)}|^2/4\epsilon}.
\end{align*}
This expression allows us to use Wick's theorem to approximate the $\epsilon\to 0$ limit. The leading term is
\begin{align*}
   &{1\over (4\pi\epsilon)^d}\bracket{\prod_{\alpha=0}^d \int_{\mathbb C^d} \d^{2d}w^{(\alpha)}} \bracket{\prod_{\alpha=1}^d \int_\epsilon^L {\d t_{\alpha}\over (4\pi t_\alpha)^{d}}} \\ &\sum_{i_1,\cdots,i_d}\epsilon_{i_1,\cdots,i_d} \left.\bracket{\prod_{\alpha=1}^d \pa^*_{w^{(\alpha)}_{i_\alpha}} \Phi}\right |_{w_1=\cdots=w_d=0} e^{-\sum\limits_{\alpha=1}^d |w^{(\alpha)}|^2/4t_\alpha-|w^{(1)}+\cdots+ w^{(d)}|^2/4\epsilon}\\
   =&{1\over (4\pi\epsilon)^d}\int_{\mathbb C^d}\d^{2d}w^{(0)} \sum_{i_1,\cdots,i_d}\epsilon_{i_1,\cdots,i_d} \left.\bracket{\prod_{\alpha=1}^d \pa_{w^{(\alpha)}_{i_\alpha}} \Phi}\right |_{w_1=\cdots=w_d=0} \det\bracket{M^{-1}_{\alpha\beta}\over t_\alpha}\\
   & \bracket{\prod_{\alpha=1}^d \int_{\mathbb C^d} \d^{2d}w^{(\alpha)}} \bracket{\prod_{\alpha=1}^d \int_\epsilon^L {\d t_{\alpha}\over (4\pi t_\alpha)^{d}}}e^{-\sum\limits_{\alpha\beta=1}^d M_{\alpha\beta}w^{(\alpha)}\cdot\bar w^{(\beta)}/4}
\end{align*}
Using $\det(M_{\alpha\beta})={\epsilon +\sum\limits_{\alpha=1}^n t_\alpha\over \epsilon\prod\limits_{\alpha=1}^n t_\alpha}$, we can integrate out $w_1,\cdots, w_n$ and get
\begin{align*}
{1\over (4\pi)^d}\int_{\mathbb C^d}\d^{2d}w^{(0)} \sum_{i_1,\cdots,i_d}\epsilon_{i_1,\cdots,i_d} \left.\bracket{\prod_{\alpha=1}^d \pa_{w^{(\alpha)}_{i_\alpha}} \Phi}\right |_{w_1=\cdots=w_d=0}\int_{[\epsilon,L]^d}{\epsilon \d t_1\cdots \d t_d\over \bracket{\epsilon+t_1+\cdots+t_d}^{d+1}}
\end{align*}
We integrate over $t_\alpha$'s. Under the limit $\epsilon\to 0$, it is proportional to
\begin{align*}
\int_{\mathbb C^d}d^{2d}w^{(0)} \sum_{i_1,\cdots,i_n}\epsilon_{i_1,\cdots,i_n} \left.\bracket{\prod_{\alpha=1}^n \pa_{w^{(\alpha)}_{i_\alpha}} \Phi}\right |_{w_1=\cdots=w_n=0}
\end{align*}
If we get back $A_\alpha$'s, we find
$$
  C\int_{\mathbb C^d} A_0(y,\bar y)\wedge \pa A_1(y,\bar y)\wedge\cdots\wedge \pa A_d(y,\bar y)
$$
where $C$ is a non-zero constant which only depends on $d$.

Now we consider higher order terms. It is of the following form
\begin{align*}
U(\epsilon, L)=&{1\over (4\pi\epsilon)^d}\bracket{\prod_{\alpha=0}^d \int_{\mathbb C^d} \d^{2d}w^{(\alpha)}} \bracket{\prod_{\alpha=1}^d \int_\epsilon^L {\d t_{\alpha}\over (4\pi t_\alpha)^{d}}}e^{-\sum\limits_{\alpha=1}^d |w^{(\alpha)}|^2/4t_\alpha-|w^{(1)}+\cdots+ w^{(d)}|^2/4\epsilon} \\ &\sum_{i_1,\cdots,i_d}\epsilon_{i_1,\cdots,i_d} \left.\bracket{\prod_{k=1}^N\bracket{M^{-1}_{\alpha_k\beta_k}{\bar \pa_{\bar w^{(\alpha_k)}_{j_k}}\pa_{w^{(\beta_k)}_{j_k}}}} \prod_{\alpha=1}^d \pa^*_{w^{(\alpha)}_{i_\alpha}} \Phi}\right |_{w_1=\cdots=w_d=0}
\end{align*}
for some $\alpha_k, \beta_k, j_k, 1\leq k\leq N$. Note that
$$
    M^{-1}_{\alpha\beta}=t_\alpha\delta_{\alpha\beta}-{t_\alpha t_\beta\over \epsilon+t_1+\cdots+t_n}
$$
therefore
$$
    |M^{-1}_{\alpha\beta}|\leq 2t_\alpha
$$
It follows that
\begin{align*}
  |U(\epsilon, L)|\leq C &{1\over (4\pi\epsilon)^d}\bracket{\prod_{\alpha=1}^d \int_{\mathbb C^d} \d^{2d}w^{(\alpha)}} \bracket{\prod_{\alpha=1}^d \int_\epsilon^L {\d t_{\alpha}\over (4\pi t_\alpha)^{d}}}\\
  &\det\bracket{{M^{-1}_{\alpha\beta}\over t_\alpha}}\bracket{\prod_{k=1}^N t_{j_k}}  e^{-\sum\limits_{\alpha=1}^d |w^{(\alpha)}|^2/4t_\alpha-|w^{(1)}+\cdots+ w^{(d)}|^2/4\epsilon}\\
  =& C \int_{[\epsilon,L]^d}\d t_1\cdots \d t_d{\epsilon \prod\limits_{k=1}^N t_{j_k}\over \bracket{\epsilon+t_1+\cdots+t_d}^{d+1}}
\end{align*}
for some constant $C$. In particular,
$$
  \lim_{L\to 0}\lim_{\epsilon\to 0}U(\epsilon, L)=0.
$$
The lemma now follows.
\end{proof}

Now we have all the necessary ingredients to compute the local obstruction $O_o$
\begin{lemma}
Given two cyclic tensors $\{A_i\}, \{B_j\}$ of $\Omega^{0,*}_c(\mathbb C^d)$, we have
$$
    O_o(A_0, A_1, \cdots, A_k; B_{k+1},\cdots, B_d)=\pm C_2\int_{\C^d} \bracket{A_0\wedge \pa A_1\wedge\cdots \pa A_k}\wedge \bracket{\pa B_{k+1}\wedge\cdots\wedge \pa B_n}
$$
Here $C_2$ is a nonzero constant which only depends on the dimension $d$.
\end{lemma}
\begin{proof} The two cyclic tensors are two sides of the wheel diagram

\includegraphics[width=2.5in]{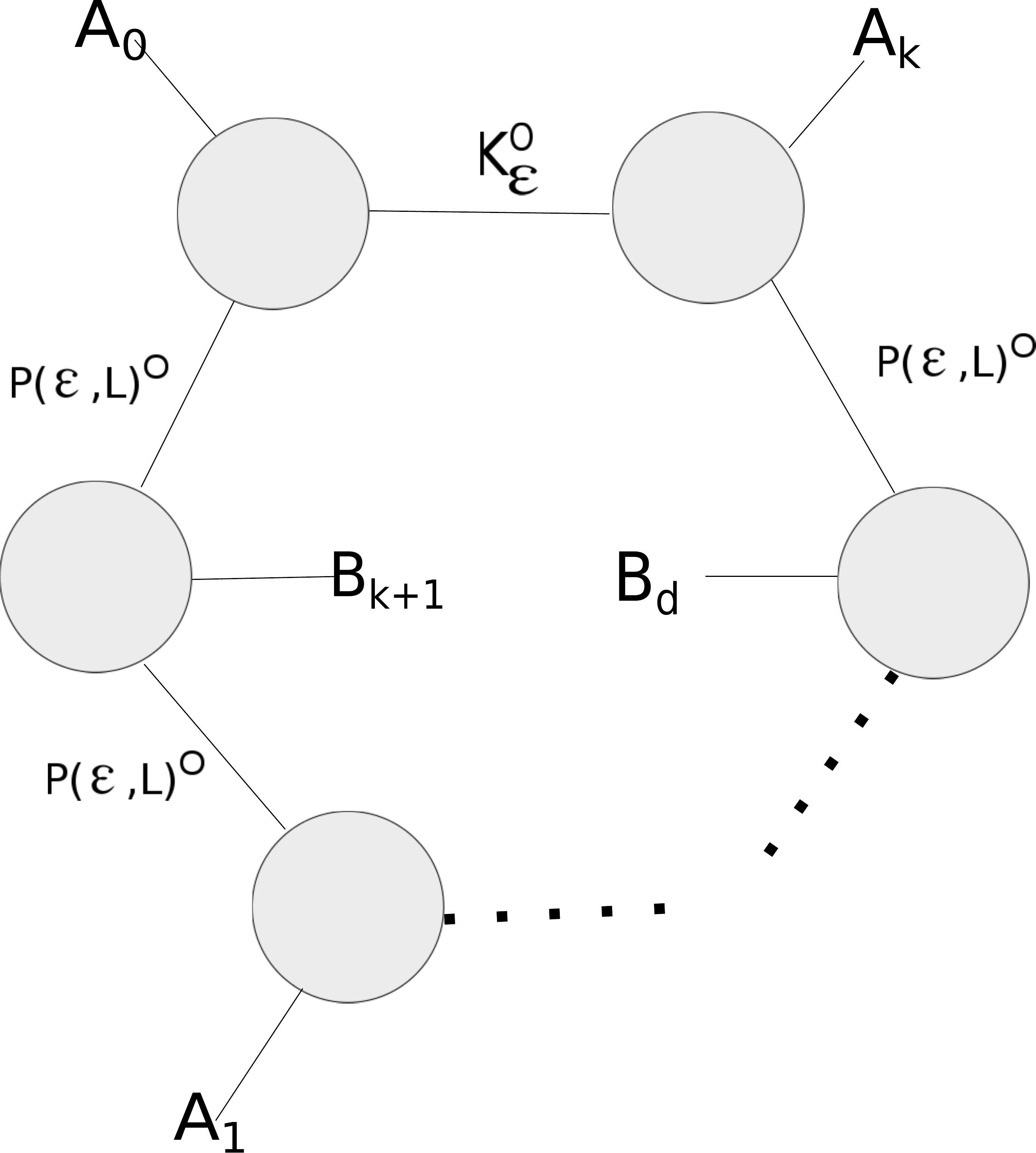}

The lemma follows from Corollary \ref{local m>n} and Lemma \ref{local n=m}.
\end{proof}

Therefore the total anomaly is
$$
  O=O_c+O_o.
$$
$O_c$ and $O_o$ differ by the term 
$$
A\to \int_{\C^d} \Tr (A\wedge (\pa A)^d)\wedge \Tr (1). 
$$
If we work with the quantum theory modulo $N$, then $O_c$ and $O_o$ are equal up to a nonzero constant. Recall the classical master equation 
\begin{align*}
  & QI_{0,1,0} +{1\over 2}\{I_{0,1,0}, I_{0,1,0}\}^O=0\\
  & Q I_{0,1,1}+\{I_{0,1,0}, I_{0,1,1}\}^O=0. 
\end{align*}
There is a rescaling symmetry preserving the master equation
$$
    I_{0,1,1}\to \lambda  I_{0,1,1}, \quad \lambda \in \C^*,
$$
which results in the rescaling of the anomaly
$$
   O\to \lambda^2 O_c+O_o.
$$
Therefore under a suitable rescaling of $I_{0,1,1}$ by a nonzero constant, 
$O=0$. This rescaling constant is uniquely fixed (up to $\pm$) by the annulus anomaly cancellation condition. Theorem \ref{annulus-cancellation} is proved.

\section{Classical BCOV interaction}
The quantum open-closed BCOV theory we constructed in this paper in particular includes a classical BCOV interaction, defined by
$$
I_{0,0,n} = \lim_{L \to 0} I_{0,0,n}[L].
$$ 
The fact that $I_{0,0,n}[L]$ satisfies the purely closed string classical RG flow and equation implies that this limit exists and is a local functional. We let
$$
I_{0,0,\ast}(\alpha) = \sum_{n \ge 3} I_{0,0,n}(\alpha). 
$$ 
The master equation for $I_{0,0,n}[L]$ implies that $I_{0,0,\ast}$ satisfies the classical master equation. 

In \cite{CosLi11} we described a classical interaction for BCOV theory on any Calabi-Yau $X$. Let us call this classical interaction $J$.   It is defined as follows. If $\alpha \in \PV^{\ast,\ast}_c(\C^d)[[t]][2]$, let $\alpha_k$ denote the coefficient of $t^k$. Then, we set
$$
J_n(\alpha) = \sum_{k_1, \dots, k_n \text{ with } \sum k_i = n-3} c_{k_1,\dots,k_n}  \int \alpha_{k_1} \wedge \dots \wedge \alpha_{k_n} 
$$
where the combinatorial constant $c_{k_1,\dots,k_n}$ can be defined by
$$
c_{k_1,\dots,k_n} = \tfrac{1}{n!}\int_{\mbar_{0,n}} \psi_1^{k_1} \dots \psi_n^{k_n}.  
$$
We then define the interaction $J$ by saying that
$$
J(\alpha) = \sum_{n \ge 3} J_n(\alpha). 
$$
One can check \cite{CosLi11} that $J(\alpha)$ satisfies the classical master equation. In particular, it defines a $L_\infty$ structure 
$$
l_n : (\PV(\C^d)[[t]])^{\otimes n} \to \PV(\C^d)[[t]]
$$  
where $Q+\{J,-\}^C$ is the corresponding Chevalley-Eilenberg differential. Such defined $L_\infty$-structure is $L_\infty$-equivalent to the standard dg Lie algebra structure on $\PV(\C^d)[[t]]$ via the nonlinear transformation 
$$
     \PV(\C^d)[[d]] \to \PV(\C^d)[[t]], \quad \mu \to \pi_+(te^{\mu/t}-t).
$$
Here $\pi_+: \PV(\C^d)((t))\to \PV(\C^d)[[t]]$ is the projection by picking up the non-negative powers of $t$. 

In the introduction we made the following conjecture. 
\begin{conjecture}
The functionals $J$ and $I_{0,0,\ast}$ are equivalent solutions to the classical master equation. 
\end{conjecture}
Let us now prove the result we stated in the introduction:
\begin{proposition}
This conjecture implies the holomorphic analog of Willwacher-Calaque's refinement of Kontsevich's formality theorem.
\end{proposition}

\begin{proof}[Sketch of proof of the proposition]
The functionals $I$ and $J$ each define an $L_\infty$-structure on the space $\PV(\C^d)[[t]][1]$, with the feature that a solution to the equations of motion is a Maurer-Cartan element in this $L_\infty$ algebra structure.  For both $I$ and $J$, the $L_\infty$ structures are local, meaning that the maps
$$
l_n : (\PV(\C^d)[[t]])^{\otimes n} \to \PV(\C^d)[[t]]
$$ 
are multi-differential operators. 

Further, the operations $I_{0,1,n}[L]$ satisfy a classical renormalization group equation, and so have an $L \to 0$ limit which is a local functional which we call $I_{0,1,n}$. Let us further refine this functional, and let $I_{0,(1|k), n}$ denote the component coming from discs with $k$ marked points on the boundary and $n$ on the interior. 

The local functional $I_{0,(1|k), n}$ defines a multi-differential operator
$$
\left( \PV(\C^d)[[t]] [2]\right)^{\otimes n} \times \Omega^{0,\ast}(\C^d)^{\otimes k-1} \to \Omega^{0,\ast}(\C^d)
$$

The $L_\infty$ structure associated to $J$, hence $I$, is equivalent to the standard one, with the only non-zero higher bracket $l_n$ being $l_2$ which is the Schouten bracket. The content of the proposition is the statement that the $L_\infty$ structure associated to ${I}$ is equivalent to the $L_\infty$ structure on cyclic cochains, which is in fact realized by the $L_\infty$-morphism $I_{0,1,n}$. 

\end{proof}

Next, we will prove the following lemma, stated in the introduction. 
\begin{lemma}\label{classical-interaction-lemma}
$I_{0,0,3}$ and $J_3$ define equivalent solutions of the classical master equation modulo quartic terms. 
\end{lemma}
\begin{proof}
$I_{0,0,3}$ is holomorphically translation invariant, $SU(d)$ invariant, and of weight $2d$ under the action of the $S^1$ which scales $\C^d$. We will use these facts to constrain the form of $I_{0,0,3}$. 

The classical master equation modulo quartic terms simply says that $I_{0,0,3}$ is a closed element of the complex of local functionals, equipped with the differential coming from the linear differential on the fields of BCOV theory.

The arguments we used earlier allow us to express the relevant cohomology group in terms of jets of holomorphic polyvector fields. We have 
$$
J_0(\PV[[t]][2]) = \C[[z_i, \partial_i, t]][2]
$$
where the variables $\partial_i$ have degree $1$ and $t$ has degree $2$.   The differential is as usual $t \partial$. Let $D_0 = \C[\dpa{z_i}]$ be the algebra of constant-coefficient differential operators.  Then, the dual to jets of polyvector fields is
$$
J_0(\PV[[t]][2])^\vee = D_0 [\alpha_i, t^{-1}][-2]
$$
where the variables $\alpha_i$ are dual to $\partial_i$ and have degree $-1$, and $t^{-1}$ has degree $-2$.  The differential on the dual is the operator $\sum_i \alpha_i \dpa{z_i} t$.  Under the action of $U(d)$, the $\alpha_i$ transform as  covectors $\d z_i$, so that the algebra generated by the $\alpha_i$ is that of constant-coefficient differential forms. 

The complex where possible interactions $I_{0,0,3}$ live is
$$
\Omega^{d,0}_0 \otimes^{\mbb L}_{D_0} \Sym^3 ( J_0(\PV[[t]])^\vee ) [-d]
$$
where the cohomological shift is so that $I_{0,0,3}$ is of degree zero. Here $\Omega^{d,0}_0$ refers to the fibre at $0$ of the bundle of $(d,0)$ forms. 

Let us ignore, for now, the internal differential on $ J_0(\PV[[t]])$.  Note that $ \Sym^3 ( J_0(\PV[[t]])^\vee )$ is flat as a $D_0$ module.  So we can take the actual tensor product instead of the derived one.  Let us also, for now, investigate the tensor cube rather than the symmetric cube, and we will take symmetric invariants at the end.

We find, after these considerations, that we are looking at
\begin{align*} 
 \Omega^{d,0}_0  \otimes_{D_0} \left(D_0[\alpha_i, t^{-1}] \right)^{\otimes 3} [-d]
= \Omega^{d,0}_0 \otimes_{\C} \C[\alpha_i,t^{-1}] \otimes_{\C} \left(D_0[\alpha_i,t^{-1}] \right) ^{\otimes 2} [ -d].
\end{align*}
Now, the $\alpha_i$ have weight $+1$ under the action of $S^1 \subset U(d)$, and we are interested in something of total weight $2d$.  Since the $\d z_1 \dots \d z_d \in \Omega^{d,0}_0$ has weight $d$ and elements in $D_0$ all have non-positive weights, this means we must have at least $d$ $\alpha_i$'s.  Now, if we have more than $d$ $\alpha_i$'s, then we are pushed into negative cohomological degree. It follows that we must have precisely $d$ $\alpha_i$'s. Further, cohomological degree consideration tells us that the only powers of $t^{-1}$ that appear is $t^{0} = 1$.

Next, $SU(d)$ invariance tells us that the only possibilities are of the form
$$
\mu_{c_1,c_2,c_3} = \sum_{\abs{I_1} = c_1, \ \abs{I_2} = c_2, \ \abs{I_3} = c_3} \pm  \d z_1 \dots \d z_d \alpha_{I_1} \otimes \alpha_{I_2} \otimes \alpha_{I_3} .
$$   
Here $\{I_1, I_2, I_3\}$ is a partition of $\{1,\dots,d\}$ and $\alpha_I$ means $\prod_{i \in I} \alpha_i$. 

As a functional on $\PV_{c}(\C^d)[[t]]$, $\mu_{c_1,c_2,c_3}$ sends a field 
$$
\alpha = \sum t^k \alpha_k = \sum t^k \alpha_k^{l}
$$
where $\alpha_{k}^l \in \PV^{l,\ast}(\C^d)$ to
$$
\int (\alpha_{0}^{c_1} \alpha_0^{c_2} \alpha_0^{c_3}\vdash \Omega)\wedge \Omega, \quad \Omega=dz_1\cdots dz_d.
$$
Let us call this functional $\rho_{c_1,c_2,c_3}$.  Our action functional $I_{0,0,3}$ must be a sum
$$
I_{0,0,3} = \sum_{c_1,c_2,c_3} A_{c_1,c_2,c_3} \rho_{c_1,c_2,c_3} 
$$
for some constants $A_{c_1,c_2,c_3}$.   It remains to show that these constants are all $1$.

To check this, note that $I_{0,0,3}$ defines a bracket on $\PV(\C^3)$ with the feature that the Maurer-Cartan equation for this bracket is the equations of motion. Let us call this $\{-,-\}_{I_{0,0,3}}$. Explicitly, this bracket is of the form
$$
\{\alpha,\beta\}_{I_{0,0,3}} = \sum A_{c_1,c_2, d-c_1-c_2} \{\alpha^{c_1}, \beta^{c_2}\}
$$
where $\alpha^{c_1}$ is the component in $\PV^{c_1,\ast}$ and similarly for $\beta^{c_2}$, and $\{-,-\}$ is the Schouten bracket.

Now, at the cohomological level, the argument in the proof of proposition \ref{classical-interaction-lemma} tells us that this must be Lie bracket on the cyclic cohomology of $\Omega^{0,\ast}(\C^3)$.  The Hochschild-Kostant-Rosenberg theorem then tells us that the bracket $\{-,-\}_{I_{0,0,3}}$ must be (at the cohomological level) the usual Schouten bracket, and this fixes all the constants $A_{c_1,c_2,c_3}$. 
\end{proof}

\section{$(1,0)$ BCOV theory}
In the introduction we described a variant of BCOV theory which we call $(1,0)$ BCOV theory. Let us recall the definition and explain how the techniques we have developed so far allow us to quantize this theory on a variety of Calabi-Yau $3$-folds. 

The fields of $(1,0)$ BCOV theory on a Calabi-Yau $3$-fold $X$ form the complex
$$
\Omega^{0,\ast}(X,TX)[1] \oplus \Omega^{0,\ast}(X)
$$  
with differential $\dbar + t \partial$, where
$$
\partial : \Omega^{0,\ast}(X,TX)\to \Omega^{0,\ast}(X)
$$
is the holomorphic divergence map.  We can view this complex of fields as the subspace
$$
\PV^{1,\ast}(X) [1] \oplus t \PV^{0,\ast}(X) \subset \PV^{\ast,\ast}(X)[[t]][2]
$$ 
of the fields of BCOV theory.

This theory is a degenerate theory in the BV formalism, just like ordinary BCOV theory.  Therefore, instead of writing down a quadratic interaction we will write down a kernel describing the BV anti-bracket. This, together with the differential in the complex of fields, allows us to construct such things as the propagator and to quantize the theory.

In ordinary BCOV theory on a CY $3$-fold $X$, the kernel for the BV bracket is  $(\partial \otimes 1) \delta_{Diag}$ where $\delta_{Diag}$ is the $\delta$-current on the diagonal of $X \times X$, viewed as a polyvector field via the isomorphism between forms and polyvector fields.  This kernel $\pi$ is a sum of the form
$$
\pi = \sum_{j+l = 3 \ i + k = 2} \pi_{(i,j),(k,l)}
$$ 
where
$$
\pi_{(i,j),(k,l)} \in \br{PV}^{i,j}(X) \what{\otimes} \br{PV}^{k,l}(X).
$$
To get the kernel for the $(1,0)$ BCOV theory, we simply take the components of this kernel which are of type $(1,j) \times (1,l)$ where $j+l=3$. 

In a similar way, the closed string propagator $P(\eps,L)^C$ and regularized BV kernel $\tr_L^{C}$ are the projections of those for the full BCOV theory onto polyvector fields of type $(1,\ast)$.  

This theory couples to holomorphic Chern-Simons for the group $\mf{sl}(N \mid N)$ via an interaction we call $I_{0,1,1}$, as before. This interaction is simply the restriction of the interaction between the fields of the full BCOV theory with holomorphic Chern-Simons to those fields which are in the subspace consisting of the fields of $(1,0)$ BCOV  theory.  One can check, from our earlier formula for $I_{0,1,1}$, that if $\mu \in \Omega^{0,\ast}(X,TX)$, $\phi \in \Omega^{0,\ast}(X)$, and  $A \in \Omega^{0,\ast}(X) \otimes \mf{sl}(N \mid N)$
\begin{align*} 
 I_{0,1,1}(\phi + \mu + A) = \tfrac{1}{2} \int_{X} \op{Tr} (A \del A)\wedge (\mu \vdash \Omega) + \tfrac{1}{3}  \op{Tr} (A^3)\wedge  \phi \Omega.
\end{align*}
In other words, we have the holomorphic Chern-Simons action for the complex structure deformed by $\mu$ with volume form scaled by $\phi$.

We can define, exactly as before, the notion of quantization of the theory coupling $(1,0)$ BCOV theory to $\mf{sl}(N \mid N)$ holomorphic Chern-Simons, in a way compatible with the inclusion maps $\mf{sl}(N \mid N) \into \mf{sl}(N + k \mid N + k)$.  Let us refer to this simply as open-closed $(1,0)$ BCOV theory. 

Let us now explain the analogs of our theorems A and B from earlier, concerning cohomological cancellation for quantization beyond the annulus level and the annulus anomaly cancellation.  

Let $X$ be a conical Calabi-Yau, meaning that there is a $\C^\times$ action on $X$ which scales the holomorphic volume form by some positive weight $w$.  As before, let $\T^{<(G,R)}(X)$ and $\T^{\le (G,R)}(X)$ refer to the simplicial sets for quantizing the theory up to level $(G,R)$, in a way compatible with the $\C^\times$ action on the conical Calabi-Yau.  
\begin{theorem}
\label{theorem_(1,0)_cohomology_cancellation}
\begin{enumerate}
\item For any  Calabi-Yau 3-fold $X$, if $R > 1$ then the map $\T^{\le (G,R)} (X) \to \T^{< (G,R)}(X)$ is a weak equivalence of simplicial sets. 
\item 
Let $X$ be a conical Calabi-Yau manifold. Further, let us assume that $H^0_{\dbar}(X, \Oo_X)$ and $H^1_{\dbar}(X, \Oo_X)$ are spanned by elements of non-negative weight under the $\C^\times$ action. (This happens, for instance, if $X$ is the total space of the canonical bundle over a complex surface where the $\C^\times$ action scales the fibres of the canonical bundle). 

Then if $2G-2+R > 0$ the map
$$
\T^{\le (G,R)} (X) \to \T^{< (G,R)}
$$ 
is a weak equivalence of simplicial sets.  
\item
If we consider holomorphically translation invariant theories on $\C^3$, then if $2G-2+R > 0$ the map 
$$
\T^{\le (G,R)}(\C^3) \to \T^{<(G,R)}(X)
$$ 
is a weak equivalence.  The same result holds when we consider theories on $Y \times \C^d$ where $Y$ is a Calabi-Yau $(3-d)$-fold and we ask for quantizations which are holomorphically translation invariant in the $\C^d$ factor. 
\end{enumerate}
\end{theorem}
\begin{theorem}
\label{theorem_(1,0)_annulus_anomaly}
For any Calabi-Yau manifold $X$, there is a unique quantization of open-closed $(1,0)$ BCOV theory to the annulus level. 
\end{theorem}
\begin{corollary}
For any conical Calabi-Yau manifold satisfying the conditions stated above, there is a unique quantization of $(1,0)$ open-closed BCOV theory. In particular, this leads to a quantization of $(1,0)$ BCOV theory. The same holds for holomorphically translation invariant on $\C^3$ or theories on $Y \times \C^d$ which are holomorphically translation invariant in the $\C^d$ direction. 
\end{corollary}
\begin{corollary}
There is a unique quantization of open-closed $(1,0)$ BCOV theory on any Calabi-Yau $X$ at genus $0$, i.e. in the planar limit.  
\end{corollary}

\begin{proof}[Proof of \ref{theorem_(1,0)_cohomology_cancellation}]
The proof is very much along the lines of the proof of theorem \ref{theorem_cohomology_cancellation}. In particular, the obstruction-deformation complex can be expressed in terms of certain $D$-modules. The $D$-module of jets of closed-string fields in this case is the jets of the complex of fields of $(1,0)$ BCOV theory. This complex is equivalent to the complex 
$$\Omega^{2,\ast}(X)[1] \xto{\del} \Omega^{3,\ast} = F^2 \Omega^{\ast,\ast}(X)[3],$$
which is the cochain level version of $F^2$ in the Hodge filtration. The differential is $\dbar + \del$.  Locally, the cohomology of this complex consists of closed holomorphic $2$-forms on $X$.  Thus, from the closed string sector jets of fields are $J F^2 \Omega^{\ast,\ast}(X) [3]$ and functionals on the jets of fields are
$$
\Omega^{3,3} \otimes_{D} \Oo (  J F^2 \Omega^{\ast,\ast} [3] ) 
$$
where $\Oo$ refers to functionals modulo constants, and we are tensoring over all differential operators (not just holomorphic ones).

For the open string sector, we can use the argument from before to express possible admissible local functionals in terms of functions on the cyclic homology of jets of $\Omega^{0,\ast}(X)$ (shifted by $[1]$).   The cyclic homology of $J \Omega^{0,\ast}(X)$ is the complex $J \Omega^{-\ast,\ast}(X)[t^{-1}]$ with differential $t \del$.  There is a direct summand in this consisting of $J \Omega^{0,\ast}(X)$, and this corresponds to the part of the Lie algebra homology of $J \Omega^{0,\ast}(X) \otimes \gl(N \mid N)$ which comes from the Abelian sub-algebra consisting of $J \Omega^{0,\ast}(X) \otimes \op{Id}$.  It follows that the variant of cyclic homology we need to take which describes $\mf{sl}(N \mid N)$ Lie algebra homology is simply the complement of this summand. Let us call this $\mf{sl}$-cyclic homology. We find that
$$
HC_{\mf{sl}}( J \Omega^{0,\ast}(X))[1] = \oplus_{i \ge 1} J \Omega^{i,\ast}(X)[i+1] \oplus t^{-1} J \Omega^{-\ast,\ast}(X)[t^{-1}][1]. 
$$ 

As in the proof of theorem \ref{theorem-vanishing} the possible local functionals that can arise in the open-closed theory are described as functionals on a complex which is a direct sum of this cyclic homology complex together with the jets of closed string fields, with a differential connecting them coming from the disk with one interior marked point.  In sum, we find that the jet complex is described by the double complex which is the sum of the rows in the following table
$$
 \xymatrix{\text{Sector} & closed & closed & open & open & open & \cdots \\
\text{Degree} & 0 & -1 & -2 & -3 & -4 & \cdots   \\  
\text{Jet complex}  &J  \Omega^{3,\ast} & \ar[l]_{\partial} J \Omega^{2,\ast} & \ar[l]_{\partial} 
J\Omega^{1,\ast} & \ar[l]_{t \partial} t^{-1}J \Omega^{0,\ast}  
&   & \\
&  &  &  & J\Omega^{2,\ast} & \ar[l]_{t \partial} t^{-1}J \Omega^{1,\ast} & \ar[l]_{t \partial} \cdots
 \\
& & & & & J \Omega^{3,\ast}&\ar[l]_{t \partial} \cdots \\
& & & & &                  & \cdots } 
$$ 
We have included here the map from the connecting differential from the closed string sector to the open string sector on the third row. 

We see that all the rows of this complex are simply the de Rham complex with a shift, with the exception of that in the fourth row, which does not include $\Omega^{3,\ast}$.

If $M$ denotes the complex of $D$-module appearing in this diagram, then our obstruction-deformation complex is $\Omega^{3,3} \otimes_{D} \Oo(M)$ where $\Oo$ indicates functions modulo constants and $D$ is the algebra of $\cinfty$ differential operators.

Since jets of the de Rham complex is quasi-isomorphic to just the trivial $D$-module $\cinfty_X$ (i.e. the sheaf of smooth functions with its trivial flat connection) we find that there is a quasi-isomorphism
$$
M \simeq J \Omega^{\le 2,\ast}[5] \bigoplus  \oplus_{i = 3,7,9\dots}\cinfty_{X}[i]. $$ 

Introduce variables $\eps_3,\eps_7,\eps_9\dots$ where $\eps_i$ is of degree $i$.  Then, we find that 
\begin{multline*} 
 \Omega^{3,3} \otimes_{D}^{\mbb L} \Oo(M) = \Omega^{3,3} \otimes_{D} \left\{ \left( \Sym^{> 0} (J \Omega^{\le 2,\ast})^\vee [-5]\right) [\eps_3,\eps_7,\eps_9\dots] \right\} \\
\oplus  \Omega^{3,3} \otimes_{D}^{\mbb L} \cinfty_X \otimes_{\C} ( \C[\eps_3,\eps_7,\eps_9\dots]/\C ).  
\end{multline*}
We can compute the relevant tensor products over $D$. First we have 
$$
\Omega^{3,3} \otimes_{D}^L \cinfty_X \simeq \Omega^{\ast,\ast}(X)[6]. 
$$
We also have
$$
\Omega^{3,3} \otimes_{D} (J \Omega^{\le 2,\ast})^\vee \simeq \Omega^{\ge 1,\ast}(X)[6]. 
$$ 
After all, the complex on the left hand side is a way of writing linear local functionals of $\Omega^{\le 2,\ast}(X)$, and of course every local functional is given by integrating against an element of $\Omega^{\ge 1,\ast}(X)$.

Note that the cohomology of the truncated de Rham complex $\Omega^{\ge 1,\ast}(X)[1]$ is the cohomology of $X$ with coefficients in the sheaf of closed holomorphic $1$-forms, so we will refer to it as $H^\ast(X, \Omega^1_{cl,hol})$. Therefore
$$
  H^\ast(\Omega^{3,3} \otimes_{D} (J \Omega^{\le 2,\ast})^\vee[-5])=H^\ast(X, \Omega^1_{cl,hol}).
$$

For the cohomology of $\Omega^{3,3} \otimes_{D}\Sym^k \left(J \Omega^{\le 2,\ast})^\vee [-5] \right)$, we need a lemma.
\begin{lemma}
If $k > 1$ then  
$$ H^i \left( \Omega^{3,3}\otimes_{D} \Sym^k \left(J \Omega^{\le 2,\ast})^\vee [-5] \right) \right)$$
is zero if $i < 2$. 
\end{lemma}
\begin{proof}
There is a short exact sequence of $D$-modules
$$
0 \to J\Omega^{3,\ast}[2] \to J\Omega^{\ast,\ast}[5] \to J\Omega^{\le 2,\ast}[5]\to 0.
$$
Applying this short exact sequence yields a spectral sequences converging to the cohomology groups we are computing whose first term involves replacing $J \Omega^{\le 2,\ast}[5]$ by $J \Omega^{\ast,\ast} [5] \oplus J \Omega^{3,\ast}[3]$. As above, the cohomology of jets of the de Rham complex is just the trivial $D$-module $\cinfty_X$, but now situated in degree $-5$. The cohomology sheaves of the dual of $J \Omega^{3,\ast}[3]$ is the $D$-module $D_{hol}$ of smooth sections of the bundle of holomorphic differential operators, situated in degree $3$. 

The functor which sends a $D$-module $M$ to $\Omega^{3,3} \otimes^{\mbb L}_D M$ can be modelled by the de Rham complex of $X$ with coefficients in $M[6]$. So if $M$ is in degree $k$ this functor yields something concentrated in degrees between $k$ and $k - 6$. If $k > 7$ then we find nothing in degrees $\le 1$.   

Now, the higher symmetric powers of $\cinfty_X[-5] \oplus D_{hol}[-3]$ has summands in degrees $6,8,\dots$. Only the term in degree $6$, namely $\wedge^2 D_{hol}$, can possibly contribute. However, $\wedge^2 D_{hol}$ is a summand of $D_{hol}^{\otimes 2}$, and the sheaf $\Omega^{3,3} \otimes_{D} D_{hol}^{\otimes 2}$ has cohomology in degrees $\ge -3$. Since we are shifting up by $6$, we find no cohomology in degrees $\le 1$. 
\end{proof}

Putting this together tells us that our obstruction-deformation complex has cohomology:
\begin{multline*}
 H^\ast (\Omega^{3,3} \otimes_D \Oo(M)) = H_{dR}^\ast(X)[6] \otimes\left( \C[\eps_3,\eps_7,\dots ] / \C \right)  \\
\oplus H^\ast(X, \Omega^1_{cl,hol}) [\eps_3,\eps_7,\dots, ] \\
\oplus \text{ things in degree } 2 \text{ and higher}. 
\end{multline*}
Since $d=3$, the relevant pieces to our computation are in degrees $\le 1$. We have
\begin{align*}
H^{i}(\Omega^{3,3} \otimes_D \Oo(M)) & = H^{3+i}(X)\eps_3\text{ if } i < 0 .\\
H^0 ( \Omega^{3,3} \otimes_D \Oo(M)) &= H^3(X) \eps_3 \oplus  H^0(X,\Omega^1_{cl,hol}) \\
H^1(\Omega^{3,3} \otimes_D \Oo(M)) &= H^4(X) \eps_3  \oplus H^0(X) \eps_7 \oplus  H^1(X, \Omega^{1}_{cl,hol}) 
\end{align*}

Note that all of these cohomology groups are built from functionals which are in the first symmetric power of the dual of our $D$-module of jets. This means they correspond either to functionals which are linear of the closed string fields and independent of the open string fields, or to functionals which do not depend on the closed string fields but which are single-trace on the open string fields. 

This implies that the map $\T^{\le(g,R)}(X) \to \T^{<(g,R)}(X)$ is a weak equivalence if $R > 1$, for any $X$.

Next, let us see what happens if $X = \C^3$ and we consider holomorphically translation invariant quantizations.  In this case, our discussion in the proof of theorem \ref{theorem-vanishing-flat-space} tells us that we should replace the $D$-module $M$ by its fibre $M$ at $0$, which we call $M_0$; and the algebra $D$ of differential operators by the algebra $D_0^{hol} = \C[\dpa{z_i}]$ of constant-coefficient holomorphic differential operators. The obstruction-deformation complex is
$$
\Omega^{3,0}_0\otimes_{D_0^{hol}}^{\mbb L} \Oo(M_0)[3].
$$
The arguments described above that calculate the cohomology of the $D$-module $M$ also calculate that of the $D_0^{hol}$-module $M_0$.  We find that we should replace every occurrence of the cohomology of $X$ by the algebra $\Omega_0^{*,0}=\C[\d z_i]$ of holomorphic differential forms with constant coefficients, and every occurrence of the cohomology of $X$ with coefficients in $\Omega^1_{cl,hol}$ by the complex $\Omega_0^{\geq 1,0}[1]=\d z_i \C[d z_i] [1]$, of holomorphic differential forms of degree $1$ and higher, with grading such that $i$ forms are in degree $i-1$. 

So we find that the obstruction-deformation complex for holomorphic translation invariant quantizations on $\C^3$ has cohomology groups:
\begin{align*}
H^{i} ( \Omega^{3,0}_0\otimes_{D_0^{hol}}^{\mbb L} \Oo(M_0)[3]) & = \Omega_0^{3+i,0} \eps_3\text{ if } i < 0 \\
H^0(\Omega^{3,0}_0\otimes_{D_0^{hol}}^{\mbb L} \Oo(M_0)[3])&= \C \d z_1 \d z_2 \d z_3 \eps_3 \oplus \Omega_0^{1,0}\\ 
H^1(  \Omega^{3,0}_0\otimes_{D_0^{hol}}^{\mbb L} \Oo(M_0)[3])&=   \C \eps_7 \oplus \Omega_0^{2,0}. 
\end{align*}
If we also ask for $SU(3)$-invariants we find
\begin{align*}
H^{-3} ( \Omega^{3,0}_0\otimes_{D_0^{hol}}^{\mbb L} \Oo(M_0)[3])^{SU(3)} & = \C \eps_3\\
H^{i}  ( \Omega^{3,0}_0\otimes_{D_0^{hol}}^{\mbb L} \Oo(M_0)[3])^{SU(3)} &= 0 \text{ if } i = -1,-2\\ 
H^0(\Omega^{3,0}_0\otimes_{D_0^{hol}}^{\mbb L} \Oo(M_0)[3])^{SU(3)} &= \C \d z_1 \d z_2 \d z_3 \eps_3 \\ 
H^1(  \Omega^{3,0}_0\otimes_{D_0^{hol}}^{\mbb L} \Oo(M_0)[3])^{SU(3)} &=   \C \eps_7. 
\end{align*}
That is, in this case there is a single possible obstruction, a single possible deformation, and no symmetries. 

Let's now discuss the possible deformations and obstructions explicitly. We will do this in the case of a general Calabi-Yau, but the same formulae apply to the holomorphically translation invariant situation on $\C^3$. 
\begin{enumerate} 
 \item A class $\alpha \in H^3(X)$ gives a deformation.  Let us write explicitly the corresponding Lagrangian of the theory. Choose a cochain representative of $\alpha$ as a sum of $(i,j)$ forms $\alpha^{i,j}$. Let $A \in \Omega^{0,\ast}(X) \otimes \mf{sl}(N \mid N)$ be an open string field, $v \in  \Omega^{0,\ast}(X,TX)$ and $\phi \in \Omega^{0,\ast}(X)$ be the closed string fields. The deformation is  by the Lagrangian
$$
\int \alpha^{3,0}\left( \tfrac{1}{3} \op{Tr} A^3 + \tfrac{1}{2}\op{Tr} \A \dbar A\right) +\tfrac{1}{2} \int \alpha^{2,1} \op{Tr} A \partial A + \int \alpha^{1,2} (v \vdash \Omega) + \int \alpha^{0,3} \phi \Omega. 
$$
The terms coming from $\alpha^{3,0}$ and $\alpha^{2,1}$ are simply the deformation of the holomorphic Chern-Simons functional when we change the complex structure and holomorphic volume form.  In the holomorphically translation invariant case on $\C^3$, $\alpha$ is of course of type $(i,0)$. 
\item Similarly, a class $\alpha \in H^i(X) \eps_3$ gives an anomaly if $i = 4$ and a symmetry if $i = 2$.  These are represented by Lagrangians of cohomological degree $1$ and $-1$ respectively.  If we expand, as above, a cochain representative of $\alpha$ in terms of $(p,q)$-forms $\alpha^{p,q}$ where $p+q = i$, the corresponding Lagrangian is 
$$
\int \alpha^{3,i-3}\left( \tfrac{1}{3} \op{Tr} A^3 + \tfrac{1}{2}\op{Tr} \A \dbar A\right) +\tfrac{1}{2} \int \alpha^{2,i-2} \op{Tr} A \partial A + \int \alpha^{1,i-1} (v \vdash \Omega) + \int \alpha^{0,i} \phi \Omega. 
$$
Note that if $i = 2$ then the first term does not appear and if $i = 4$ the last term does not appear.
\item A closed holomorphic $(1,0)$-form  $\beta \in H^0(X,\Omega^1_{cl,hol})$ gives rise to a deformation.  This is the functional
$$
A \mapsto \int \beta \op{Tr}\left( A (\partial A)^2\right).  
$$
This functional is the variation in the holomorphic Chern-Simons functional as we make $X$ non-commutative in a way coming form the holomorphic bi-vector whose contraction with $\Omega$ is $\beta$. 
\item A class $\beta \in H^1(X,\Omega^1_{cl,hol})$ gives a potential anomaly.  If, at the cochain level, we write $\beta$ as a sum of $\beta^{1,1}$ and $\beta^{2,0}$, then the functional is of the form
$$
A \mapsto c_1 \int \beta^{1,1} \op{Tr} \left( A (\partial A)^2 \right) + c_2 \int \del \beta^{1,1} \op{Tr} \left(A^2 \partial A \right)  + c_3 \int \beta^{2,0} \op{Tr} ( A^3 \partial A).  
$$
This is the functional $I^{1-disk}(\beta^{1,1} + t \beta^{2,0}, A)$ where  we think of $\beta^{1,1}$ as a polyvector field in $\PV^{2,1}$ and $\beta^{2,0}$ as in $\PV^{1,0}$, and where $I^{1-disk}$ indicates the interaction for the full BCOV theory corresponding to a disk with a single marked point in the interior.    
\item The class $1 \eps_7 \in H^0(X) \eps_7$ gives a potential anomaly of the form
$$
\int \op{Tr} \left(A (\partial A)^3 \right). 
$$
\end{enumerate}
If $X$ is a conical Calabi-Yau, then we can analyze how these functionals scale when we scale our fields using the $\C^\times$ action on $X$ and use this analysis to complete the proof of the theorem.  We will scale polyvector fields using the natural scaling on the de Rham complex and the isomorphism between polyvector fields and the de Rham complex. If we do this, then we find that the functionals corresponding to classes in $H^\ast(X)$ and in $H^\ast(X,\Omega^1_{cl,hol})$ scale according to the natural action of $\C^\times$ on these cohomology groups.  Note that $\C^\times$ acts trivially on $H^\ast(X)$ but possibly non-trivially on $H^\ast(X,\Omega^1_{cl,hol})$.  

The axiom of compatibility with the scaling action tells the obstruction-deformation group controlling lifts from $\T^{<(G,R)}(X)$ to $\T^{\le (G,R)}(X)$ consists of Lagrangians of weight $(-2G +2 -R)w$, where $w$ is the weight of the holomorphic volume form on $X$ under scaling. To prove our theorem, we need to show that there are no elements in the obstruction-deformation complex of negative weight.   Elements in this complex which come from de Rham cohomology of $X$ are of weight zero, so don't contribute. It suffices to show that elements coming from $H^\ast(X,\Omega^1_{cl,hol})$ also don't contribute.  

The short exact sequence
$$
  0\to \C\to \Oo_X\stackrel{d}{\to} \Omega^1_{cl,hol}\to 0
$$
leads to a long exact sequence of the form
\begin{equation*}
\dots \to H^i_{\dbar}(X,\Oo_X) \to H^{i}(X,\Omega^1_{cl,hol}) \to H^{i+1}_{dR}(X) \to \dots \tag{$\dagger$}
\end{equation*}
and the maps in the exact sequence commute with the $\C^\times$ action on everything.  Since everything in $H^{i+1}_{dR}(X)$ is $\C^\times$-invariant, any elements of $H^{i+1}(X,\Omega^{1}_{cl,hol})$ which are not $\C^\times$ invariant come from $H^i_{\dbar}(X,\Oo_X)$. In the statement of our theorem, we assume that every element of $H^0(X,\Oo_X)$ and $H^1_{\dbar}(X,\Oo_X)$ is of non-negative weight under the $\C^\times$-action, so that it can not contribute to the obstruction-deformation group, thus completing the proof of the theorem. 

Next, let us discuss the holomorphically translation invariant case on $\C^3$. In this case, once we also impose $SU(3)$-invariance, all possible obstructions and also higher symmetries (coming in this case from $H^{-3}$ of the obstruction-deformation complex) are scale invariant. The one possible deformation -- corresponding to the Lagrangian $\int \d z_1 \d z_2 \d z_3 \op{Tr} A^3$ -- has positive weight, so also can not contribute. The result follows. 

Finally let us discuss the case of a Calabi-Yau of the form $Y \times \C^d$, where $Y$ is a Calabi-Yau of dimension $3-d$.  In this case we are interested in quantizations which are holomorphically translation-invariant in the $\C^d$ directions, and also satisfying the scaling axiom with respect to scaling of $\C^d$. A small variant of the arguments discussed above shows that, in the obstruction-deformation complex, there is nothing of negative weight, so again the result follows.  

\end{proof}

Now we will prove theorem \ref{theorem_(1,0)_annulus_anomaly}, which we restate now for convenience.  
\begin{theorem*}
For any Calabi-Yau manifold $X$ of dimension $3$, there is a unique quantization of open-closed BCOV theory to the annulus level. 
\end{theorem*}
\begin{proof}
The strategy is the following.  We will start by examining the cohomology groups containing potential obstructions and deformations at the annulus level. We will find that for any Calabi-Yau, there are no possible deformations but there might be obstructions. Next, we will see that to show the obstruction vanishes it suffices to show it locally on $X$. This reduces the problem to the case of $\C^3$, where we can apply a  variant of the annulus anomaly cancellation argument for the full BCOV theory.

The cohomology calculation we start with is almost identical to that of proposition \ref{proposition_annulus_cohomology} where we analyzed the corresponding groups for $\gl(N \mid N)$ holomorphic Chern-Simons.  The cohomology group we are computing is described in terms of the $\mf{sl}$-version of the cyclic homology of jets of functions on $X$.  Recall that the $\mf{sl}$ cyclic homology is
$$
 HC_{\mf{sl}}( J \Omega^{0,\ast}(X))[1] = \oplus_{i \ge 1} J \Omega^{i,\ast}(X)[i+1] \oplus t^{-1} J \Omega^{-\ast,\ast}(X)[t^{-1}][1]
$$
with a differential which is the jet of $\dbar + t \partial$.  We will denote this $D$-module (in the model used on the right hand side) by $M$. There is a filtration on $M$ by the powers of $t$ occurring, whose associated graded is the same object with differential $\dbar$. Let us denote the associated graded by $M'$. 
 
The obstruction-deformation complex is global sections of the sheaf of complexes on $X$
$$
\Omega^{3,3} \otimes_{D} \Sym^2 M^\vee.  
$$
We will first show that global sections of this has no cohomology in degrees $\le 0$.   

There is a spectral sequence converging to this cohomology whose first page is the cohomology of 
$$
\Omega^{3,3} \otimes_{D} \Sym^2 (M')^\vee. 
$$
Now,
$$
(M')^\vee = \oplus_{i \ge 1} (J \Omega^{i,\ast})^\vee [-i-1] \oplus t (J \Omega^{0,\ast})^\vee[[t]][-1]
$$
where $t$ has degree $2$.  Thus, $(M')^\vee$ is a direct sum of $D$-modules of the form $(J \Omega^{k,\ast})^\vee[-r]$ where the shift is by $r \ge 2$.  The cohomology of $(J \Omega^{k,\ast})^\vee$ is the bundle whose holomorphic differential operators from holomorphic $k$-forms $\Omega^{k}_{hol}$ to $\Oo$. Call this bundle $\op{Diff}_{hol}(\Omega^k,\Oo)$. We thus find
$$
(M')^\vee \simeq  \oplus_{i \ge 1} (\op{Diff}_{hol}(\Omega^i, \Oo) [-i-1] \oplus t \op{Diff}_{hol}(\Oo,\Oo)[[t]][-1].
$$   

A holomorphic $D$-module is a (possibly infinite rank) holomorphic bundle with a flat holomorphic connection, and so is in particular a $\cinfty$ $D$-module.  The $\cinfty$ $D$-module $\op{Diff}_{hol}(\Omega^i,\Oo)$ arises, of course, from a holomorphic $D$-module of the same name.  Further, for any holomorphic $D$-module $V$, if $V^{\cinfty}$ refers to the corresponding $\cinfty$ $D$-module, then there is a quasi-isomorphism of sheaves
$$
\Omega^3_{hol} \otimes_{D_{hol}}^{\mbb L} V[3] \iso \Omega^{3,3} \otimes_{D}^{\mbb L} V^{\cinfty}.
$$
Applying this to our situation, we find that 
$$
\Omega^3_{hol} \otimes^{\mbb L}_{D_{hol}} \Sym^2 (\mc{H}^\ast ((M')^\vee)[3] \simeq \Omega^{3,3} \otimes^{\mbb L}_{D} \Sym^2 (M')^\vee
$$
where $\mc{H}^\ast ((M')^\vee)$ refers to the cohomology sheaves, which are holomorphic $D$-modules. 

Now, the $D_{hol}$-modules $\op{Diff}_{hol}(\Omega^k, \Oo)\otimes \op{Diff}_{hol}(\Omega^l,\Oo)$ are flat. It follows that $\Sym^2 \mc{H}^\ast ( (M')^\vee)$ is a flat $D_{hol}$-module, and that there is a quasi-isomorphism
 $$
\Omega^3_{hol} \otimes_{D_{hol}} \Sym^2 (\mc{H}^\ast ((M')^\vee)[3] \simeq \Omega^{3,3} \otimes^{\mbb L}_{D} \Sym^2 (M')^\vee.
$$
Since the cohomology sheaves of $\Sym^2 (M')^\vee$ are in degrees $4$ and higher, we find that the sheaf on the left hand side of this equation has cohomology sheaves in degrees $1$ and higher. This proves that there is no ambiguity in quantizing open-closed $(1,0)$-BCOV theory on $X$ to the annulus level (although there might be an anomaly).  

The only possible contribution to the group controlling anomalies comes from the piece of $\Sym^2 (\mc{H}^\ast((M')^\vee)$ in degree $4$, which is $\Sym^2 \op{Diff}_{hol} ( \Omega^1,\Oo)$.  We find that 
$$ 
 H^1 (X, \Omega^{3,3} \otimes_{D} \Sym^2 (M')^\vee ) = H^0(X, \Omega^3_{hol} \otimes_{D_{hol}} \Sym^2 \op{Diff}_{hol}(\Omega^1,\Oo) ).  
$$
Since we have a spectral sequence converging to our anomaly group from the group in the displayed equation, we have
$$
H^1 (X, \Omega^{3,3} \otimes_{D} \Sym^2 M^\vee) \subset H^0(X, \Omega^3_{hol} \otimes_{D_{hol}} \Sym^2 \op{Diff}_{hol}(\Omega^1,\Oo) ).  
$$ 
(On the left hand side of this is the group containing possible anomalies).

Next, we will show that for any $U \subset X$, the map
$$
H^1 (X, \Omega^{3,3} \otimes_{D} \Sym^2 M^\vee) \to H^1 (U, \Omega^{3,3} \otimes_{D} \Sym^2 M^\vee)  
$$
given by restricting an anomaly to $U$ is injective. To show this, it suffices to show that the map
$$
 H^0(X, \Omega^3_{hol} \otimes_{D_{hol}} \Sym^2 \op{Diff}_{hol}(\Omega^1,\Oo) ) \to  H^0(U, \Omega^3_{hol} \otimes_{D_{hol}} \Sym^2 \op{Diff}_{hol}(\Omega^1,\Oo) )   
$$ 
is injective. This is immediate, since the sheaf in question is built from complex analytic data.  

One of the results of \cite{Cos11} is that anomalies are local, that is, the restriction of the anomaly cohomology class to some open subset $U \subset X$ is determined by the restriction of the theory on $X$ to $U$. To show that the annulus anomaly vanishes, we now need to do so on any open subset $U \subset X$, which we can take to be a polydisc.  

Further, since the annulus anomaly on a polydisc is the restriction of the annulus anomaly on $\C^3$, it suffices to show that the anomaly on $\C^3$ vanishes. For this, we will use a variant of our argument for the full BCOV theory.  

For the full BCOV theory, we found that the open-string anomaly was of the form
$$
\int \op{Tr} (A \partial A) \op{Tr} ( (\partial A)^2) + \int \op{Tr} A \op{Tr} (\partial A)^3. 
$$
There, we were using $\gl(N \mid N)$ as our gauge group. Further, we saw that this anomaly cancelled with a contribution from the closed-string sector. 

Clearly the second term in this expression vanishes if we use $\mf{sl}(N \mid N)$.  Since our $(1,0)$ BCOV theory has fewer fields, it suffices to show that only the first term in the open-string anomaly cancels with the contribution from $(1,0)$ BCOV theory.

Recall that the kernel for the BV odd Poisson bracket for the full BCOV theory is
$$
\pi = (\del \otimes 1) \delta_{Diag} \in \br{\PV}^{\ast,\ast}(\C^3) \what{\otimes} \br{\PV}^{\ast,\ast}(\C^3).
$$
Here $\br{\PV}$ indicates polyvector fields with distributional coefficients, and $\what{\otimes}$ is the completed topological tensor product. We can view the right hand side of this as being polyvector fields with distributional coefficients on $\C^3 \times \C^3$. Also, $\delta_{Diag}$ is the delta-current on the diagonal. Normally this would be a distributional form on $\C^3 \times \C^3$, but we  view it as a polyvector field using the isomorphism between forms and polyvector fields.

We let $\pi^{i,j}$ denote the component of this Kernel which lives in $\br{\PV}^{i,\ast} \what{\otimes} \br{\PV}^{j,\ast}$.  The kernel for the odd Poisson bracket for $(1,0)$ BCOV theory is $\pi^{1,1}$. 

The interaction $I^{1-disk}$ for the full BCOV theory is linear in the closed string fields of full BCOV theory, and so can be written as a sum of terms 
$$I^{1-disk} = \sum_{i,k\ge 0} I^{1-disk}_{i,k}$$
where  $I^{1-disk}_{i,k}$ only depends on the fields in $t^k \PV^{i,\ast}$.  The corresponding interaction for $(1,0)$ BCOV theory is 
$$I^{1-disk, \ (1,0)} = I^{1-disk}_{1,0} + I^{1-disk}_{0,1}.$$ 

Since the BV Poisson bracket only lives in the subspace of fields which have no powers of $t$, we find that the closed string BV Poisson bracket of $I^{1-disk}$ can be written as a sum
\begin{align*} 
 \{I^{1-disk}, I^{1-disk}\}^{C} &= \sum_{i,j, = 0 \dots 3} \{ I^{1-disk}_{i,0}, I^{1-disk}_{j,0} \}^{C}\\ 
&= \sum_{i,j = 0\dots 3} \ip{ \pi^{i,j}, I^{1-disk}_{i,0} \boxtimes I^{1-disk}_{j,0} }
\end{align*}
where $\ip{-,-}$ indicates contraction between polyvector fields on $\C^3 \times \C^3$ and elements of the dual space.

Further,
$$
 \{I^{1-disk,\ (1,0) }, I^{1-disk, \ (1,0) }\}^{C} =  \{ I^{1-disk}_{1,0}, I^{1-disk}_{1,0} \}^{C}. 
$$
Now, $I^{1-disk}_{i,0}$ depends on $i+1$ open-string fields $A \in \Omega^{0,\ast}(\C^3)\otimes \gl(N \mid N)$ and has a single trace.  We showed in \ref{lemma_closed_string_annulus_anomaly} that
$$
\{I^{1-disk}, I^{1-disk}\}^{C} (A) =  \int \op{Tr} (A \partial A) \op{Tr} (\partial A)^2 + \op{\int} \op{Tr} (A) \op{Tr} ( (\partial A)^3)   
$$
(up to some non-zero constant that can be scaled away by redefinition of the fields).
 
It follows from this and from the dependence of $I^{1-disk}_{i,0}$ on the open string fields that
$$
\{I^{1-disk}_{1,0}, I^{1-disk}_{1,0}\}^{C}(A) =   \int \op{Tr} (A \partial A) \op{Tr} (\partial A)^2.  
$$
Since this is the closed-string BV bracket of $I^{1-disk, \ (1,0)}$ with itself, we find that the closed-string annulus anomaly for $(1,0)$ BCOV theory precisely cancels that form the open string sector,  as desired.  
\end{proof} 

One question we have not yet addressed is the following. 
\begin{theorem*}
On any Calabi-Yau manifold $X$, the dynamically-generated classical interaction for $(1,0)$ BCOV theory is equivalent to the restriction of the classical interaction of \cite{CosLi11} for the full BCOV theory to the fields of $(1,0)$ BCOV theory. Explicitly, this interaction is the following. Let us denote the fields of $(1,0)$ BCOV theory by $\alpha \in \PV^{1,\ast}(X)[1]$ and $\phi \in t PV^{0,\ast}(X)$. Then the interaction is
$$
I(\alpha,\phi) = \sum_{n \ge 3}   \int \alpha^3 \phi^{n-3}  
$$
where we are using the integration map on polyvector fields specified earlier. 

\end{theorem*}
\begin{proof}
Recall that $(1,0)$ BCOV theory can be quantized at genus $0$ on any Calabi-Yau manifold, so that the statement makes sense.

The argument goes as follows. First, we will show that any possible classical interactions on $X$ are determined by their behaviour on an arbitrary open subset of $X$. By taking this subset to be a polydisc, and then embedding this polydisc in $\C^3$, this reduces us to the case $X = \C^3$. In that case, $SU(3)$ and scale invariance, as well as the classical master equation, fix the form of the classical interaction uniquely up to a scale. This scale can be fixed using an argument similar to that given in \ref{classical-interaction-lemma}.

Let us analyze the obstruction-deformation complex controlling possibly classical interactions on $X$.  As usual, this is built from jets of fields, this time of only the closed-string sector.  These fields can be identified with the double complex 
$$
\E = \Omega^{2,\ast}[1] \to \Omega^{3,\ast}.
$$ 
The obstruction-deformation complex is
$$
\Omega^{3,3} \otimes_{D} \Oo( J (\E)) 
$$
as usual. 
 
Now, the complex $\E$ of fields has cohomology the sheaf of closed holomorphic $2$-forms, with a shift of $[1]$. There is a quasi-isomorphism between the sheaf of closed holomorphic $2$-forms and the complex of sheaves 
$$
\C[2] \to \Oo[1] \to \Omega^1_{hol}
$$
where $\C$ is the constant sheaf and $\Omega^1_{hol}$ is the sheaf of holomorphic $1$-forms.  Replacing each term in this complex by a resolution, we find that $\E$ is quasi-isomorphic to the total complex of the double complex of the form
$$
\Omega^{\ast,\ast}[3] \to \Omega^{0,\ast}[2] \to \Omega^{1,\ast}[1]
$$
where the first map is projection onto $(0,\ast)$-forms and the second map is the operator $\partial$. 

It follows that $\E$ (in this model) has a filtration whose associated graded is the direct sum 
$$
\op{Gr} \E = \Omega^{\ast,\ast}[3] \oplus \Omega^{0,\ast}[2] \oplus \Omega^{1,\ast}[1].
$$
This filtration leads to a convergent spectral sequence for our obstruction-deformation complex, where the first page of this spectral sequence is built from jets of $\op{Gr} \E$. 

The $D$-module $J \op{Gr} \E$ of jets of $\E$ is quasi-isomorphic to the direct sum
$$
J \op{Gr} \E \simeq \cinfty_X[3] \oplus J \Oo_{hol}[2] \oplus J \Omega^1_{hol}[1].
$$
Here $\cinfty_X$ is the trivial $D$-module, and $J \Oo_{hol}$, $J \Omega^1_{hol}$ refer to jets of holomorphic functions and one-forms.  It follows that the $D$-module $J \op{Gr} \E$ arises from a $D_{hol}$-module. An argument on $D_{hol}$-module we discussed in the proof of Theorem \ref{theorem_(1,0)_cohomology_cancellation} tells us that for any holomorphic $D$-module $M$ with associated $\cinfty$ $D$-module $M^{\cinfty}$, we have a quasi-isomorphism of sheaves
$$
\Omega^{3,3} \otimes^{\mbb L}_{D} M^{\cinfty} \simeq \Omega^{3,0}_{hol} \otimes^{\mbb L}_{D_{hol}} M[3]. 
$$
Applying this to our situation yields a quasi-isomorphism
$$
\Omega^{3,3} \otimes_D^{\mbb L} \Sym^k (J \op{Gr} \E)^\vee \simeq 
\Omega^{3,0}_{hol} \otimes^{\mbb L}_{D_{hol}} \Sym^k \left( \Oo_{hol}[-3] \oplus  J \Oo_{hol}^{\vee}[-2] \oplus (J \Omega^1_{hol})^{\vee}[-1] \right)[3].
$$
Now, the dual of $J \Oo_{hol}$ is $D_{hol}$, and the dual of $J \Omega^1_{hol}$ is the space of holomorphic differential operators $\op{Diff}_{hol} (\Omega^1, \Oo)$ from $\Omega^1_{hol}$ to $\Oo_{hol}$.   It follows  that the complex $\Omega^{3,3} \otimes_D^{\mbb L} \Sym^k (J \op{Gr} \E)^\vee$ is a direct summand of
$$
\Omega^{3,0}_{hol} \otimes^{\mbb L}_{D_{hol}} \left\{ \oplus_{a + b + c = k}\Oo^{\otimes a} \otimes D_{hol}^{\otimes b} \otimes \op{Diff}_{hol}(\Omega^1,\Oo)^{\otimes c} [-3a -2b -c +3]\right\}.  
$$
Also, $a \le 1$ for symmetry reasons. We are only interested in the case when $k \ge 3$ (corresponding to cubic and higher interactions).  The $D$-module inside the tensor product is flat over $D_{hol}$ when $b + c \ge 1$, which is automatic since $k \ge 3$.  Therefore, we can replace the derived tensor product  by the underived tensor product. 

Each summand in the above expression is a $D$-module concentrated in a single degree, $3a+2b+c-3$.   The only summands that can contribute to cohomology in degree $0$ occur when $3a+2b+c -3 \le 0$, which can only happen (given $k \ge 3$) when $a=b=0$, $c = 3$.  

It follows from this discussion that
\begin{align*} 
 H^0\left(X, \Omega^{3,3} \otimes_{D} \Sym^k (( J \E)^\vee \right) & = 0 \text{ if } k > 3 \\
 H^0\left(X, \Omega^{3,3} \otimes_{D} \Sym^3 (( J \E)^\vee \right) & \subset H^0\left(X, \Omega^{3,0}_{hol} \otimes_{D_{hol}} \wedge^3  ( \op{Diff}_{hol}(\Omega^1,\Oo ) ) \right). 
\end{align*}
From this it follows that, for every $U \subset X$, the map
$$
 H^0\left(X, \Omega^{3,3} \otimes_{D} \Sym^k (( J \E)^\vee \right) \to  H^0\left(U, \Omega^{3,3} \otimes_{D} \Sym^k (( J \E)^\vee \right)  
$$
is injective when $k \ge 3$.  

Therefore, to identify our classical interaction on $X$, we just need to do it locally on $X$, and so on $\C^3$.   

On $\C^3$, we are, as usual, interested in $SU(3)$-invariant and holomorphically translation invariant interactions.  In this case, the obstruction-deformation complex is the $SU(3)$ invariants in
$$
\prod_{k \ge 3} \Omega^{3,0}_0 \otimes_{D_0} \Sym^k J_0(\E)^\vee[3]
$$  
where $\Omega^{3,0}_0$ is the fibre at zero of the bundle of $(3,0)$ forms, $D_0$ indicates constant-coefficient holomorphic differential operators, and $J_0(\E)$ is the fibre at zero of the jet bundle.  

A variant of the above calculation tells us that $H^0$ of the obstruction-deformation complex for constructing the interaction $I_k$ is zero if $k > 3$, so that we are only concerned with the cubic interaction. 

The argument we used in the proof of \ref{classical-interaction-lemma} shows that there is up to a scale precisely one cubic interaction which is $SU(3)$-invariant and of the correct scaling weight.  Because the cohomology groups where the interactions $I_k$ for $k > 3$ live are zero, there is (up to a scale) a single $SU(3)$ invariant solution of the classical master equation with the correct scaling dimension.   It follows that, up to scaling, the dynamically-generated classical interaction for $(1,0)$ BCOV theory on $\C^3$ is equivalent to this unique solution, and so, up to a constant, it is equivalent to the interaction $I_n$ specified in the statement of the theorem.  

The constant is fixed, as in the proof of \ref{classical-interaction-lemma},  by the fact that at the level of cohomology we know that the Lie bracket on the closed-string fields coming from the classical interaction must be the standard one on divergence-free holomorphic  vector fields. 
\end{proof}


\def\cprime{$'$}

\end{document}